\documentclass[10pt,a4paper,onecolumn]{IEEEtran}
\usepackage{blindtext, graphicx}
\usepackage{tabularx}
\usepackage{pgf,tikz,pgfplots}
\usepackage[ruled,vlined]{algorithm2e}
\usepackage{hyperref}
\pgfplotsset{compat=1.14}
\usepackage{mathrsfs}
\usetikzlibrary{arrows}
\usetikzlibrary[patterns]
\hypersetup{
 colorlinks,
 linkcolor={blue!100!black},
 citecolor={blue!100!black},
 urlcolor={blue!80!black}
}

\usepackage{amsmath,amssymb,amsfonts,amsthm,graphicx,bm,bbm,dsfont}
\newtheorem{theorem}{Theorem}
\newtheorem{remark}{Remark}

\newtheorem{proposition}{Proposition}
\newtheorem{example}{Example}
\newtheorem{lemma}{Lemma}
\newtheorem{definition}{Definition}

\DeclareSymbolFont{bbold}{U}{bbold}{m}{n}
\DeclareSymbolFontAlphabet{\mathbbold}{bbold}
\newcommand{\1}{\mathbbold{1}}

\usepackage{color}

\newcommand{\bF}{\mathbb{F}}

\newcommand{\cB}{\mathcal{B}}

\newcommand{\cD}{\mathcal{D}}
\newcommand{\cE}{\mathcal{E}}

\newcommand{\cI}{\mathcal{I}}

\newcommand{\cS}{\mathcal{S}}

\newcommand{\cU}{\mathcal{U}}

\newcommand{\bolda}{\mathbf{a}}
\newcommand{\boldb}{\mathbf{b}}
\newcommand{\boldc}{\mathbf{c}}
\newcommand{\boldd}{\mathbf{d}}
\newcommand{\bolde}{\mathbf{e}}

\newcommand{\boldz}{\mathbf{z}}

\makeatletter
\def\namedlabel#1#2{\begingroup
	\def\@currentlabel{#2}%
	\label{#1}\endgroup
}
\makeatother
\begin{document}
\title{Correcting~$k$ Deletions and Insertions\\ in Racetrack Memory}

\author{\textbf{Jin Sima} \IEEEauthorblockN{ and \textbf{Jehoshua Bruck}}\\
	\IEEEauthorblockA{
	Department of Electrical Engineering,
California Institute of Technology, Pasadena 91125, CA, USA\\}}

\maketitle

\begin{abstract}
Racetrack memory is a tape-like structure where data is stored sequentially as a track of single-bit memory cells. The cells are accessed through read/write ports, called heads. When reading/writing the data, the heads stay fixed and the track is shifting. One of the main challenges in developing racetrack memory systems is the limited precision in controlling the track shifts, that in turn affects the reliability of reading and writing the data. A current proposal for combating deletions in racetrack memories is to use redundant heads per-track resulting in multiple copies (potentially erroneous) and recovering the data by solving a specialized version of a sequence reconstruction problem. Using this approach, $k$-deletion correcting codes of length $n$, with $d \geq 2$ heads per-track, with redundancy $\log \log n + 4$ were constructed. However, the known approach requires that $k \leq d$, namely, that the number of heads ($d$) is larger than or equal to the number of correctable deletions ($k$). Here we address the question: What is the best redundancy that can be achieved for a $k$-deletion code ($k$ is a constant) if the number of heads is fixed at $d$ (due to implementation constraints)? One of our key results is an answer to this question, namely, we construct codes that can correct $k$ deletions, for any $k$ beyond the known limit of $d$. The code has $4k \log \log n+o(\log \log n)$ redundancy for $k \leq 2d-1$. In addition, when $k \geq 2d$, our codes have $2 \lfloor k/d\rfloor \log n+o(\log n)$ redundancy, that we prove it is order-wise optimal,  specifically, we prove that the redundancy required for correcting $k$ deletions is at least $\lfloor k/d\rfloor \log n+o(\log n)$. The encoding/decoding complexity of our codes is $O(n\log^{2k}n)$.  Finally, we ask a general question: What is the optimal redundancy for codes correcting a combination of at most $k$ deletions and insertions in a $d$-head racetrack memory? We prove that the redundancy sufficient to correct a combination of $k$ deletion and insertion errors is similar to the case of $k$ deletion errors. 
\end{abstract}

\section{Introduction}
Racetrack memory is a promising non-volatile memory that possesses the advantages of ultra-high storage density and low latency (comparable to SRAM latency) \cite{parkin2008magnetic,sun2013cross}. It has a tape-like structure where the data is stored sequentially as a track of single-bit memory cells. The cells are accessed through read/write ports, called heads. When reading/writing the data, the heads stay fixed and the track is shifting.

One of the main challenges in developing racetrack memory systems is the limited precision in controlling the track shifts, that in turn affects the reliability of reading and writing the data~\cite{hayashi2008current,zhang2015hi}. Specifically, the track may either not shift or shift more steps than expected. When the track does not shift, the same cell is read twice, causing a sticky insertion. When the track shifts more than a single step, cells are skipped, causing deletions in the reads~\cite{chee2018coding}.  

It is natural to use deletion and sticky insertion correcting codes to deal with shift errors. Also, it is known that a code correcting~$k$ deletions is capable of correcting~$s$ deletions and~$r$ insertions when~$s+r\le k$ \cite{levenshtein1966binary}. However, designing redundancy and complexity efficient deletion correcting codes has been an open problem for decades, though there is a significant advance toward the solution recently. In fact, no deletion correcting codes with rate approaching~$1$ were known until~\cite{brakensiek2016efficient} proposed a code with redundancy~$128k^2\log k\log n +o(\log n)$.  
Evidently, for $k$, a constant number of deletions, the redundancy of this code is orders of magnitude away from optimal, known to be in the range~$k\log n +o(\log n)$ to~$2k\log n +o(\log n)$~\cite{levenshtein1966binary}. After \cite{brakensiek2016efficient}, the work of  \cite{CJLW18} and \cite{sima2019deletion} independently proposed $k$-deletion codes with $O(k\log n)$ bits of redundancy, which are order-wise optimal. Following \cite{sima2019deletion},  \cite{sima2020deletion} proposed a systematic deletion code with $4k\log n +o(\log n )$ bits of redundancy and is computationally efficient for constant $k$. The redundancy was later improved in \cite{song2021deletion} to $(4k-1)\log n +o(\log n)$. Despite the recent progress in deletion and insertion correcting codes, it is still tempting to explore constructions of deletion and insertion correcting codes that are specialized for racetrack memories and might provide more efficient redundancy and lower complexity encoding/decoding algorithms.  

There are two approaches for construction of codes for racetrack memories. The first is to leverage the fact that there are multiple parallel tracks with a single head per-track, and the second, is to add redundant heads per-track. For the multiple parallel head structure, the proposed codes in~\cite{vahid2017correcting} can correct up to two deletions per head and the proposed codes in~\cite{chee2018codes} can correct~$l$ bursts of deletions, each of length at most~$b$. The codes in~\cite{chee2018codes} are asymptotically (in the number of heads) rate-optimal.
The second approach for combating deletions in racetrack memories is to use redundant heads per-track~\cite{zhang2015hi,chee2018coding,cheereconstruction}. As shown in Fig.~\ref{figure:racetrackmemory}, a track is read by multiple heads, resulting in multiple copies (potentially erroneous) of the same sequence.  This can be regarded as a sequence reconstruction problem, where a sequence~$\boldc$ needs to be recovered from multiple copies, each obtained after~$k$ deletions in~$\boldc$. We emphasize that the general sequence reconstruction problem~\cite{levenshtein1997recon} is different from the current settings, as here the heads are at fixed and known positions, hence, the set of deletions locations in one head is a shift of that in another head~\cite{chee2018coding}. This is because the heads stay fixed and thus the deletion locations in their reads have fixed relative distances.
Demonstrating the advantage of multiple heads, the paper~\cite{cheereconstruction} proposed an efficient $k$-deletion code of length~$n$ with redundancy~$\log\log n+4$ and a $(k-1)$-deletion code with $O(1)$ redundancy, both using $k$ heads. In contrast, for general $k$-deletion codes the lower bound on the redundancy is $k\log n$. However, the code in~\cite{cheereconstruction} is required to use $d$~heads and is limiting $k$ to be smaller or equal to $d$\footnote{Throughout the paper, it is assumed that $d\ge 2$.}. It is known that the number of heads affects the area overhead of the racetrack memory device~\cite{chee2018coding}, hence, it motivates the following \textbf{natural question}: What is the best redundancy that can be achieved for a $k$-deletion code ($k$ is a constant) if the number of heads is fixed at $d$ (due to area limitations)? 

One of our \textbf{key results} is an answer to this question, namely, we construct codes that can correct $k$ deletions, for any $k$ beyond the known limit of $d$. Our code has~$O(4k\log\log n)$ redundancy for the case when~$d\le k\le 2d-1$. In addition, when~$k\ge 2d$, the code has~$2\lfloor k/d\rfloor \log n+o(\log n)$ redundancy.
Our key result is summarized formally by the following theorem. Notice that the theorem implies that the redundancy of our codes is asymptotically larger than optimal by a factor of at most four.
\begin{theorem}\label{theorem:main}
For a constant integer~$k$, let the distance $t_i$ between the $i$-th and $(i+1)$-th heads be~$t_i\ge\max\{(3k+\lceil \log n \rceil+2)[k(k-1)/2+1]+(7k-k^3)/6,(4k+1)(5k+\lceil \log n \rceil+3)\}$ for $i\in\{1,\ldots,d-1\}$. Then for~$d\le k\le 2d-1$, there exists a length~$N=n+4k \log\log n+o(\log\log n)$~$d$-head~$k$-deletion correcting code with redundancy~$4k \log\log n+o(\log\log n)$.
For~$k\ge 2d$, there exists a length~$N=n+2\lfloor k/d\rfloor \log n+o(\log n)$~$d$-head~$k$-deletion correcting code with redundancy~$2\lfloor k/d\rfloor \log n+o(\log n)$.
The encoding and decoding functions can be computed in~$O(n\log^{2k}n)$ time. 
Moreover, for~$k\ge 2d$ and~$t_i=n^{o(1)}$, the amount of redundancy of~a~$d$-head~$k$-deletion correcting code is lower bounded by~$\lfloor k/2d\rfloor \log n+o(\log n)$.
\end{theorem}
Since in addition to deletion errors, sticky insertion errors and substitution errors occur in racetrack memory, we are
interested in codes that correct not only deletions, but a combination of deletion, sticky insertion, and substitution errors in a multiple head racetrack memory. However, in contrast to single head cases where a deletion code is also a deletion/insertion code, there is no such equivalence in multiple head racetrack memories. Correcting a combination of at most $k$ deletions and sticky insertions in total turns out to be more difficult than correcting $k$ deletion errors. It is not known whether the $k$-deletion code with $\log\log n +O(1)$ redundancy and the $(k-1)$-deletion code with $O(1)$ redundancy in \cite{chee2018coding} apply to a combination of deletion and sticky insertion errors in a $k$-head racetrack memory. 

Our second result, which is the main result in this paper, provides an answer for such scenarios. We consider a more general problem of correcting a combination of deletions and insertions in a $d$-head racetrack memory, rather than deletions and sticky insertions, and show that the redundancy result for deletion cases extends to cases with a combination of deletions and insertions. Note that this covers the cases with deletion, insertion, and substitution errors, since a substitution is a deletion followed by an insertion.
\begin{theorem}\label{theorem:mainedit}
For a constant integer~$k$, let the distance $t_i$ between the $i$-th and $(i+1)$-th heads be equal and~$t_i=t>(\frac{k^2}{4}+ 3k)(6k+\lceil \log n \rceil+3)+8k+\lceil \log n \rceil+3$ for $i\in\{1,\ldots,d-1\}$. Then for~$k<d$, there exists a length~$N=n+k+1+O(1)$~code correcting a combination of at most $k$ insertions and deletions in a $d$-head racetrack memory with redundancy~$k+1+O(1)$. The encoding and decoding complexity is $poly(n)$.
For~$d\le k\le 2d-1$, there exists a length $N=n+4k\log\log n+o(\log\log n)$ code correcting a combination of at most $k$ insertions and deletions in a $d$-head racetrack memory with redundancy $4k\log\log n+o(\log\log n)$.
 Finally, when $d\ge 2d$,
there exists a length~$N=n+2\lfloor k/d\rfloor \log n+o(\log n)$~code that corrects a combination of at most $k$ insertions and deletions in a $d$-head racetrack memory with redundancy~$2\lfloor k/d\rfloor \log n+o(\log n)$.
The encoding and decoding functions can be computed in~$O(n\log^{2k}n)$ time. 
\end{theorem}
\begin{remark}
Theorem \ref{theorem:mainedit} improves the head distance in Theorem \ref{theorem:main} when $k\ge 15$ and $n$ is sufficiently large.
\end{remark}
\textbf{Organization:} In Section~\ref{section:preliminaries}, we  present the problem settings and some basic lemmas needed in our proof. Section~\ref{section:lessthan2mdeletions} presents the proof of the main result for the case~$k\le 2d-1$.  Section~\ref{section:synchronization} describes in detail how to synchronize the reads. The case~$k\ge 2d$ is addressed in Section 
~\ref{section:greaterthan2d}. Section \ref{section:correctediterrors} shows how to correct deletion and insertion errors and proves Theorem \ref{theorem:mainedit}. Section~\ref{section:conclusion} concludes the paper.
\begin{figure}
\centering
\definecolor{ffqqqq}{rgb}{1,0,0}
\definecolor{rvwvcq}{rgb}{0.08235294117647059,0.396078431372549,0.7529411764705882}
\begin{tikzpicture}[line cap=round,line join=round,>=triangle 45,x=1cm,y=1cm]
xtick={-7,-6,...,7},
ytick={-7,-6,...,3},]
\clip(-3,3.6) rectangle (7,6);
\fill[line width=0pt,color=rvwvcq,fill=rvwvcq,fill opacity=0.5] (2,5) -- (2.3,5) -- (2.3,5.5) -- (2,5.5) -- cycle;
\fill[line width=0pt,color=rvwvcq,fill=rvwvcq,fill opacity=0.5] (1.8,5) -- (2.5,5) -- (2.15,4.75) -- cycle;

\fill[line width=0pt,color=rvwvcq,fill=rvwvcq,fill opacity=0.5] (0,5) -- (0.3,5) -- (0.3,5.5) -- (0,5.5) -- cycle;
\fill[line width=0pt,color=rvwvcq,fill=rvwvcq,fill opacity=0.5] (-0.2,5) -- (0.5,5) -- (0.15,4.75) -- cycle;

\fill[line width=0pt,color=rvwvcq,fill=rvwvcq,fill opacity=0.5] (-1,5) -- (-0.7,5) -- (-0.7,5.5) -- (-1,5.5) -- cycle;
\fill[line width=0pt,color=rvwvcq,fill=rvwvcq,fill opacity=0.5] (-1.2,5) -- (-0.5,5) -- (-0.85,4.75) -- cycle;

\draw [line width=1pt] (-0.35,4.5)-- (-0.35,3.8);
\draw [line width=1pt] (0.65,4.5)-- (0.65,3.8);
\draw [line width=1pt] (-1.35,4.5)-- (-1.35,3.8);
\draw [line width=1pt] (1.65,4.5)-- (1.65,3.8);
\draw [line width=1pt] (2.65,4.5)-- (2.65,3.8);
\draw [line width=1pt] (4.15,4.5)-- (4.15,3.8);
\draw [line width=2pt] (5.15,4.5)-- (5.15,3.8);
\draw [line width=2pt] (5.15,4.5)-- (-3,4.5);
\draw [line width=2pt] (5.15,3.8)-- (-3,3.8);

\draw (4.65,4.15) node[anchor=center] {$c_1$};
\draw (3.4,4.15) node[anchor=center] {$\ldots$};
\draw (2.15,4.15) node[anchor=center] {$c_6$};
\draw (1.15,4.15) node[anchor=center] {$c_7$};
\draw (0.15,4.15) node[anchor=center] {$c_8$};
\draw (-0.85,4.15) node[anchor=center] {$c_9$};
\draw (-2.1,4.15) node[anchor=center] {$\ldots$};
\fill[line width=0pt,color=rvwvcq,fill=black,fill opacity=0.8] (5.75,4.1) -- (6.25,4.1) -- (6.25,4.2) -- (5.75,4.2) -- cycle;
\fill[line width=0pt,color=rvwvcq,fill=black,fill opacity=0.8] (6.25,4.03) -- (6.25,4.27) -- (6.5,4.15) -- cycle;
\end{tikzpicture}\caption{Racetrack memory with multiple heads.
}
\label{figure:racetrackmemory}
\end{figure}
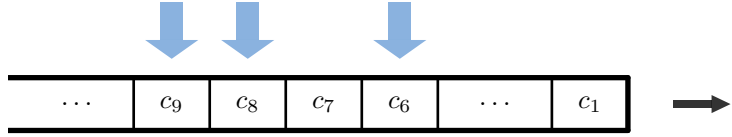
\section{Preliminaries}\label{section:preliminaries}
\subsection{Problem Settings}
We now describe the problem settings and the notations needed.
 For any two integers~$i\le j$, let~$[i,j]=\{i,i+1,\ldots,j-1,j\}$ be an integer interval that contains all integers between~$i$ and~$j$. Let~$[i,j]=\emptyset$ for~$i>j$.
For a length~$n$ sequence~$\boldc=(c_1,\ldots,c_n)$, an index set~$\mathcal{I}\subseteq [1,n]$, let
\begin{align*}
    \boldc_{\mathcal{I}}=(c_i:i\in \mathcal{I})
\end{align*}
be a subsequence of~$\boldc$, obtained by choosing bits with locations in the location set~$\mathcal{I}$. Denote by~$\mathcal{I}^c=[1,n]\backslash \mathcal{I}$ the complement of~$\mathcal{I}$. 

In the channel model of a~$d$-head racetrack memory, the input is a binary sequence $\boldc\in\{0,1\}^n$. The channel output consists of $d$ subsequences of $\boldc$ of length $n-k$, obtained by the $d$ heads after $k$ deletions in the channel input $\boldc$, respectively. Each subsequence is called a \emph{read}. 
Let $\boldsymbol{\delta}_i=\{\delta_{i,1},\ldots,\delta_{i,k}\}\subseteq [1,n]$ be the deletion locations in the $i$-th head such that $\delta_{i,1}<\ldots<\delta_{i,k}$. Then, the read from the $i$-th head is given by $\boldc_{\boldsymbol{\delta}^c_i}$,~$i\in[1,d]$, i.e., bits $c_{\ell}$, $\ell\in \boldsymbol{\delta}_i$ are deleted. 

Note that in a~$d$-head racetrack memory, the heads are placed in fixed positions, and the deletions are caused by "over-shifts" of the track. Hence when a deletion occurs at the $j$-th bit in the read of the $i$-th head, a deletion also occurs at the $(j+t_i)$-th bit in the read of the $(i+1)$-th head, where $t_i$ is the distance between the $i$-th head and the $(i+1)$-th head, $i\in[d-1]$. Then,
the deletion location sets 
$\{\boldsymbol{\delta}_i\}^d_{i=1}$ satisfy
\begin{align*}
    \boldsymbol{\delta}_{i+1}=\boldsymbol{\delta}_i + t_i,
\end{align*}
for positive integers~$t_i$,~$i\in [1,d-1]$, where for an integer set~$\mathcal{S}$ and an integer~$t$,~$\mathcal{S}+t=\{x+t:x\in \mathcal{S}\}$.

To formally define a code for the $d$-head racetrack memory, 
we represent the $d$ reads from the $d$ heads by a $d\times (n-k)$ binary matrix, called the \emph{read matrix}. The $i$-th row of the read matrix is the read from the $i$-th head. Let~$\boldsymbol{D}(\boldc,\boldsymbol{\delta}_1,\ldots,\boldsymbol{\delta}_d)\in\{0,1\}^{d\times (n-k)}$ be the read matrix of a $d$-head racetrack memory, where the input is $\boldc\in\{0,1\}^n$ and the deletion locations in the $i$-th head are given by $\boldsymbol{\delta}_i$, $i\in[1,d]$. By this definition, the~$i$-th row of~$\boldsymbol{D}(\boldc,\boldsymbol{\delta}_1,\ldots,\boldsymbol{\delta}_d)$ is~$\boldc_{\boldsymbol{\delta}^c_i}$. 
\begin{example}\label{example:deletions}
Consider a~$3$ head racetrack memory with head distance~$t_1=1$ and~$t_2=2$. Let the deletion location set~$\boldsymbol{\delta}_1=\{2,5,7\}$. Then, we have that~$\boldsymbol{\delta}_2=\{3,6,8\}$ and~$\boldsymbol{\delta}_3=\{5,8,10\}$. Let $\boldc=(1,1,0,1,0,0,0,1,0,1)$ be a sequence of length~$10$. Then, the read matrix is given by
\begin{align*}
    \boldsymbol{D}(\boldc,\boldsymbol{\delta}_1,\boldsymbol{\delta}_2,\boldsymbol{\delta}_3)=\begin{bmatrix}
    1 & 0 & 1 & 0  & 1 & 0 & 1 \\
    1 & 1 & 1 & 0  & 0 & 0 & 1 \\
    1 & 1 & 0 & 1  & 0 & 0 & 0 
\end{bmatrix}. 
\end{align*}
\end{example}
The deletion ball~$\cD_k(\boldc,t_1,\ldots,t_{d-1})$ of a sequence~$\boldc\in \{0,1\}^n$ is the set of all possible read matrices in a $d$-head racetrack memory with input $\boldc$ and head distance $t_i$, $i\in[1,d-1]$, i.e.,
\begin{align*}
 \cD_k(\boldc,t_1,\ldots,t_{d-1})=\{\boldsymbol{D}(\boldc,\boldsymbol{\delta}_1,\ldots,\boldsymbol{\delta}_d):&\boldsymbol{\delta}_{i+1}=\boldsymbol{\delta}_i + t_i,\boldsymbol{\delta}_i\subseteq [1,n], |\boldsymbol{\delta}_i|=k,i\in [1,d-1]\}.   
\end{align*}
A~$d$-head~$k$-deletion code~$\mathcal{C}$ is the set of all sequences such that the deletion balls of any two do not intersect, i.e., for any~$\boldc,\boldc'\in \mathcal{C}$,~$\cD_k(\boldc,t_1,\ldots,t_{d-1})\cap \cD_k(\boldc',t_1,\ldots,t_{d-1})=\emptyset$.

The following notations will be used throughout the paper.
For a matrix~$\boldsymbol{A}$ and two index sets~$\mathcal{I}_1\subseteq [1,d]$ and~$\mathcal{I}_2\subseteq [1,n-k]$, let~$\boldsymbol{A}_{\mathcal{I}_1,\mathcal{I}_2}$ denote the submatrix of~$\boldsymbol{A}$ obtained by selecting the rows~$i\in\mathcal{I}_1$ and the columns~$j\in\mathcal{I}_2$. For any two integer sets~$\cS_1$ and~$\cS_2$, the set $\cS_1\backslash \cS_2=\{x:x\in \cS_1,x\notin \cS_2\}$ denotes the difference between sets~$\cS_1$ and~$\cS_2$.

A sequence~$\boldc\in\{0,1\}^n$ is said to have period~$\ell$ if~$c_{i}=c_{i+\ell}$ for~$i\in [1,n-\ell]$. We use~$L(\boldc,\ell)$ to denote the length of the longest subsequence of consecutive bits in~$\boldc$ that has period~$\ell$. Furthermore, define
\begin{align*}
    L(\boldc,\le k)\triangleq\max_{\ell\le k} L(\boldc,\ell).
\end{align*}
\begin{example}
Let the sequence~$\boldc$ be~$\boldc=(1,1,0,1,1,0,1,0,0)$. Then we have that $L(\boldc,1)=2$, since $\boldc =(\boldsymbol{1},\boldsymbol{1},0,\boldsymbol{1},\boldsymbol{1},0,1,\boldsymbol{0},\boldsymbol{0})$, that $L(\boldc,2)=4$, since $\boldc=(1,1,0,1,\boldsymbol{1},\boldsymbol{0},\boldsymbol{1},\boldsymbol{0},0)$, and that $L(\boldc,3)=7$, since $\boldc=(\boldsymbol{1},\boldsymbol{1},\boldsymbol{0},\boldsymbol{1},\boldsymbol{1},\boldsymbol{0},\boldsymbol{1},0,0)$. Thus, we have~$L(\boldc,\le 3)=7$.
\end{example}

\subsection{Racetrack Memory with Insertion and Deletion errors}
We now describe the notations and problem settings for $d$-head racetrack memories with a combination of insertion and deletion errors, which is similar to $d$-head racetrack memories with deletion errors only. In addition to the deletion errors described by deletion location sets $\{\boldsymbol{\delta}_i\}^d_{i=1}$ satisfying
\begin{align*}
    \boldsymbol{\gamma}_{i+1}=\boldsymbol{\gamma}_i + t_i,
\end{align*}
$i\in[1,d-1]$, 
and $|\boldsymbol{\delta}_i|=r$, $i\in[1,d]$, we consider insertion errors described by insertion location sets $\{\boldsymbol{\gamma}_i\}^d_{i=1}$ satisfying 
\begin{align*}
    \boldsymbol{\gamma}_{i+1}=\boldsymbol{\gamma}_i + t_i,
\end{align*}
$i\in[1,d-1]$, 
where $\boldsymbol{\gamma}_i=\{\gamma_{i,1},\ldots,\gamma_{i,s}\}$ for $i\in [1,d]$,
and the inserted bits  $\boldb_i=(b_{i,1},b_{i,2},\ldots,b_{i,s})$, $i\in[1,d]$. It is assumed that $\gamma_{i,j}\in[0,n]$ for $i\in[1,d]$ and $j\in[1,s]$.
As a result of the insertion errors, bit $b_{i,j}$ is inserted after the $\gamma_{i,j}$-th bit of $\boldc$ in the $i$-th head, for $i\in [1,d]$ and $j\in[1,s]$. When $\gamma_{i,j}=0$, the insertion occurs before $c_1$ in the $i$-th head. We note that $\boldb_i$ can be different for different $i$'s.

We call a deletion error or an insertion error an \emph{edit error}, or \emph{error} in Section \ref{section:correctediterrors}. For edit errors, define the read matrix $\boldsymbol{E}(\boldc,\boldsymbol{\delta}_1,\ldots,\boldsymbol{\delta}_d,\boldsymbol{\gamma}_1,\ldots,\boldsymbol{\gamma}_d,\boldb_1,\ldots,\boldb_d)\in\{0,1\}^{d\times (n+s-r)}$, for $i\in[1,d]$, as follows. The $i$-th row of $\boldsymbol{E}(\boldc,\boldsymbol{\delta}_1,\ldots,\boldsymbol{\delta}_d,\boldsymbol{\gamma}_1,\ldots,$  $\boldsymbol{\gamma}_d,\boldb_1,\ldots,\boldb_d)\in\{0,1\}^{d\times (n+s-r)}$ is obtained by deleting the bits $c_{\ell:\ell\in \boldsymbol{\delta}_i}$ and insert $b_{i,j}$ after  $c_{\gamma_{i,j}}$, for $i\in [1,d]$ and $j\in[1,s]$. In this paper, we consider $k$ edit errors. Hence, $r+s\le k$.
\begin{example}
(Follow-up of Example \ref{example:deletions}).
Consider a~$3$ head racetrack memory with head distance~$t_1=1$ and~$t_2=2$. Let the deletion location set~$\boldsymbol{\delta}_1=\{2,5,7\}$. Then, we have that~$\boldsymbol{\delta}_2=\{3,6,8\}$ and~$\boldsymbol{\delta}_3=\{5,8,10\}$. 
In addition, the insertion location set is given by $\boldsymbol{\gamma}_1=\{0,2\}$. Then, we have $\boldsymbol{\gamma}_2=\{1,3\}$, and $\boldsymbol{\gamma}_3=\{3,5\}$. Let $\boldb_1=(1,1)$, $\boldb_2=(1,0)$, $\boldb_3=(0,1)$
Let $\boldc=(1,1,0,1,0,0,0,1,0,1)$ be a sequence of length~$10$. Then, the read
matrix is given by
\begin{align*}
    \boldsymbol{E}(\boldc,\boldsymbol{\delta}_1,\boldsymbol{\delta}_2,\boldsymbol{\delta}_3,\boldsymbol{\gamma}_1,\boldsymbol{\gamma}_2,\boldsymbol{\gamma}_3,\boldb_1,\boldb_2,\boldb_3)=\begin{bmatrix}
    1 & 1 & 1 & 0 & 1 & 0 & 1 & 0 & 1 \\
    1 & 1 & 1 & 0 & 1 & 0 & 0 & 0 & 1 \\
    1 & 1 & 0 & 0 & 1 & 1 & 0 & 0 & 0 
\end{bmatrix}. 
\end{align*}
\end{example}
Define 
an edit ball~$\cE_k(\boldc,t_1,\ldots,t_{d-1})$ of a sequence~$\boldc\in \{0,1\}^n$ 
as the set of all possible read matrices in an~$d$-head racetrack memory with input $\boldc$ and head distance $t_i$, $i\in[1,d-1]$, i.e.,
\begin{align*}
 \cE_k(\boldc,t_1,\ldots,t_{d-1})=\{&\boldsymbol{E}(\boldc,\boldc,\boldsymbol{\delta}_1,\ldots,\boldsymbol{\delta}_d,\boldsymbol{\gamma}_1,\ldots,\boldsymbol{\gamma}_d,\boldb_1,\ldots,\boldb_d):\boldsymbol{\delta}_{i+1}=\boldsymbol{\delta}_i + t_i,
 \boldsymbol{\gamma}_{i+1}=\boldsymbol{\gamma}_i + t_i, \text{ for $i\in [1,d]$,}\\
 &
 \text{and }
 \boldsymbol{\delta}_i\subseteq [1,n], |\boldsymbol{\delta}_i|=r,
 \boldsymbol{\gamma}_i\subseteq [0,n], |\boldsymbol{\gamma}_i|=s, \boldb_i\in\{0,1\}^s
 \text{ for }i\in [1,d], r+s\le k, \}.   
\end{align*}
A~$d$-head~$k$ edit correction code~$\mathcal{C}$ is the set of all sequences such that the edit balls of any two do not intersect, i.e., for any~$\boldc,\boldc'\in \mathcal{C}$,~$\cE_k(\boldc,t_1,\ldots,t_{d-1})\cap \cE_k(\boldc',t_1,\ldots,t_{d-1})=\emptyset$.
\subsection{Lemmas}
In this section we present lemmas that will be used throughout the paper. Some of them are existing results. The following lemma describes a systematic Reed-Solomon code that can correct a constant number of erasures and can be efficiently computed (See for example \cite{welch1986error}).  
\begin{lemma}\label{lemma:rs}
Let~$k$, $q$, and~$n$ be integers that satisfy~$n+k\le q$. Then,  there exists a map~$RS_k:\bF^n_{q}\rightarrow \bF^k_q$, computable in $poly(n)$ time, such that~$\{(\boldc,RS_k(\boldc)):\boldc\in\bF^n_q\}$ is a~$k$ erasure correcting code.
\end{lemma}
The Reed-Solomon code requires~$O(\log n)$ redundancy for correcting~$k$ erasures.  Correcting a burst of two erasures requires less redundancy when the alphabet size of the code has order~$o(\log n)$. The following code for correcting consecutive two erasures will be used for the case when the number of deletions~$k$ is less than $2d$.
\begin{lemma}\label{lemma:consecutiveerasure}
For any integers~$n$ and~$q$, there exists a map~$ER: \bF^n_{q}\rightarrow \bF^2_q$, computable in~$O(n)$ time, such that the code~$\{(\boldc,ER(\boldc):\boldc\in\bF^n_q\}$ is capable of correcting two consecutive erasures.
\end{lemma}
\begin{proof}
For a sequence~$\boldc=(c_1,\ldots,c_n)$
Let the code~$ER$ be given by 
\begin{align*}
    ER(\boldc)= (\sum^{\lfloor (n-1)/2 \rfloor}_{i=0}c_{2i+1},\sum^{\lfloor n/2 \rfloor}_{i=0}c_{2i}),
\end{align*}
which are the sums of symbols with odd and even indices respectively over field~$\bF_q$. Note that  two consecutive erasures are reduced to two single erasures, one in the even symbols and one in the odd symbols, which can be recovered with the help of~$ER(\boldc)$. Hence,~$(\boldc,ER(\boldc))$ can be recovered from two consecutive erasures.  
\end{proof}
Our construction is based on a systematic deletion code for a single read $d=1$, which was presented in \cite{sima2020deletion}.
\begin{lemma}\label{lemma:hash}
Let~$k$ be a fixed integer. For integers~$m$ and~$n$. There exists a hash function
\begin{align*}
  Hash:\{0,1\}^m\rightarrow \{0,1\}^{\lceil 4k\log m  +o(\log m)\rceil}  
\end{align*}
computable in~$O(m^{2k+1})$ time, such that any sequence~$\boldc\in\{0,1\}^{m}$ can be recovered from its 
length~$m-k$ subsequence with the help of~$Hash(\boldc)$. 
\end{lemma}
We also use the following fact, proved in \cite{levenshtein1966binary}, which implies that a deletion correcting code can be used to correct a combination of deletions and insertions.
\begin{lemma}\label{lemma:deletionequivalenttoinsertion}
    A $k$-deletion correcting code is capable of correcting a combination of $r$ deletions and $s$ insertions, where $r+s\le k$.
\end{lemma}
\begin{remark}
Note that the lemma does not hold in general in a multiple head racetrack memory considered in this paper.
\end{remark}
In addition, in order to synchronize the sequence~$\boldc$ in the presence of deletions, we need to transform~$\boldc$ to a sequence that has a limited length constraint on its periodic subsequences. Such constraint was used in~\cite{chee2018coding}, where it was proved that the redundancy of the code~$\{\boldc:L(\boldc,\le k)\le \lceil \log n \rceil+k+1\}$ is at most~$1$ bit. 
In the following lemma we present a method to transform any sequence to one that satisfies this constraint. The redundancy of our construction is~$k+1$ bits. However, it is small compared to the redundancy of the~$d$-head~$k$-deletion code.  
\begin{lemma}\label{lemma:periodfree}
For any integers~$k$ and~$n$, there exists an injective function~$F:\{0,1\}^n\rightarrow\{0,1\}^{n+k+1}$, computable in~$O_k(n^3\log n)$ time, such that for any sequence~$\{0,1\}^n$, we have that~$L(F(\boldc),\le k )\le 3k+2+\lceil \log n \rceil$. 
\end{lemma}
\begin{proof}
Let~$\boldsymbol{1}^{x}$ and~$\boldsymbol{0}^y$ denote consecutive~$x$~$1$'s and consecutive~$y$~$0$'s respectively.
The encoding procedure for computing~$F(\boldc)$ is as follows.
\begin{enumerate}
    \item \textbf{Initialization:} Let~$F(\boldc)=\boldc$.
    Append~$(\boldsymbol{1}^{k},0)$ to the end of the sequence~$F(\boldc)$. Let~$i=1$ and~$n'=n$. Go to Step~$1$.
    \item \textbf{Step 1:} If~$i\le n'-2k-\lceil \log n \rceil-1$ and~$F(\boldc)_{[i,i+2k+\lceil \log n \rceil+1]}$ has period~$p\le k$, let~$p_{min}$ be the smallest period of~$F(\boldc)_{[i,i+2k+\lceil \log n \rceil+1]}$.
    Delete~$F(\boldc)_{[i,i+2k+\lceil \log n \rceil+1]}$ from~$F(\boldc)$ and  append~$(\boldsymbol{1}^{k-p_{min}},0,F(\boldc)_{[i,i+p_{min}-1]},i,\boldsymbol{0}^{k+1})$ to the end of~$F(\boldc)$, i.e., set~$F(\boldc)_{[i,n-k-\lceil \log n \rceil-1]}=F(\boldc)_{[i+2k+\lceil \log n \rceil+2,n+k+1]}$ and $F(\boldc)_{[n-k-\lceil \log n \rceil,n+k+1]}=(\boldsymbol{1}^{k-p_{min}},0,$ $F(\boldc)_{[i,i+p_{min}-1]},i,\boldsymbol{0}^{k+1})$. 
    Let~$n'=n'-2k-\lceil \log n \rceil-2$ and~$i=1$. Repeat.
    Else go to Step~$2$.
    \item \textbf{Step 2:}
    If~$i\le n'-2k-\lceil \log n \rceil-1$, let~$i=i+1$ and go to Step~$1$.
    Else output~$F(\boldc)$. 
\end{enumerate}
It can be verified that the length of the sequence~$F(\boldc)$ remains to be~$n+k+1$ during the procedure. 
The number~$n'$ in the procedure denotes the number such that~$F(\boldc)_{[n'+1,n+k+1]}$ are appended bits and~$F(\boldc)_{[1,n']}$ are the remaining bits in~$\boldc$ after deletions.
Since either~$i$ increases to~$n'$ or~$n'$ decreases in Step~$1$.  The algorithm terminates within~$O(n^2)$ times of Step~$1$ and Step~$2$. Since it takes~$O(k(3k+2+\log n)n)$ time to check the periodicity in Step~$1$. The total complexity is~$O_k(n^3\log n)$.

We now prove that~$L(F(\boldc),\le k)\le 3k+2+\lceil \log n \rceil$. Let~$n'$ be the number computed in the encoding procedure.
According to the encoding procedure, we have that~$L(F(\boldc)_{[j,j+2k+1+\lceil \log n \rceil]},\le k)\le 2k+1+\lceil \log n \rceil$ for~$j\le n'-2k-\lceil \log n \rceil-1$, since any subsequence~$F(\boldc)_{[j,j+2k+1+\lceil \log n \rceil]}$ with period not greater than~$k$ is deleted. Therefore~$L(F(\boldc)_{[j,j+3k+1+\lceil \log n \rceil]},\le k)\le 3k+2+\lceil \log n \rceil$ for~$j\le n'-2k-\lceil \log n \rceil-1$.  For~$n'-2k-\lceil \log n \rceil\le j\le n'$, the sequence~$F(\boldc)_{[j,j+2k+1+\lceil \log n \rceil]}$ contains~$F(\boldc)_{[n'+1,n'+k+1]}=(\boldsymbol{1}^k,0)$, which does not have period not greater than~$k$. 
Hence we have that~$L(F(\boldc)_{[j,j+3k+1+\lceil \log n \rceil]},\le k)\le 3k+2+\lceil \log n \rceil$.
For~$j>n'$, the sequence~$F(\boldc)_{[j,j+3k+1+\lceil \log n \rceil]}$ contains~$\boldsymbol{0}^{k+1}$ as~$k+1$ consecutive bits. Hence, if~$F(\boldc)_{[j,j+3k+1+\lceil \log n \rceil]}= 3k+2+\lceil \log n \rceil$, we have that~$F(\boldc)_{[j,j+3k+1+\lceil \log n \rceil]}=\boldsymbol{0}^{3k+2+\lceil \log n \rceil}$. However, this is impossible since $F(\boldc)_{[j,j+3k+1+\lceil \log n \rceil]}$ contains either the location index~$i$ to the left of~$\boldsymbol{0}^{k+1}$ or the bits~$(\boldsymbol{1}^{k-p_{min}},0,F(\boldc)_{[i,i+p_{min}-1]})$ to the right of~$\boldsymbol{0}^{k+1}$, both of which can not be all zero. Therefore, we conclude that~$L(\boldc,\le k)\le 3k+2+\lceil \log n \rceil$.
Given~$F(\boldc)$, the decoding procedure for computing~$\boldc$ is as follows.
\begin{enumerate}
    \item \textbf{Initialization:} Let~$\boldc=F(\boldc)$ and go to Step~$1$.
    \item \textbf{Step 1:} If~$\boldc_{[n+1,n+k+1]}\ne (\boldsymbol{1}^{k},0)$, let~$j$ be the length of the first~$1$ run in~$\boldc_{[n-k-\lceil \log n \rceil,n+k+1]}$ and let~$p$ be the decimal representation of~$\boldc_{n-\lceil \log n \rceil+1,n}$.
    Let~$\bolda$ be a sequence of length~$2k+\lceil \log n \rceil+2$ and period~$k-j$. The first~$k-j$ bits of~$\bolda$ is given by~$\boldc_{[n-k-\lceil \log n \rceil+j+1,n-\lceil \log n \rceil]}$. 
    Delete~$\boldc_{[n-k-\lceil \log n \rceil,n+k+1]}$ from~$\boldc$ and insert~$\bolda$ at location~$p$ of~$\boldc$, i.e., let~$\boldc_{[p+2k+\lceil \log n \rceil+2,n+k+1]}=\boldc_{[p,n-k-\lceil \log n \rceil-1]}$ and~$\boldc_{[p,p+2k+\lceil \log n \rceil+1]}=\bolda$.
    Repeat. 
    Else output~$\boldc$
\end{enumerate}
Note that the encoding procedure consists of a series of deleting and appending operations. The decoding procedure consists of a series of deletion and inserting operations.
Let~$F_i(\boldc)$,~$i\in [0,R]$ be the sequence~$F(\boldc)$ obtained after the~$i$-th deleting and appending operation in the encoding procedure, where~$R$ is the number of deleting and appending operations in total in the encoding procedure. We have that~$F_0(\boldc)=\boldc$ and~$F_R(\boldc)$ is the final output~$F(\boldc)$. It can be seen that the decoding procedure obtains~$F_{R-i}(\boldc)$,~$i\in [0,R]$ after the~$i$-th deleting and inserting operation. Hence the function~$F(\boldc)$ is injective.
\end{proof}
Finally, we restate one of the main results in~\cite{chee2018coding} that will be used in our construction. The result guarantees a procedure to correct~$d-1$ deletions in a~$d$-head racetrack memory, given that the distance between consecutive heads are large enough. 
\begin{lemma}\label{lemma:kdeletionkheads}
Let~$d\le k$ be two integers and~$\mathcal{C}$ be a~$(k-d+1)$-deletion code, then~$\mathcal{C}\cap\{\boldc:L(\boldc,\le k)\le T\}$ is a~$d$-head~$k$-deletion correcting code, given that the distance between consecutive heads~$t\ge T[k(k-1)/2+1]+(7k-k^3)/6$ for~$i\in [1,d-1]$. 
\end{lemma}

\section{Correcting up to~$2d-1$ deletions with~$d$ heads}\label{section:lessthan2mdeletions}
In this section we construct a $d$-head~$k$-deletion~code for cases when~$k\le 2d-1$. To this end, we first present a lemma that is crucial in our code construction. The lemma states that the range of deletion locations can be narrowed down to a list of short intervals. Moreover, the number of deletions within each interval can be determined. The proof of the lemma will be given in Section \ref{section:synchronization}. 
Before presenting the lemma, we give the following definition, which describes a property of the intervals we look for.
\begin{definition}\label{definition:interval}
Let~$\boldsymbol{\delta}_i=\{\delta_{i,1},\ldots,\delta_{i,k}\}$ be the set of deletion locations in the~$i$-th head of a~$d$-head racetrack memory, i.e. $\boldsymbol{\delta}_{i+1}=\boldsymbol{\delta}_{i}+t_i$, for~$i\in [1,d-1]$. 
An interval~$\mathcal{I}$ is \emph{deletion isolated} if
\begin{align*}
\boldsymbol{\delta}_{i+1}\cap \mathcal{I}= t_i + \boldsymbol{\delta}_{i}\cap \mathcal{I},
\end{align*}
for~$i\in [1,d-1]$. 
\end{definition}
\begin{example}
Consider a~$3$-head racetrack memory with head distances~$t_1=1$ and~$t_2=2$. Let the length of the sequence~$\boldc$ be~$n=22$ and the deletion positions in three heads be given by
\begin{align*}
    &\boldsymbol{\delta}_1=\{1,2,4,8,14,17\},\\
    &\boldsymbol{\delta}_2=\{2,3,5,9,15,18\},~\mbox{and}\\
    &\boldsymbol{\delta}_3=\{4,5,7,11,17,20\},
\end{align*}
Then the intervals~$[1,7]$,~$[8,12]$, and~$[14,22]$ are all deletion isolated.
\end{example}
Intuitively, an interval~$\mathcal{I}$ is deletion isolated when the subsequences~$\boldc_{\mathcal{I}\cap \boldsymbol{\delta}^c_i}$ for~$i\in [1,d]$ can be regarded as the~$d$ reads of the sequence~$\boldc_{\mathcal{I}}$ in a~$d$-head racetrack memory after~$|\boldsymbol{\delta}_1\cap \mathcal{I}|$ deletions in each head. 
\begin{figure}
\centering
\definecolor{ffqqqq}{rgb}{1,0,0}
\definecolor{rvwvcq}{rgb}{0.08235294117647059,0.396078431372549,0.7529411764705882}
\begin{tikzpicture}[line cap=round,line join=round,>=triangle 45,x=1cm,y=1cm]
xtick={-7,-6,...,7},
ytick={-7,-6,...,3},]
\clip(-9.5,-1) rectangle (7.2,5);





\draw (-3.75,2.8) node[anchor=center] {$\mathcal{I}_1$};
\draw (3.75,2.8) node[anchor=center] {$\mathcal{I}_J$};

\draw [line width=1pt] (-7,2.5)-- (7,2.5);
\draw [line width=1pt] (-7,1.5)-- (7,1.5);
\draw [line width=1pt] (-7,0.5)-- (7,0.5);
\draw [line width=1pt] (-7,-0.5)-- (7,-0.5);
\draw [line width=1pt] (-7,2.5)-- (-7,-0.5);
\draw [line width=1pt] (7,2.5)-- (7,-0.5);

\draw [line width=1pt] (0.5,2.5)-- (0.5,-0.5);
\draw [line width=1pt] (-0.5,2.5)-- (-0.5,-0.5);

\draw (-8.2,2) node[anchor=center] {$D_{1,[1,n+R-k]}$};
\draw (-8.2,1) node[anchor=center] {$D_{2,[1,n+R-k]}$};
\draw (-8.2,0) node[anchor=center] {$D_{3,[1,n+R-k]}$};

\draw (-6.5,2) node[anchor=center] {$=$};
\draw (-6.5,1) node[anchor=center] {$=$};
\draw (-6.5,0) node[anchor=center] {$=$};

\draw (-1,2) node[anchor=center] {$=$};
\draw (-1,1) node[anchor=center] {$=$};
\draw (-1,0) node[anchor=center] {$=$};

\draw (0,2.8) node[anchor=center] {$\ldots$};
\draw (0,2) node[anchor=center] {$\ldots$};
\draw (0,1) node[anchor=center] {$\ldots$};
\draw (0,0) node[anchor=center] {$\ldots$};

\draw (1,2) node[anchor=center] {$=$};
\draw (1,1) node[anchor=center] {$=$};
\draw (1,0) node[anchor=center] {$=$};

\draw (6.5,2) node[anchor=center] {$=$};
\draw (6.5,1) node[anchor=center] {$=$};
\draw (6.5,0) node[anchor=center] {$=$};

\draw (-5.25,2) node[anchor=center] {$*$};
\draw (-3.75,2) node[anchor=center] {$*$};
\draw (-4.25,1) node[anchor=center] {$*$};
\draw (-2.75,1) node[anchor=center] {$*$};
\draw (-3.75,0) node[anchor=center] {$*$};
\draw (-2.25,0) node[anchor=center] {$*$};

\draw (1.75,2) node[anchor=center] {$*$};
\draw (3.75,2) node[anchor=center] {$*$};
\draw (2.75,1) node[anchor=center] {$*$};
\draw (4.75,1) node[anchor=center] {$*$};
\draw (3.25,0) node[anchor=center] {$*$};
\draw (5.25,0) node[anchor=center] {$*$};

\draw [line width=1pt,dash pattern=on 2pt off 2pt] (-6,3.2)-- (-6,-0.5);
\draw [line width=1pt,dash pattern=on 2pt off 2pt] (-1.5,3.2)-- (-1.5,-0.5);

\draw [line width=1pt] (-5.5,2.5)-- (-5.5,1.5);
\draw [line width=1pt] (-5,2.5)-- (-5,1.5);
\draw [line width=1pt] (-4,2.5)-- (-4,1.5);
\draw [line width=1pt] (-3.5,2.5)-- (-3.5,1.5);

\draw [line width=1pt] (-4.5,1.5)-- (-4.5,0.5);
\draw [line width=1pt] (-4,1.5)-- (-4,0.5);
\draw [line width=1pt] (-3,1.5)-- (-3,0.5);
\draw [line width=1pt] (-2.5,1.5)-- (-2.5,0.5);

\draw [line width=1pt] (-4,0.5)-- (-4,-0.5);
\draw [line width=1pt] (-3.5,0.5)-- (-3.5,-0.5);
\draw [line width=1pt] (-2.5,0.5)-- (-2.5,-0.5);
\draw [line width=1pt] (-2,0.5)-- (-2,-0.5);

\draw [line width=1pt,dash pattern=on 2pt off 2pt] (1.5,3.2)-- (1.5,-0.5);
\draw [line width=1pt,dash pattern=on 2pt off 2pt] (6,3.2)-- (6,-0.5);

\draw [line width=1pt] (1.5,2.5)-- (1.5,1.5);
\draw [line width=1pt] (2,2.5)-- (2,1.5);
\draw [line width=1pt] (3.5,2.5)-- (3.5,1.5);
\draw [line width=1pt] (4,2.5)-- (4,1.5);

\draw [line width=1pt] (2.5,1.5)-- (2.5,0.5);
\draw [line width=1pt] (3,1.5)-- (3,0.5);
\draw [line width=1pt] (4.5,1.5)-- (4.5,0.5);
\draw [line width=1pt] (5,1.5)-- (5,0.5);

\draw [line width=1pt] (3,0.5)-- (3,-0.5);
\draw [line width=1pt] (3.5,0.5)-- (3.5,-0.5);
\draw [line width=1pt] (5,0.5)-- (5,-0.5);
\draw [line width=1pt] (5.5,0.5)-- (5.5,-0.5);

\fill[line width=2pt,fill=black,fill opacity=0.1] (-6,2.5) -- (-6,-0.5) -- (-1.5,-0.5) -- (-1.5,2.5) -- cycle;
\fill[line width=2pt,fill=black,fill opacity=0.1] (6,2.5) -- (6,-0.5) -- (1.5,-0.5) -- (1.5,2.5) -- cycle;
\end{tikzpicture}\caption{An illustration of Lemma~\ref{lemma:synchronization}. The~$*$ entries denote deletion in the heads. The read~$D_{i,[1,n+R]}$ in each head is obtained after deleting the~$*$ entries from~$\boldc$.}
\label{figure:synchronizationlemma}
\end{figure}
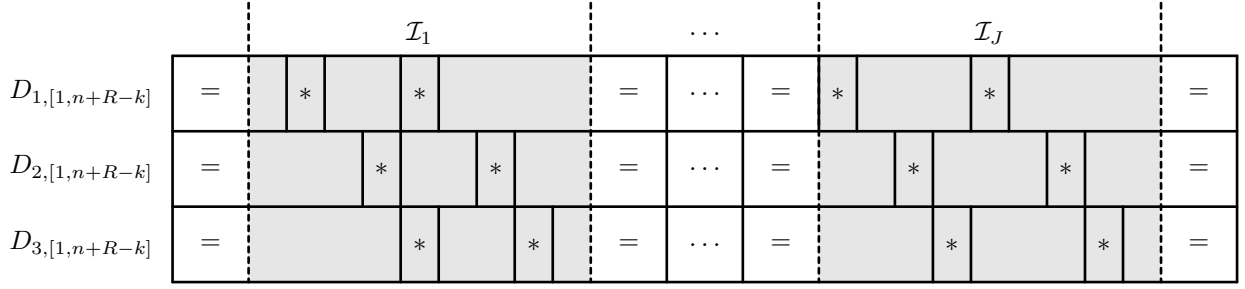
\begin{lemma}\label{lemma:synchronization}
(Proofs appear in Section \ref{section:synchronization}.) For any positive integers~$n$ and~$R\ge k+1$, let~$\boldc\in\{0,1\}^{n+R}$ be a sequence such that~$L(\boldc_{[1,n+k+1]},\le k)\le 3k+\lceil \log n \rceil +2\triangleq T$. Let the distance~$t_i$ between head~$i$ and head $i+1$ satisfy~$t_i\ge (4k+1)(T+2k+1)\triangleq T_{min}$, $i\in[1,d-1]$. Let $t_{max}=\max_{i\in\{1,\ldots,d-1\}} t_i$ be the largest distance between two consecutive heads. Then given~$\boldsymbol{D}\in \cD_k(\boldc,t_1,\ldots,t_{d-1})$, it is possible to find a set of $J\le k$ disjoint and deletion isolated  intervals~$\mathcal{I}_j\subseteq [1,n+R]$,~$j\in [i,J]$ such that~$\boldsymbol{\delta}_w\subset \cup^J_{j=1}\mathcal{I}_j$ for~$w\in [1,d]$ and
\begin{align*}
|\mathcal{I}_j\cap [1,n+k+1]|\le (2\lfloor (2t_{max}+T+1)/2\rfloor +1) kd+\lfloor (2t_{max}+T+1)/2\rfloor+k\triangleq B,    
\end{align*}
for~$j\in [1,J]$. Moreover,~$|\boldsymbol{\delta}_1\cap\mathcal{I}_j|$ can be determined for~$j\in [1,J]$.  
\end{lemma}
An illustration of Lemma~\ref{lemma:synchronization} is shown in Fig.~\ref{figure:synchronizationlemma}. Since the interval~$\mathcal{I}_j$ is deletion isolated for~$j\in [1,J]$, all rows of~$\boldsymbol{D}$ are aligned in locations $[1,n+R]\backslash (\cup^J_{j=1}\mathcal{I}_j)$, i.e., the entries in the $i$-th column of~$A$ correspond to the same bit in $\boldc$ for $i\in [1,n+R]\backslash (\cup^J_{j=1}\mathcal{I}_j)$. Let~$\boldc\in\{0,1\}^{n+R}$ be a sequence satisfying~$L(\boldc_{[1,n+k+1]},\le k)\le T$.
By virtue of Lemma~\ref{lemma:synchronization}, 
 the bit~$c_i$  can be determined by 
 \begin{align}\label{equation:recoverblocks}
 c_i = \boldsymbol{D}_{1,i-\sum_{j: \mathcal{I}_j\subseteq [1,i-1]}|\boldsymbol{\delta}_1\cap\mathcal{I}_j|}
 \end{align}
 for~$i\in [1,n+k+1]\backslash (\cup^J_{j=1}\mathcal{I}_j)$. In addition, let~$\mathcal{I}_j=[b^{min}_j,b^{max}_j]$ for~$j\in [1,J]$ such that $b^{max}_{j-1}<b^{min}_{j}$ for $j\in [2,J]$. Since~$\mathcal{I}_j$ is deletion isolated for~$j\in [1,J]$, the submatrix
\begin{align*}
\boldsymbol{D}_{[1,d],[b^{min}_j-\sum^{j-1}_{i=1}|\boldsymbol{\delta_1}\cap \mathcal{I}_i|,b^{max}_j-\sum^j_{i=1}|\boldsymbol{\delta_1}\cap \mathcal{I}_i|]}\in\cD_{|\boldsymbol{\delta_1}\cap[b^{min}_j,b^{max}_j]|}(\boldc_{\mathcal{I}_j},t_1,\ldots,t_{d-1})
\end{align*}
can be regarded as the~$d$ reads of the subsequence~$\boldc_{\mathcal{I}_j}$ in a $d$-head racetrack memory.
According to Lemma~\ref{lemma:kdeletionkheads}, the bits~$\boldc_{\mathcal{I}_j}$ with~$|\boldsymbol{\delta}_1\cap\mathcal{I}_j|< d$ can be recovered from~
\begin{align*}
\boldsymbol{D}_{[1,d],[b^{min}_j-\sum^J_{i=j+1}|\boldsymbol{\delta_1}\cap \mathcal{I}_i|,b^{max}_j-\sum^J_{i=j}|\boldsymbol{\delta_1}\cap \mathcal{I}_i|]}
\end{align*}
if the head distance satisfies~$t_i\ge T[k(k-1)/2+1]+(7k-k^3)/6$ for~$i\in\{1,\ldots,d-1\}$. Note that there is at most a single interval~$\mathcal{I}_{j_1}$ satisfying~$|\boldsymbol{\delta}_1\cap\mathcal{I}_{j_1}|\ge d$ when~$k\le 2d-1$. Hence, we are left to recover interval~$\mathcal{I}_{j_1}$.
 
Split $\boldc_{[1,n+k+1]}$ into blocks
 \begin{align}\label{equation:block}
 \bolda_i = \boldc_{[(i-1)B+1,\min \{iB,n+k+1\}]},~i\in [1,\lceil (n+k+1)/B \rceil]
 \end{align}
of length~$B$ except that~$\bolda_{\lceil (n+k+1)/B \rceil}$ may have length shorter than~$B$. 
Since~$|\mathcal{I}_{j_1}\cap [1,n+k+1]|\le B$, the interval~$\mathcal{I}_{j_1}$ spans over at most two blocks~$\bolda_{j'_1}$ and~$\bolda_{j'_1+1}$. 
It then follows that there are at most two consecutive blocks, where~$\mathcal{I}_{j_1}$ lies in, that remain to be recovered. Moreover, at most~$k$ deletions occur in  interval~$\mathcal{I}_{j_1}$, and hence in blocks~$\bolda_{j'_1}$ and~$\bolda_{j'+_1+1}$.  

For an integer~$n$ and a sequence~$\boldc\in\{0,1\}^{n+k+1}$ of length~$n+k+1$, let the function~$S:\{0,1\}^{n+k+1}\rightarrow \bF^{\lceil (n+k+1)/B \rceil}_{4k\log B+o(\log B)}$ be defined by
\begin{align}\label{equation:split}
    S(\boldc) = (Hash(\bolda_1),Hash(\bolda_2),\ldots,Hash(\bolda_{\lceil (n+k+1)/B \rceil})),
\end{align}
where~$\bolda_{i}$,~$i\in[1,\lceil (n+k+1)/B \rceil]$ are the blocks of~$\boldc$ defined in Eq.~\eqref{equation:block}. The function~$Hash(\bolda_{\lceil (n+k+1)/B \rceil})$, defined in Lemma \ref{lemma:hash}, takes~$\bolda_{\lceil (n+k+1)/B \rceil}$ as input of length at most~$B$.
The sequence~$S(\boldc)$ is a concatenation of the hashes~$Hash$ of blocks of~$\boldc$. 
\begin{lemma}\label{lemma:splithash}
If~$B>k$,
there exists a function~$DecS: \{0,1\}^{n+1}\times \{0,1\}^{\lceil (n+k+1)/B \rceil (4k\log B+o(\log B))}\rightarrow \{0,1\}^{n+k+1}$, such that for any sequence~$\boldc\in\{0,1\}^{n+k+1}$ and its length~$n+1$ subsequence~$\boldd\in\{0,1\}^{n+1}$, we have that~$DecS(\boldd,S(\boldc))=\boldc$, i.e., the sequence~$\boldc$ can be recovered from~$k$ deletions with the help of~$S(\boldc)$.
\end{lemma}
\begin{proof}
Note that~$\boldd_{[(i-1)B+1,\min \{iB,n+k+1\}-k]}$ is a length~$B-k$ subsequence of the block~$\bolda_i$ for~$i\in\{1,\ldots,\lceil (n+k+1)/B \rceil\}$. 
According to Lemma~\ref{lemma:hash}, the block~$\bolda_i$ can be recovered from~$\boldd_{(i-1)B+1,\max \{iB,n+k+1\}-k}$ with the help of~$Hash(\bolda_i)$. Thus the  sequence~$\boldc$ can be recovered.
\end{proof}
We are now ready to present the code construction.   For any sequence~$\boldc\in\{0,1\}^n$, define the following encoding function: 
\begin{align}\label{equation:encoding}
    Enc_1(\boldc)=(F(\boldc),R^{'}_1(\boldc),R^{''}_1(\boldc))
\end{align}
where~
\begin{align}
R^{'}_1(\boldc)&=ER(S(F(\boldc))),\nonumber\\
R^{''}_1(\boldc)&=Rep_{k+1}(Hash(R'_1(\boldc))),
\end{align}
and the function~$Rep_{k+1}$ is a $k+1$-fold repetition function that repeats each bit~$k+1$ times.
Note that we use~$F(\boldc)\in\{0,1\}^{n+k+1}$, where $F$ is defined in Lemma \ref{lemma:periodfree}, to obtain a sequence satisfying~$L(F(\boldc),\le k)\le T$ so that Lemma~\ref{lemma:synchronization} can be applied. The redundancy consists of two layers. The function~$R^{'}_1(\boldc)$ can be regarded as the first layer redundancy, with the help of which~$F(\boldc)$ can be recovered from~$k$ deletions. It
computes the redundancy of a code that corrects two consecutive symbol erasures in~$S(F(\boldc))$. The function~$R^{''}_1(\boldc)$ can be seen as the second layer redundancy that helps recover itself and~$R^{'}_1(\boldc)$ from~$k$ deletions.
 
When the head distance~$t_i$ satisfies~$t_i=\max\{(3k+\lceil \log n \rceil+2)[k(k-1)/2+1]+(7k-k^3)/6,(4k+1)(5k+\lceil \log n \rceil+3)\}$ for~$i\in [1,d-1]$, the length of~$R^{'}_1(\boldc)$ is given by~$N_1=4k\log B+o(\log B)=4k\log \log n+o(\log\log n)$. The length of~$R^{''}_1(\boldc)$ is~$N_2=4k(k+1)\log N_1+O(1)=o(\log\log n)$. 
The length of the codeword~$Enc_1(\boldc)$ is given by $N=n+k+1+N_1+N_2=4k\log\log n+o(\log \log n)$. 
The next theorem proves Theorem~\ref{theorem:main} for cases when~$d\le k\le 2d-1$.
\begin{theorem}\label{theorem:lessthan2d}
The set~$\mathcal{C}_1=\{Enc_1(\boldc):\boldc\in\{0,1\}^n\}$ is a $d$-head~$k$-deletion~correcting code for~$d\le k\le 2d-1$, if the distance between any two consecutive heads satisfies~$t_i=\max\{(3k+\lceil \log n \rceil+2)[k(k-1)/2+1]+(7k-k^3)/6,(4k+1)(5k+\lceil \log n \rceil+3)\}$,~$i\in [1,d-1]$. The code~$\mathcal{C}_1$ can be constructed, encoded, and decoded in~$O(n^{2k+2})$ time. The redundancy of~$\mathcal{C}_1$ is~$N-n=4k\log \log n +o(\log\log n)$.
\end{theorem}
\begin{proof}
For any~$\boldsymbol{D}\in \cD_k(\boldc)$, let~$\boldd=\boldsymbol{D}_{1,[1,N-k]}$ be the first row of~$\boldsymbol{D}$, i.e., the first read. The sequence~$\boldd$ is a length~$N-k$ subsequence of~$Enc_1(\boldc)$.
We first show how to recover~$R^{'}_1(\boldc)$ from~$\boldd$. 
Note that~$\boldd_{[N-N_2+1,N-k]}$ is a length~$N_2-k$ subsequence of~$R^{''}_1(\boldc)$, the~$k+1$-fold repetition of~$Hash(R^{'}_1(\boldc))$. Since a~$k+1$-fold repetition code is a~$k$-deletion code, the hash function~$Hash(R^{'}_1(\boldc))$ can be recovered. Furthermore, we have that~$\boldd_{[n+k+2,n+k+1+N_1-k]}$ is a length~$N_1-k$ subsequence of~$R^{'}_1(\boldc)$. Hence according to Lemma~\ref{lemma:hash}, we can obtain~$R^{'}_1(\boldc)$ from~$\boldd_{[n+k+2,n+k+1+N_1-k]}$, with the help of~$Hash(R^{'}_1(\boldc))$.

Next, we show how to use~$R^{'}_1(\boldc)$ to recover~$F(\boldc)$. 
Note the fact that~$L(F(\boldc),\le k)\le T$.
From Lemma~\ref{lemma:synchronization}
and the discussion that follows, we can separate~$F(\boldc)$ into blocks~$\bolda_{i}$,~$i\in [1,\lceil (n+K+1)/B \rceil ]$, of length~$B$. 
and recover all but at most two consecutive blocks~$\bolda_{j_1}$ and~$\bolda_{j_1+1}$. 
This implies that $S(F(\boldc))$ can be retrieved with consecutive at most two symbol errors, the position of which can be identified, by looking for the unique interval $\cI_j$ such that $|\cI_j\cap\boldsymbol{\delta}_1|\ge d$. Hence we can use~$R^{'}_1(\boldc)$ to recover~$S(F(\boldc))$ and find the hashes~$Hash(\bolda_{j_1})$ and~$Hash(\bolda_{j_1+1})$. Note that~$D_{1,[1,n+1]}$ is a length~$n+1$ subsequences of~$F(\boldc)$. Hence from Lemma~\ref{lemma:splithash} the sequence~$F(\boldc)$ and thus~$\boldc$ can be recovered given~$S(F(\boldc))$. The computation of~$S(F(\boldc))$, which computes ~$O(n/B)$ times the hashes~$Hash(\bolda_i)$,~$[1,\lceil (n+k+1)/B \rceil]$, constitutes the main part of the computation complexity of~$Enc_1(\boldc)$.
Since the computation of~$Hash(\bolda_i)$ takes~$O(\log^{2k+1}n )$ time for each $i\in[1,\lceil (n+K+1)/B \rceil ]$. It takes~$O(n\log^{2k}n)$ time to compute~$Enc_1(\boldc)$.
\end{proof}

\section{Proof of Lemma \ref{lemma:synchronization} }
\label{section:synchronization}
Let~$\boldsymbol{D}\in \cD_k(\boldc,t_1,\ldots,t_{d-1})$ be the~$d$ reads from all heads, where~$\boldc\in\{0,1\}^{n+R}$ satisfies~$L(\boldc_{[1,n+k+1]},\le k)\le T$. Then~$\boldsymbol{D}$ is a~$d$ by~$n+R-k$ matrix.
The proof of Lemma~\ref{lemma:synchronization} consists of two steps. The first step is to identify a set of disjoint intervals~$\mathcal{I}^*_j$,~$j\in [1,J]$ that satisfy
\begin{enumerate}
	\item[\bf{(P1)}]There exist a set of disjoint and 
	deletion isolated intervals~$\mathcal{I}_j$,~$j\in [1,J]$, such that~$\boldsymbol{D}_{w,\mathcal{I}^*_j} = \boldc_{\mathcal{I}_j\cap \boldsymbol{\delta}^c_w}$ for~$w\in [1,d]$ and~$j\in [1,J]$, i.e., the subsequence~$\boldsymbol{D}_{w,\mathcal{I}^*_j}$ comes from~$\boldc_{\mathcal{I}_j}$ in the~$w$-th read after deleting $\boldc_{\mathcal{I}_j\cap \boldsymbol{\delta}_w}$.
	\item[\bf{(P2)}]$J\le k$ and~$\boldsymbol{\delta}_w\subseteq\cup^J_{j=1}\mathcal{I}_j$ for $w\in[1,d]$.
	\item[\bf{(P3)}]$|\mathcal{I}^*_j\cap [1,n+1]|\le B-k$
\end{enumerate}
The second step is to determine the number of deletions~$|\boldsymbol{\delta}_{w}\cap \mathcal{I}_j|$  for~$w\in [1,d]$ and~$j\in [1,J]$, that happen in each interval in each head, based on~$\boldsymbol{D}_{[1,d],\mathcal{I}^*_j}$. 
Then we have that
\begin{align*}
\mathcal{I}_j = [i_{2j-1}+\sum^{j-1}_{\ell=1}|\boldsymbol{\delta}_{1}\cap \mathcal{I}_{\ell}|,i_{2j}+\sum^{j}_{\ell=1}|\boldsymbol{\delta}_{1}\cap \mathcal{I}_{\ell}|],
\end{align*}where~$i_{2j-1}$ and~$i_{2j}$ are the starting and ending points of the interval~$\mathcal{I}^*_j$. It is assumed that~$i_j>i_l$ for~$j>l$. The disjointness of~$\mathcal{I}_j$,~$j\in [1,J]$ follows from the fact that~$\mathcal{I}^*_j$,~$j\in [1,J]$ are disjoint.
The two steps will be made explicit in the following two subsections respectively. 


\subsection{Identifying Intervals~$\mathcal{I}^*_j$}\label{subsection:a}
The procedure for identifying intervals~$\mathcal{I}^*_j$,~$j\in [1,J]$,  is as follows. 
\begin{enumerate}
\item \textbf{Initialization:}
Set all integers~$m\in [1,n+R-k]$ unmarked.
Let~$i=1$. Find the largest positive integer~$L$ such that the sequences~$\boldsymbol{D}_{w,[i,i+L-1]}$ are equal for all~$w\in [1,d]$. If such~$L$ exists and satisfies $L>t_{max}$, mark the integers~$m\in[1,L-t_{max}]$ and go to Step~$1$. Otherwise, go to Step~$1$.
\item \textbf{Step 1:} Find the largest positive integer~$L$ such that the sequences~$\boldsymbol{D}_{w,[i,i+L-1]}$ are equal for all~$w\in [1,d]$. Go to Step~$2$. If no such~$L$ is found, set~$L=0$ and go to Step $2$. 
\item \textbf{Step 2:} If~$L\ge 2{t_{max}}+T+1$, mark the integers $m\in[i+t_{max},\min\{i+L-1,n+1\}-t_{max}]$. Set~$i=i+L$ and go to Step $3$. Else~$i=i+1$ and go to Step $3$.
\item \textbf{Step 3:} If~$i\le n+1$, go to Step 1. Else go to Step 4.
\item \textbf{Step 4:} If the number of unmarked intervals\footnote{An unmarked interval~$[i,j]$ means that~$m\in[i,j]$ are not marked and~$i-1$ and~$j+1$ are marked. It is assumed that~$0$ and~$n+R-k+1$ are marked.} within~$[1,n+1]$ is not greater than~$k$, output all unmarked intervals. Else output the first~$k$ intervals, i.e., the intervals with the minimum~$k$ starting indices.
\end{enumerate}
We prove that the output intervals 
satisfy the above constraints \textbf{(P1)}, \textbf{(P2)}, and \textbf{(P3)}. The following lemma will be used. 
\begin{lemma}\label{lemma:identifyintervals}
Let~$\boldsymbol{D}\in \cD_k(\boldc)$ for some sequence~$\boldc$ satisfying~$L(\boldc_{[1,n+k+1]},\le k)\le T$.
Let~$t_{max}=\max_{i\in [1,d-1]}t_i$ such that~$t_w\ge k(T+1)+1$ for $w\in[1,d-1]$.
If the sequences~$\boldsymbol{D}_{w,[i_1,i_2]}$ are equal for all $w\in[1,d]$ in some interval~$[i_1,i_2]\subseteq [1,n+1]$ with length~$i_2-i_1+1\ge 2t_{max}+T+1$, then no deletions occur within bits~$\boldsymbol{D}_{w,[i_1+t_{max},i_2-t_{max}]}$ for all~$w$, i.e., there exists integers~$i'_1=i_1+t_{max}+|\boldsymbol{\delta}_j\cap[1,i'_1-1]|$ and~$i'_2=i_2-t_{max}+|\boldsymbol{\delta}_j\cap[1,i'_2-1]|$, such that~$\boldc_{[i'_1,i'_2]}=\boldsymbol{D}_{w,[i_1+t_{max},i_2-t_{max}]}$ and~$[i'_1,i'_2]\cap \boldsymbol{\delta}_w =\emptyset$ for~$w\in [1,d]$. In addition, both intervals~$[1,i'_1-1]$ and~$[i'_2+1,n+R]$ are deletion isolated.
\end{lemma}
\begin{proof}
Let~$c_{i'_0}$,~$c_{i'_1}$,~$c_{i'_2}$, and~$c_{i'_3}$ be the bits that become~$\boldsymbol{D}_{1,i_1}$,~$\boldsymbol{D}_{1,i_1+t_{max}}$,~$\boldsymbol{D}_{1,i_2-t_{max}}$, and~$\boldsymbol{D}_{1,i_2}$ respectively after deletions, i.e.,~$i'_0-|\boldsymbol{\delta}_1\cap[1,i'_0-1]|=i_1$,~$i'_1-|\boldsymbol{\delta}_1\cap[1,i'_1-1]|=i_1+t_{max}$,~$i'_2-|\boldsymbol{\delta}_1\cap[1,i'_2-1]|=i_2-t_{max}$, and~$i'_3-|\boldsymbol{\delta}_1\cap[1,i'_3-1]|=i_2$. We show that no deletions occur within~$\boldsymbol{D}_{w,[i_1,i_2-t_{max}]}$ for~$w\in [1,d-1]$ or within~$\boldsymbol{D}_{w,[i_1+t_{max},i_2]}$ for~$w\in [2,d]$, i.e.,~$\boldsymbol{\delta}_w\cap [i'_0,i'_2]=\emptyset$ for~$w\in [1,d-1]$, and~$\boldsymbol{\delta}_w\cap [i'_1,i'_3]=\emptyset$ for~$w\in [2,d]$.

Suppose on the contrary, there are deletions within~$\boldsymbol{D}_{w,[i_1,i_2-t_{max}]}$ for~$w\in [1,d-1]$. Then there exist some~$w_1\in[1,d-1]$ and~$k_1\in[1,k]$, such that~$\delta_{w_1,k_1}\in [i'_0,i'_2]$ (recall that $\delta_{w_1,k_1}$ is the location of the $k_1$-th deletion in the  $w_1$-th read). Then we have that~$\delta_{w_1+1,k_1}=\delta_{w_1,k_1}+t_{w_1}\in [i'_0,i'_3]$. Note that there are~$k-k_1$ deletions~$\{\delta_{w_1,k_1+1},\ldots,\delta_{w_1,k}\}$ to the right of~$\delta_{w_1,k_1}$ and~$k_1-1$ deletions~$\{\delta_{w_1+1,1},\ldots,\delta_{w_1+1,k_1-1}\}$ to the left of~$\delta_{w_1+1,k_1}$. 
Hence we have that
\begin{align*}
&|(\boldsymbol{\delta}_{w_1}\cup\boldsymbol{\delta}_{w_1+1})\cap [\delta_{w_1,k_1}+1,\delta_{w_1,k_1}+t_{w_1}-1]|\\
\le &|(\boldsymbol{\delta}_{w_1}\cup\boldsymbol{\delta}_{w_1+1})\cap [\delta_{w_1,k_1}+1,\delta_{w_1+1,k_1}-1]|\\
\le& k-k_1+k_1-1\\
=&k-1,
\end{align*}
meaning that there are at most~$k-1$ deletions in the $w_1$-th or $(w_1+1)$-th heads that lie in interval~$[\delta_{w_1,k_1}+1,\delta_{w_1,k_1}+t_{w_1}-1]$. Since~$t_{w_1}\ge k(T+1)+1$, there are at least~$k$ disjoint intervals of length~$T+1$ that lie in interval~$[\delta_{w_1,k_1}+1,\delta_{w_1,k_1}+t_{w_1}-1]$.
It then follows that there exists an interval~$[i',i'+T]\subset [\delta_{w_1,k_1}+1,\delta_{w_1,k_1}+t_{w_1}-1]$ such that~$[i',i'+T]\cap (\boldsymbol{\delta}_{w_1}\cup \boldsymbol{\delta}_{w_1+1}) = \emptyset$. Let~$l'_1=|\boldsymbol{\delta}_{w_1}\cap [1,i'-1]|$ and~$l'_2=|\boldsymbol{\delta}_{w_1+1}\cap [1,i'-1]|$ be the number of deletions in heads~$w_1$ and~$w_1+1$, respectively that is to the left of~$i'$. We have that~$l'_1>l'_2$ since~$\delta_{w_1,k_1}<i'$ and~$\delta_{w_1+1,k_1}>i'+T$.
Since~$[i',i'+T]\cap (\boldsymbol{\delta}_{w_1}\cup \boldsymbol{\delta}_{w_1+1}) = \emptyset$ and $l'_1-l'_2\le k<T$, we have that
\begin{align*}
l'_1=&|\boldsymbol{\delta}_{w_1}\cap [1,i'-1]|\\
=&|\boldsymbol{\delta}_{w_1}\cap [1,i'+l'_1-l'_2-1]|\\
=&|\boldsymbol{\delta}_{w_1}\cap [1,i'+T-1]|,~\mbox{and}\\
l'_2=&|\boldsymbol{\delta}_{w_1+1}\cap [1,i'-1]|\\
=&|\boldsymbol{\delta}_{w_1+1}\cap [1,i'+T+l'_2-l'_1-1]|.
\end{align*}
Therefore,
\begin{align*}
&\boldc_{[i'+l'_1-l'_2,i'+T]}\\
=&\boldsymbol{D}_{w_1,[i'+l'_1-l'_2-|\boldsymbol{\delta}_{w_1}\cap [1,i'+l'_1-l'_2-1]|,i'+T-|\boldsymbol{\delta}_{w_1}\cap [1,i'+T-1]|]}\\
=&\boldsymbol{D}_{w_1,[i'-l'_2,i'+T-l'_1]}\\
=&\boldsymbol{D}_{w_1+1,[i'-l'_2,i'+T-l'_1]}\\
=&\boldsymbol{D}_{w_1+1,[i'-|\boldsymbol{\delta}_{w_1+1}\cap [1,i'-1]|,i'+T+l'_2-l'_1-|\boldsymbol{\delta}_{w_1+1}\cap [1,i'+T+l'_2-l'_1-1]|]}\\
=&\boldc_{[i',i'+T+l'_2-l'_1]},    
\end{align*}
which implies that~$L(\boldc_{[i',i'+T]},l'_1-l'_2)=T+1>T$. Since~$[i',i'+T]\subset [i'_0,i'_3]\subset [1,n+k+1]$, this is a contradiction to the assumption that~$L(\boldc_{[1,n+k+1]},l'_1-l'_2)\le T$. Therefore, there are no deletions within~$\boldsymbol{D}_{w,[i_1,i_2-t_{max}]}$ for~$w\in [1,d-1]$, i.e.,~$\boldsymbol{\delta}_w\cap [i'_0,i'_2]=\emptyset$ for~$w\in [1,d-1]$. Similarly, we have that~$\boldsymbol{\delta}_w\cap [i'_1,i'_3]=\emptyset$ for~$w\in [2,d]$. Since $[i'_1,i'_2]\subset[i'_0,i'_2]$ and  $[i'_1,i'_2]\subset[i'_1,i'_3]$, it follows that
\begin{align}\label{equation:emptyi1i2}
    [i'_1,i'_2]\cap \boldsymbol{\delta}_w =\emptyset
\end{align}
and hence
$\boldc_{[i'_1,i'_2]}=\boldsymbol{D}_{w,[i_1+t_{max},i_2-t_{max}]}$ 
for~$w\in [1,d]$.

Next we show that the intervals~$[1,i'_1-1]$ and~$[i'_2+1,n+R]$ are deletion isolated. Suppose on the contrary, there exists some~$w_2\in[1,d]$ for which~$(\boldsymbol{\delta}_{w_2}\cap [1,i'_1-1])+t_{w_2}\ne (\boldsymbol{\delta}_{w_2+1}\cap [1,i'_1-1])$. Then we have that $|\boldsymbol{\delta}_{w_2}\cap [1,i'_1-1]|> |\boldsymbol{\delta}_{w_2+1}\cap [1,i'_1-1]|$. Let~$x = |\boldsymbol{\delta}_{w_2}\cap [1,i'_1-1]|-|\boldsymbol{\delta}_{w_2+1}\cap [1,i'_1-1]|$, then,
\begin{align}
    &\boldc_{[i'_1,i'_2-x]}\nonumber\\
    \overset{(a)}{=}&\boldsymbol{D}_{w_2,[i_1+t_{max}+|\boldsymbol{\delta}_{w_2}\cap [1,i'_1-1]|,i_2-t_{max}-x+|\boldsymbol{\delta}_{w_2}\cap [1,i'_2-x-1]|]}\nonumber\\
    =&\boldsymbol{D}_{w_2+1,[i_1+t_{max}+|\boldsymbol{\delta}_{w_2}\cap [1,i'_1-1]|,i_2-t_{max}-x+|\boldsymbol{\delta}_{w_2}\cap [1,i'_2-x-1]|]}\nonumber\\
    =&\boldsymbol{D}_{w_2+1,[i_1+t_{max}+x+|\boldsymbol{\delta}_{w_2+1}\cap [1,i'_1-1]|,i_2-t_{max}+|\boldsymbol{\delta}_{w_2+1}\cap [1,i'_2-1]|]}\nonumber\\
    \overset{(b)}{=}&\boldc_{[i'_1+x,i'_2]},\label{equation:shift}
\end{align}
where~$(a)$ and~$(b)$ hold since we have E.q.~\eqref{equation:emptyi1i2}.
This implies that~
\begin{align*}
   L(\boldc_{i'_1,i'_2},x)=&i'_2-i'_1+1\\
   \overset{(a)}{\ge} & i_2-i_1-2t_{max}+1\\
   \ge &T+1,    
\end{align*}
where $(a)$ holds since $\boldc_{[i'_1,i'_2]}=\boldsymbol{D}_{w,[i_1+t_{max},i_2-t_{max}]}$.
 This contradicts to the fact that~$L(\boldc_{[1,n+k-1],\le k})\le T$. Therefore, the interval~$[1,i'_1-1]$ is deletion islolated. Similarly,~$[i'_2+1,n+R]$ is deletion isolated.
\end{proof}
In the following, we show that the output intervals satisfy \textbf{(P1)}, \textbf{(P2)}, and \textbf{(P3)}, respectively. Let~$[p_{2j-1},p_{2j}]$,~$j\in [1,J']$ be the marked intervals in the algorithm, where~$p_1<\ldots<p_{2J'}$. Let~$p_0=0$ and~$p_{2J'+1}=n+R+1-k$,
then the output intervals are the leftmost up to~$k$ nonempty intervals among
~$\{[p_{2j}+1,p_{2j+1}-1]\}^{J'}_{j=0}$. 
Note that from the marking operation in the \textbf{Initialization} step and \textbf{Step 2}, the interval~$[n+1-t_{max},n+R-k]$ is not marked. In addition, for any~$j\in [1,J']$, sequences~$\boldsymbol{D}_{w,[p_{2j-1},p_{2j}]}$ are equal for all~$w\in [1,d]$.
Hence, according to Lemma~\ref{lemma:identifyintervals}, there exist intervals~$[p'_{2j-1},p'_{2j}]$,~$j\in [1,J']$, where
\begin{align}\label{equation:shift}
    &p'_{j}=p_j+|\boldsymbol{\delta}_w\cap [1,p'_{j}-1]|\mbox{, and}\nonumber\\
    &[p'_{2\ell-1},p'_{2\ell}]\cap\boldsymbol{\delta}_w=\emptyset,
\end{align}
for all~$j\in [1,2J']$,~$\ell\in [1,J']$, and~$w\in [1,d]$. In addition, intervals~$[1,p'_{2j-1}-1]$ are deletion isolated\footnote{The interval~$[p_1,p_2]$ may be marked in the \textbf{Initialization} step and have length less than~$T+2t_{max}+1$. In that case, apply Lemma~\ref{lemma:identifyintervals} by considering an interval~$[-t_{max}+T+1,0]$ where~$\boldsymbol{D}_{w,[-t_{max}+T+1,0]}$ are equal for ~$w\in [1,d]$.} for~$j\in [1,J']$. It follows that~$[p'_{2j-1},p'_{2j+1}-1]$ is deletion isolated for~$j\in [1,J']$, where $p'_{2J'+1}=n+R+1$. Since~$[p'_{2j-1},p'_{2j}]\cap\boldsymbol{\delta}_w=\emptyset$ for~$j\in [1,J']$ and~$w\in [1,d]$,
then we have that the intervals~$[p'_{2j}+1,p'_{2j+1}-1]$,~$j\in [0,J']$, where~$p'_0=0$ and~$p'_{2J+1}=n+R+1$, are deletion isolated. From~\eqref{equation:shift} we have that~$\boldsymbol{D}_{w,[p_{2j}+1,p_{2j+1}-1]}=\boldc_{[p'_{2j}+1,p'_{2j+1}-1]\cap \boldsymbol{\delta}^c_w}$. In addition, the intervals~$\{[p'_{2j}+1,p'_{2j+1}-1]\}^{J'}_{j=0}$ are disjoint since
\begin{align*}
    &(p'_{2(j+1)}+1)-(p'_{2j+1}-1)\\\
    =& p_{2(j+1)}+|\boldsymbol{\delta}_w\cap [1,p'_{2(j+1)}-1]|+2 -p_{2j+1}-|\boldsymbol{\delta}_w\cap [1,p'_{2j+1}-1]|\\
    \overset{(a)}{\ge} &T +|\boldsymbol{\delta}_w\cap [1,p'_{2(j+1)}-1]|-|\boldsymbol{\delta}_w\cap [1,p'_{2j+1}-1]|\\
    \ge &T-k>0,
\end{align*}
for~$j\in [0,J'-1]$, where $(a)$ follows from the fact that marked intervals have length at least $T$. Therefore, the output intervals~$\{[p_{2j}+1,p_{2j+1}-1]\}^{J'}_{j=0}$ satisfy \textbf{(P1)}.

Next, we show that the output intervals satisfy \textbf{(P2)}. 
For any output interval~$[p_{2j}+1,p_{2j+1}-1]$ with~$[p_{2j}+1,p_{2j+1}-1]\subseteq [1,n+1-t_{max}]$, the corresponding interval~$[p'_{2j}+1,p'_{2j+1}-1]$ contains at least one deletion in~$\boldsymbol{\delta}_w$, i.e.,~$[p'_{2j}+1,p'_{2j+1}-1]\cap \boldsymbol{\delta}_w\ne \emptyset$, for some~$w\in [1,d]$. Otherwise, we have that~$[p'_{2j'}+1,p'_{2j'+1}-1]\cap \boldsymbol{\delta}_w= \emptyset$ for~$w\in [1,d]$ for some $j'$. Combining with \eqref{equation:shift} and the fact that intervals~$[1,p'_{2j-1}-1]$ are deletion isolate for~$j\in [1,J']$,
 it follows that
the sequences~$\boldsymbol{D}_{w,[p_{2j'}+1,p_{2j'+1}-1]}$ are equal for~$w\in [1,d]$. This implies that the interval~$[p_{2j'}+1,p_{2j'+1}-1]$ is marked during the procedure, which is a contradiction to the fact that~$[p_{2j'}+1,p_{2j'+1}-1]$ is not marked. 
Therefore, there are at most~$k$ unmarked intervals that lie within the interval~$[1,n+1]$. 
Note that there is one unmarked interval containing $[n+1-t_{max},n+R-k]$ that does not lie in~$[1,n+1]$. It follows that there are at most~$k+1$ unmarked intervals in total. When there are~$k+1$ unmarked intervals, the deletions~$\boldsymbol{\delta}_w$ are contained in the~$k$ output intervals since each output interval within~$[1,n+1]$ contains at least one deletion. When there are no more than~$k$ intervals, the deletions are contained in the unmarked output intervals since the marked intervals do not contain deletions. Therefore we have that~$\boldsymbol{\delta}_w\subseteq \{[p_{2j}+1,p_{2j+1}-1]\}^{J}_{j=1}$, where~$\{[p_{2j}+1,p_{2j+1}-1]\}^{J}_{j=1}$ are the output intervals and~$J\le k$. 

Finally, we show that~$|\mathcal{I}^*_j\cap [1,n+1]|\le B-k$ for~$j\in [1,J]$, which is \textbf{(P3)}. We first prove that for any unmarked index~$i\in [1,n+1-\lfloor t_{max}+(T+1)/2 \rfloor]$,
there exist some~$w\in [1,d]$ and~$k_1\in [1,k]$, such that a deletion at~$\delta_{w,k_1}$ occurs within distance~$\lfloor t_{max}+(T+1)/2 \rfloor$ to the  bit~$\boldc_{i'=i+|\boldsymbol{\delta}_w\cap [1,i'-1]|}$
that becomes~$\boldsymbol{D}_{w,i}$, i.e.,~$\delta_{w,k_1}\in [i'-\lfloor t_{max}+(T+1)/2 \rfloor,i'+\lfloor t_{max}+(T+1)/2 \rfloor]$\footnote{When~$i'-\lfloor t_{max}+(T+1)/2 \rfloor<0$, consider bits~$\boldsymbol{D}_{w,[i'-\lfloor t_{max}+(T+1)/2 \rfloor,0]}$ that are equal for~$w\in [1,d]$}. Otherwise, we have that~$[i'-\lfloor t_{max}+(T+1)/2 \rfloor,i'+\lfloor t_{max}+(T+1)/2 \rfloor]\cap \delta_w=\emptyset$ for~$w\in [1,d]$. Since ~$[i'-\lfloor  t_{max}+(T+1)/2 \rfloor,i'+\lfloor t_{max}+(T+1)/2 \rfloor]$ has length more than~$t_{w}$ for~$w\in [1,d]$, we have that $\delta_{w+1,j}=\delta_{w,j}+t_w\in [1,i'-\lfloor t_{max}+(T+1)/2 \rfloor-1]$ for every $\delta_{w,j}+t_w\in [1,i'-\lfloor t_{max}+(T+1)/2 \rfloor-1]$. It follows that
~$[1,i'-\lfloor t_{max}+(T+1)/2 \rfloor-1]$ is deletion isolated. Therefore, we have that
\begin{align*}
&\boldsymbol{D}_{w,[i-\lfloor t_{max}+(T+1)/2 \rfloor,i+\lfloor t_{max}+(T+1)/2 \rfloor]}\\
=&\boldc_{[i-\lfloor t_{max}+(T+1)/2 \rfloor+|\boldsymbol{\delta}_w\cap [i'-1]|,i+\lfloor t_{max}+(T+1)/2 \rfloor+|\boldsymbol{\delta}_w\cap [i'-1]|]}\\
=&\boldc_{[i'-\lfloor t_{max}+(T+1)/2 \rfloor,i'+\lfloor t_{max}+(T+1)/2 \rfloor]}
\end{align*}
are equal for all~$w\in [1,d]$, which means that the interval~$[i-\lfloor t_{max}+(T+1)/2 \rfloor,i+\lfloor t_{max}+(T+1)/2 \rfloor]$ and thus the index~$i$ should be marked. Therefore, every unmarked index~$i\in [1,n+1-\lfloor t_{max}+(T+1)/2 \rfloor]$ is associated with a deletion index~$\delta_{w,k_1}$ that is within distance~$\lfloor t_{max}+(T+1)/2 \rfloor$ to~$i'=i+|\boldsymbol{\delta}_w\cap [1,i'-1]|$.
On the other hand, any deletion~$\delta_{w,k_1}$ is associated with at most~$2\lfloor (2t_{max}+T+1)/2\rfloor +1$ unmarked indices. 
Therefore, the number of unmarked bits within~$[1,n+1-\lfloor t_{max}+(T+1)/2 \rfloor]$ is at most~$(2\lfloor (2t_{max}+T+1)/2\rfloor +1) kd$. The number of unmarked bits within~$[1,n+1]$ is at most~$(2\lfloor (2t_{max}+T+1)/2\rfloor +1) kd+\lfloor (2t_{max}+T+1)/2\rfloor=B-k$. 

\subsection{Determining the Number of Deletions}\label{subsection:determinenumbers}
In this subsection we present the algorithm for determining the number of deletions~$|\boldsymbol{\delta}_w\cap \mathcal{I}_j|$,~$w\in [1,d]$, for any deletion isolated interval~$\mathcal{I}_j\subseteq [1,n+k+1]$. Fix $j$. The input for this algorithm are the reads~$\boldsymbol{D}_{[1,d],\mathcal{I}^*_j}$ obtained by deleting~$\boldc_{\boldsymbol{\delta}_w\cap \mathcal{I}_j}$,~$w\in [1,d]$ from~$\boldc_{\mathcal{I}_j}$. The interval~$\mathcal{I}^*_j$ is the $j$-th output interval obtained from the procedure in Subsection~\ref{subsection:a}.
Note that~$\mathcal{I}_j$ is not known at this point.
In the algorithm only the first two reads~$\boldsymbol{D}_{[1,2],\mathcal{I}^*_j}$ are used.
Let~$\mathcal{I}_j=[b_{min},b_{max}]$ for some integers~$b_{min}$ and~$b_{max}$.
Consider the following intervals,
\begin{align*}
\mathcal{B}_{i,m} &= \begin{cases}
&[b_{min}+(i-1)t_{1}+(m-1)(T+2k+1),\min\{b_{min}+(i-1)t_{1}+m(T+2k+1)-1,b_{max}\}],\\&\text{for $i\in [1,\lceil (b_{max}-b_{min}+1)/t_1 \rceil]$
and~$m\in [1,\min\{4k+1,\lceil ((b_{max}-b_{min}+1)\bmod{t_1})/(T+2k+1) \rceil\}]$}\\
\end{cases}.
\end{align*}
Recall that here~$t_1$ is the distance between head~$1$ and head~$2$.
The intervals~$\mathcal{B}_{i,m}$ are disjoint and have length~$T+2k+1$ except when~$i=\lceil (b_{max}-b_{min}+1)/t_1 \rceil$ and~$m=\min\lceil ((b_{max}-b_{min}+1)\bmod{t_1})/(T+2k+1) \rceil$ the length might be less. Let~$\mathcal{U}_m=\cup_{i}\mathcal{B}_{i,m}$ be the union of intervals~$\mathcal{B}_{i,m}$ with the same~$m$ for~$m\in [1,4k+1]$. Then the unions~$U_m$ are disjoint since~$t_1\ge (4k+1)(T+2k+1)$. Since the deletions occur in at most~$2k$ positions in the first two heads,
at least~$2k+1$ unions~$\{\mathcal{U}_{m_1},\ldots,\mathcal{U}_{m_{2k+1}}\}$ satisfy~$\mathcal{U}_{m_{l}}\cap (\boldsymbol{\delta}_1\cup \boldsymbol{\delta}_2)=\emptyset$ for~$l\in [1,2k+1]$.

Similarly, let~$\mathcal{I}^*_j=[b'_{min},b'_{max}]$ for some integers~$b'_{min}$ and~$b'_{max}$. Define the intervals
\begin{align*}
\mathcal{B}'_{i,m} &= \begin{cases}
&[b'_{min}+(i-1)t_{1}+(m-1)(T+2k+1),\min\{b'_{min}+(i-1)t_{1}+m(T+2k+1)-k-1,b'_{max}\}],\\ 
&\text{for $i\in [1,\lceil (b'_{max}-b'_{min}+1)/t_1 \rceil]$
and~$m\in [1,\min\{4k+1,\lceil ((b'_{max}-b'_{min}+1)\bmod{t_1})/(T+2k+1) \rceil\}]$}\\
\end{cases}.
\end{align*}
Then~$\mathcal{B}_{i,m}$ are disjoint length~$T+k+1$ intervals except when~$i=\lceil (b'_{max}-b'_{min}+1)/t_1 \rceil$ and~$m=\min\{4k+1,\lceil ((b'_{max}-b'_{min}+1)\bmod{t_1})/(T+2k+1) \rceil\}$ the length might be less.
Let~
\begin{align*}
\mathcal{IM}'=&\{(i,m):|\mathcal{B}'_{i,m}| =T+k+1\}\\
\end{align*}
be the set of~$(i,m)$ pairs for which~$\mathcal{B}'_{i,m}$ has length~$T+k+1$. Since~$|\mathcal{I}^*_j|=|\mathcal{I}_j|-|\mathcal{I}_j\cap \boldsymbol{\delta}_w|$ for~$w\in [1,d]$, we have that
\begin{align*}
&b'_{max}-b'_{min}+1=|\mathcal{I}^*_j|\\
\le &|\mathcal{I}_j|=
    b_{max}-b_{min}+1
\end{align*}
It follows that~$\mathcal{B}_{i,m}\ne \emptyset$ when~$(i,m)\in \mathcal{IM}'$.
For notation convenience, let~$p_{i,m}$ and~$q_{i,m}$ be the beginning and end points of interval~$\mathcal{B}_{i,m}$, i.e.,~$\mathcal{B}_{i,m}=[p_{i,m},q_{i,m}]$ for~$(i,m)\in\mathcal{IM}'$. Similarly, let~$\mathcal{B}'_{i,m}=[p'_{i,m},q'_{i,m}]$ for $(i,m)\in\mathcal{IM}'$.


The algorithm is given as follows.
\begin{enumerate}
\item \textbf{Step 1:} For all~$(i,m)\in\mathcal{IM}'$, find a unique integer~$0\le x_{i,m}\le k$ such that~$\boldsymbol{D}_{1,[p'_{i,m},q'_{i,m}-x_{i,m}]}=\boldsymbol{D}_{2,[p'_{i,m}+x_{i,m},q'_{i,m}]}$. If no or more than one such integers exist, let~$x_{i,m}=0$. Go to Step 2.
\item \textbf{Step 2:} For all~$m\in [1,4k+1]$, compute the sum~$s_{m}=\sum_{i:(i,m)\in \mathcal{IJ}'}x_{i,m}$. Go to step 3.
\item \textbf{Step 3:} Output the majority among~$\{s_m\}^{4k+1}_{m=1}$.
\end{enumerate}
Note that the set~$\mathcal{IM}'$ and the intervals~$\mathcal{B}'_{i,m}=[p'_{i,m},q'_{i,m}]$ can be determined from Lemma \ref{lemma:synchronization} and the definition of $\cB'_{i,m}$.
We now show that the algorithm outputs~$|\mathcal{I}_j\cap \boldsymbol{\delta}_1|$. It suffices to show that~$s_{m_l}=|\mathcal{I}_j\cap\boldsymbol{\delta}_1|$ for~$l\in [1,2k+1]$.    
First, we show that the unique integer~$x_{i,m_l}$ satisfying~$\boldsymbol{D}_{1,[p'_{i,m_l},q'_{i,m_l}-x_{i,m_l}]}=\boldsymbol{D}_{2,[p'_{i,m_l}+x_{i,m_l},q'_{i,m_l}]}$ exists for~$l\in [1,2k+1]$ and~$i$ such that~$(i,m_l)\in\mathcal{IM}'$. Moreover, the integer~$x_{i,m_l}$ equals~$|\boldsymbol{\delta}_1\cap [p_{1,1},p_{i,m_l}-1]|-|\boldsymbol{\delta}_2\cap [p_{1,1},p_{i,m_l}-1]|$, the difference between the number of deletions in the first two heads that happen before the interval~$\mathcal{B}_{i,m_l}$.
Recall that~$m_l$ satisfies~$\mathcal{U}_{m_l}\cap \boldsymbol{\delta}_w=\emptyset$ for~$w\in \{1,2\}$ and that 
$\boldsymbol{D}_{w,p'_{1,1}}=\boldsymbol{D}_{w,b'_{min}}$ comes from $\boldc_{b_{min}}=\boldc_{p_{1,1}}$ after deletions for $w\in\{1,2\}$. Hence, the bit $\boldsymbol{D}_{w,p'_{i,m_l}}$ comes from $\boldc_{p_{i,m_l}+|\boldsymbol{\delta}_w\cap[p_{1,1,},p_{i,m_l}-1]|}$ after deletions for $w\in\{1,2\}$, by definitions of $p_{i,m}$ and $p'_{i,m}$. In addition, $\boldsymbol{D}_{w,[p'_{i,m_l},q'_{i,m_l}]}$ comes from $\boldc_{[p_{i,m_l}+|\boldsymbol{\delta}_w\cap[p_{1,1,},p_{i,m_l}-1]|,p_{i,m_l}+|\boldsymbol{\delta}_w\cap[p_{1,1,},p_{i,m_l}-1]|+T+k]}$. 
Let $x=|\boldsymbol{\delta}_1\cap [p_{1,1},p_{i,m_l}-1]|-|\boldsymbol{\delta}_2\cap [p_{1,1},p_{i,m_l}-1]|$, we have that
\begin{align}\label{equation:shiftx}
  &\boldsymbol{D}_{1,[p'_{i,m_l},q'_{i,m_l}-x]}\nonumber\\
  =&\boldc_{[p_{i,m_l}+|\boldsymbol{\delta}_1\cap [p_{1,1},p_{i,m_l}-1]|,p_{i,m_l}+|\boldsymbol{\delta}_1\cap [p_{1,1},p_{i,m_l}-1]|+T+k-x]}\nonumber\\ 
  =& \boldsymbol{D}_{2,[p'_{i,m_l}+x,q'_{i,m_l}]}.
\end{align}
Therefore, the integer~$x_{i,m_l}=x$ satisfies~$\boldsymbol{D}_{1,[p'_{i,m_l},q'_{i,m_l}-x_{i,m_l}]}=\boldsymbol{D}_{2,[p'_{i,m_l}+x_{i,m},q'_{i,m_l}]}$. We show this~$x_{i,m_l}$ is unique. Suppose there exists another integer~$y> x$ for which~$\boldsymbol{D}_{1,[p'_{i,m_l},q'_{i,m_l}-y]}=\boldsymbol{D}_{2,[p'_{i,m_l}+y,q'_{i,m_l}]}$. Then we have that
\begin{align*}
  &\boldsymbol{D}_{1,[p'_{i,m_l},q'_{i,m_l}-y]} \\
  =&\boldsymbol{D}_{2,[p'_{i,m_l}+y,q'_{i,m_l}]}\\
  \overset{(a)}{=}&\boldsymbol{D}_{1,[p'_{i,m_l}+y-x,q'_{i,m_l}-x]}\\
  =&\boldc_{[p_{i,m_l}+|\boldsymbol{\delta}_1\cap [p_{1,1},p_{i,m_l}-1]|+y-x,p_{i,m_l}+|\boldsymbol{\delta}_1\cap [p_{1,1},p_{i,m_l}-1]|+T+k-x]}, 
\end{align*}
where~$(a)$ follows from~Eq.~\eqref{equation:shiftx}. 
Since, 
\begin{align*}
  \boldsymbol{D}_{1,[p'_{i,m_l},q'_{i,m_l}-y]}=  \boldc_{[p_{i,m_l}+|\boldsymbol{\delta}_1\cap [p_{1,1},p_{i,m_l}-1]|,p_{i,m_l}+|\boldsymbol{\delta}_1\cap [p_{1,1},p_{i,m_l}-1]|+T+k-y]},
\end{align*}
it follows that
\begin{align*}
  \boldc_{[p_{i,m_l}+|\boldsymbol{\delta}_1\cap [p_{1,1},p_{i,m_l}-1]|+y-x,p_{i,m_l}+|\boldsymbol{\delta}_1\cap [p_{1,1},p_{i,m_l}-1]|+T+k-x]} = \boldc_{[p_{i,m_l}+|\boldsymbol{\delta}_1\cap [p_{1,1},p_{i,m_l}-1]|,p_{i,m_l}+|\boldsymbol{\delta}_1\cap [p_{1,1},p_{i,m_l}-1]|+T+k-y]}.
\end{align*}
It then follows that
\begin{align*}
    &L(\boldc_{[p_{i,m_l}+|\boldsymbol{\delta}_1\cap [p_{1,1},p_{i,m_l}-1]|,p_{i,m_l}+|\boldsymbol{\delta}_1\cap [p_{1,1},p_{i,m_l}-1]|+T+k-x]},y-x)\\
    = &T+k-x+1 \ge T +1,
\end{align*}
which is a contradiction to the fact that~$L(\boldc,\le k)\le T$.
Similarly, such contradiction occurs when~$y<x$. Hence such~$x_{i,m_l}$ is unique.

Next, we show that~$s_{m_l}=|\boldsymbol{\delta}_1\cap \mathcal{I}_j|$ for~$l\in [1,2k+1]$. Since~$p_{i,m_l}-p_{i-1,m_l}=t_1$ for~$i\in [2,\max_{(i,m_l)\in\mathcal{IM}'}i]$, we have that
\begin{align*}
    &|\boldsymbol{\delta}_1\cap [p_{1,1},p_{i,m_l}-1]|\\ =&|\boldsymbol{\delta}_1\cap [p_{1,1},p_{1,m_l}-1]|+ \sum^{i-1}_{w=1}|\boldsymbol{\delta}_1\cap [p_{w,m_l},p_{w+1,m_l}-1]|\\
    \overset{(a)}{=}&|\boldsymbol{\delta}_2\cap [p_{2,1},p_{2,m_l}-1]|+ \sum^{i-2}_{w=1}|\boldsymbol{\delta}_2\cap [p_{w+1,m_l},p_{w+2,m_l}-1]|+|\boldsymbol{\delta}_1\cap [p_{i-1,m_l},p_{i,m_l}-1]|\\
    =&|\boldsymbol{\delta}_2\cap [p_{2,1},p_{i,m_l}-1]|+|\boldsymbol{\delta}_1\cap [p_{i-1,m_l},p_{i,m_l}-1]|\\
    \overset{(b)}{=}&|\boldsymbol{\delta}_2\cap [p_{1,1},p_{i,m_l}-1]|+|\boldsymbol{\delta}_1\cap [p_{i-1,m_l},p_{i,m_l}-1]|,
\end{align*}
where~$(a)$ hold since~$|\boldsymbol{\delta}_1\cap [p_{1,1},p_{1,m_l}-1]|=|\boldsymbol{\delta}_2\cap [p_{2,1},p_{2,m_l}-1]|$ and $|\boldsymbol{\delta}_1\cap [p_{w-1,m_l},p_{w,m_l}-1]|=|\boldsymbol{\delta}_2\cap [p_{w,m_l},p_{w+1,m_l}-1]|$ for~$w\in [2,i-1]$. Equality $(b)$ holds since $\cI_j$ is deletion isolated and hence $\boldsymbol{\delta}_2\cap[p_{1,1},p_{2,1}-1]=\emptyset$. 
It then follows that~$x_{i,m_l}=|\boldsymbol{\delta}_1\cap [p_{i-1,m_l},p_{i,m_l}-1]|$~($p_{0,m_l}=p_{1,1}$) and that
\begin{align*}
    s_{m_l}= |\boldsymbol{\delta}_1\cap [p_{1,1},p_{\max_{(i,m_l)\in\mathcal{IM}'}i,m_l}-1]|
\end{align*}
Note that~$\boldsymbol{\delta}_1\cap [p_{\max_{(i,m)\in\mathcal{IM}'}i,m_l},b_{max}]\subseteq \boldsymbol{\delta}_1\cap [b_{max}-t_1+1,b_{max}]$. Since~$\boldsymbol{\delta}_1\cap [b_{max}-t_1+1,b_{max}]=\emptyset$ because $\cI_{j}$ is deletion isolated, we have that~$s_{m_l}=|\boldsymbol{\delta}_1\cap \mathcal{I}_j|$. Then the majority rule works.

\section{Correcting~$k\ge 2d$ deletions}\label{section:greaterthan2d}
In this section we present the code for correcting~$k\ge 2d$ deletions as well as a lower bound on the redundancy when~$t_i=o(n)$. 
The code construction is similar to the one presented in Section~\ref{section:lessthan2mdeletions}. We use Lemma~\ref{lemma:synchronization} to identify the location of deletions within a set of disjoint intervals~$\mathcal{I}_j$, each with length no more than~$B$. Note that in order to apply Lemma~\ref{lemma:synchronization}, the sequence~$\boldc\in\{0,1\}^{n}$ has to be transformed into a sequence~$F(\boldc)\in\{0,1\}^{n+k+1}$ (see Lemma~\ref{lemma:periodfree}) that satisfies~$L(F(\boldc),\le k)\le T$.
Then we use a concatenated code construction. Specifically, to protect a sequence~$\boldc\in \{0,1\}^{n+k+1}$ from~$k$ deletions, we split~$\boldc$ into blocks ~$\bolda_i$,~$i\in [1,\lceil (n+k+1)/B \rceil]$ of length~$B$ as in~Eq.~\eqref{equation:block}. Then the function~$S$ defined in~Eq.~\eqref{equation:split}, which is a concatenation of hashes~$Hash$ (see Lemma~\ref{lemma:hash}) of~$\bolda_i$,~$i\in [1,\lceil (n+k+1)/B \rceil]$, can be used to corret~$k$ deletions in~$\boldc$ (see Lemma~\ref{lemma:splithash}). Finally, a Reed-Solomon code is used to protect the~$S$ hashes. The encoding function is as follows
\begin{align}\label{equation:encoding2d}
    Enc_2(\boldc)=(F(\boldc),R^{'}_2(\boldc),R^{''}_2(\boldc))
\end{align}
where~
\begin{align}\label{equation:encoding2d1}
R^{'}_2(\boldc)&=RS_{2\lfloor k/d \rfloor}(S(F(\boldc))),\nonumber\\
R^{''}_2(\boldc)&=Rep_{k+1}(Hash(R^{'}_2(\boldc))),
\end{align}
function $S(\cdot)$ is defined in \eqref{equation:split}, and $RS_{2\lfloor k/d \rfloor}$ is the systematic Reed-Solomon code given in Lemma \ref{lemma:rs}.
The length of~$R^{'}_2(\boldc)$ is~$N_1=2\lfloor k/d \rfloor\max\{\log (n+k+1),(4k\log B+o(\log B))\}=2\lfloor k/d \rfloor\log n +o(\log n)$. The length of~$R^{''}_2(\boldc)$ is~$N_2=4k(k+1)\log N_1+O(\log N_1)=o(\log n)$. The length of~$Enc_2(\boldc)$ is~$N=n+k+1+N_1+N_2=n+2\lfloor k/d \rfloor\log n+o(\log n)$.
\begin{theorem}\label{theorem:2d}
The set~$\mathcal{C}_2=\{Enc_2(\boldc):\boldc\in\{0,1\}^n\}$ is a $d$-head~$k$-deletion correcting code for~$2d\le k$, if the distance between any two consecutive heads satisfies~$t_i\ge\max\{(3k+\lceil \log n \rceil+2)[k(k-1)/2+1]+(7k-k^3)/6,(4k+1)(5k+\lceil \log n \rceil+3)\}$ for $i\in\{1,\ldots,d-1\}$. The code~$\mathcal{C}_2$ can be constructed, encoded, and decoded in~$n^{2k+2}$ time. The redundancy of~$\mathcal{C}_2$ is~$N-n=+2\lfloor k/d \rfloor\log n+o(\log n)$.
\end{theorem}
\begin{proof}
The proof is essentially the same as the proof of Theorem~\ref{theorem:lessthan2d}.
For any~$\boldsymbol{D}\in \cD_k(\boldc)$, let~$\boldd=\boldsymbol{D}_{1,[1,N-k]}$ be the first row of~$\boldsymbol{D}$. The sequence~$\boldd$ is a length~$N-k$ subsequence of~$Enc_2(\boldc)$.
Then it is possible to recover~$Hash(R^{'}_2(\boldc))$ from the last $N_2-k$ bits of $d$, which is a length $N_2-k$ subsequence of~$R^{''}_2(\boldc)$. Then, we can recover $Hash(R'_2(\boldc))$, and  recover~$R^{'}_2(\boldc)$ from $\boldd_{[n+1,n+N_1-k]}$.

It suffices to show how to use~$R'(\boldc)$ to recover~$F(\boldc)$.
According to Lemma~\ref{lemma:synchronization}, 
we can identify a set of~$J\le k$ deletion isolated intervals~$\{\mathcal{I}_j\}^J_{j=1}$, each with length not greater than~$B$, such that~$\boldsymbol{\delta}_1\subseteq(\cup^J_{j=1}\mathcal{I}_j)$. Note that according to Lemma \ref{lemma:kdeletionkheads}, the bits $\boldc_{\cI_j}$ with $|\boldsymbol{\delta}_w\cap\cI_j|\le d-1$ errors can be recovered, when $t_i\ge \max\{(3k+\lceil \log n \rceil+2)[k(k-1)/2+1]+(7k-k^3)/6,(4k+1)(5k+\lceil \log n \rceil+3)\}$. Note that
each interval~$\mathcal{I}_j$ with $|\boldsymbol{\delta}_w\cap\cI_j|\ge d$ spans over at most two blocks~$\bolda_i$. Therefore, at most~$2\lfloor k/d \rfloor$ blocks, the indices of which can be identified, contain at least $d$ deletions. Hence the sequence~$S(F(\boldc))$ can be recovered with at most~$2\lfloor k/d \rfloor$ symbol errors, with known error locations.
With the help of
the Reed-Solomon code redundancy
~$RS_{2\lfloor k/d \rfloor}(S(F(\boldc)))$,
 the sequence~$S(F(\boldc))$
 can be recovered. Then from Lemma~\ref{lemma:splithash} and Lemma~\ref{lemma:periodfree} the sequence~$F(\boldc)$ and thus~$\boldc$ can be recovered. 
The computation complexity of~$Enc_2(\boldc)$ has the same order as that of~$Enc_1(\boldc)$. It takes~$O(n^{2k+2})$ time to construct~, encode, and decode~$Enc_2(\boldc)$.
\end{proof}
Now we present a lower bound on the redundancy for small head distances~$t_i=o(n)$,~$i\in [1,d-1]$, which proves the last part of Theorem \ref{theorem:main}.
\begin{theorem}\label{theorem:lowerbound}
Let~$\mathcal{C}$ be a~$d$-head~$k$-deletion code with length~$n$. If the distance~$t_i$ satisfies~$t_i=n^{o(1)}$ for~$i\in [1,d-1]$, then we have that~$ |\mathcal{C}|\le 2^{\lfloor k/2d\rfloor \log n+o(\log n)}$. 
\end{theorem}
\begin{proof}
Let~$T_{sum}=\sum^{d-1}_{i=1}t_i$.
Sample the sequence~$\boldc$ with period~$T_{sum}$,
\begin{align*}
\boldc'=(c_{1+T_{sum}},c_{1+3T_{sum}},\ldots,c_{1+(2j+1)T_{sum}},\ldots,c_{1+(2\lfloor (n-1-T_{sum})/2T_{sum} \rfloor-1)T_{sum}})
\end{align*}
We show that correcting~$k$ deletions in~$\boldc$ is at least as hard as correcting~$\lfloor k/d\rfloor$ erasures in~$\boldc'$. It suffices to show that~$d$ deletions in heads~$i\in [1,d]$ can erase the information of any bit in~$\boldc'$.
For~$j\in [1,\lfloor (n-1-T_{sum})/2T_{sum} \rfloor]$, let   $d$ deletions occur at positions
\begin{align*}
\{1+(2j-1)T_{sum}-\sum^w_{i=1} t_i:w\in[0,d-1]\},
\end{align*}
at head~$1$. Then the corresponding~$d$ deletion in head~$m$ occur at positions
\begin{align*}
\{1+(2j-1)T_{sum}-\sum^w_{i=1} t_i+\sum^{m-1}_{i=1} t_i:w\in[0,d-1]\}
\end{align*}
for $m\in[1,d]$. 
It follows that the bit~$c_{1+(2j-1)T_{sum}}$ is deleted in all heads. Suppose a genie tells the locations and values of all the~$d$ deleted bits in each head except the value of the bit~$c_{1+(2j-1)T_{sum}}$. Then this reduces to a erasure of the bit~$c_{1+(2j-1)T_{sum}}$ in~$\boldc'$. Note that in this way,~$k$ deletions in~$\boldc$ can cause~$\lfloor k/d\rfloor$ erasures in~$\boldc'$. From the Hamming bound, the size~$|\mathcal{C}|$ is upper bounded by
\begin{align*}
    |\mathcal{C}|\le &2^n/(\sum^{\lfloor k/2d \rfloor }_{i=1}\binom{\lfloor (n-1-T_{sum})/2T_{sum} \rfloor}{i})\\
    =&2^{n-\lfloor k/2d \rfloor(\log n-\log (2T_{sum})) +o(\log n)}\\
    =&2^{n-\lfloor k/2d \rfloor \log n +o(\log n)}.
\end{align*}
\end{proof}
According to Theorem~\ref{theorem:lowerbound}, the redundancy of a~$d$-head~$k$-deletion code is lower bounded by~$\lfloor k/2d \rfloor \log n +o(\log n)$.
\section{Correcting $k$ Deletions and Insertions}\label{section:correctediterrors}
In this section we show how to correct a combination of up to $k$ deletions and insertions in the $d$-head racetrack memory. In this scenario, more challenges arise since there may not be "shifts" between different reads, as we observed in Lemma \ref{lemma:identifyintervals}, after a combination of deletions and insertions. This makes detection of errors harder. Moreover, Lemma \ref{lemma:kdeletionkheads} does not apply.

The encoding and decoding algorithms for this task can be regarded as a generalization of the algorithms for correcting $k$ deletions. Similar to the 
idea in Section \ref{section:lessthan2mdeletions} and Section \ref{section:greaterthan2d}, 
we notice that the location of errors $(\boldsymbol{\delta}_i,\boldsymbol{\gamma}_i)$, $i\in[1,d]$ are contained in a set of disjoint edit isolated intervals (the definition of edit isolated intervals will be given later), each with bounded length. Yet, different from the cases in Section \ref{section:lessthan2mdeletions} and Section \ref{section:greaterthan2d}, some of the edit isolated intervals cannot be detected and identified from the reads. Fortunately, the intervals that cannot be detected contain at least $2d$ errors in each read. In addition, the "shift" in bits outside the edit isolated intervals, caused by the errors in those edit isolated intervals, can be determined in a similar manner to the one in Section \ref{subsection:determinenumbers}. Therefore, the bits outside the edit isolated intervals can be recovered similarly to the method in Section \ref{section:lessthan2mdeletions} and Section \ref{section:greaterthan2d}. 
In addition, we will provide a result similar to Lemma \ref{lemma:kdeletionkheads} (correcting deletion errors), for correcting both deletions and insertions. Specifically, We will show that the intervals with less than $d$ errors can be recovered using the reads. Then, by using Reed-Solomon codes to protect the deletion correcting hashes as we did in Section \ref{section:lessthan2mdeletions} and Section \ref{section:greaterthan2d}, the $2\lfloor k/d\rfloor \log n+o(\log n)$ redundancy can be achieved. We note that in this section, we let the head distances $t_i=t$ to be equal for $i\in [1,d-1]$. In the following, we provide the definition of edit isolated intervals.
\begin{definition}
Let~$\boldsymbol{\delta}_i=\{\delta_{i,1},\ldots,\delta_{i,r}\}$ and $\boldsymbol{\gamma}_i=\{\gamma_{i,1},\ldots,\gamma_{i,s}\}$ be the sets of deletion and insertion locations, respectively, in the~$i$-th head of a~$d$-head racetrack memory, i.e. $\boldsymbol{\delta}_{i+1}=\boldsymbol{\delta}_{i}+t_i$ and $\boldsymbol{\gamma}_{i+1}=\boldsymbol{\gamma}_{i}+t_i$, for~$i\in [1,d-1]$. 
An interval~$\mathcal{I}$ is \emph{edit isolated} if
\begin{align*}
\boldsymbol{\delta}_{i+1}\cap \mathcal{I}=& t_i + \boldsymbol{\delta}_{i}\cap \mathcal{I},\text{ and}\\
\boldsymbol{\gamma}_{i+1}\cap \mathcal{I}=& t_i + \boldsymbol{\gamma}_{i}\cap \mathcal{I}.
\end{align*}
for~$i\in [1,d-1]$. 
\end{definition}

We begin with the the algorithm for identifying a set of intervals $[b_{1j},b_{2j}]$, $j\in [1,J]$, such that for each $j\in[1,J]$, there is an interval $[p_{1j},p_{2j}]$ satisfying:
\begin{enumerate}
    \item[\textbf{(A)}] $[p_{1j},p_{2j}]\subseteq [b_{1j},b_{2j}]$
    \item[\textbf{(B)}] $\boldsymbol{E}_{w,i}=\boldsymbol{E}_{w',i}$ for any $w,w'\in[1,d]$ and $i\in ([b_{1j},p_{1j}-1]\cup[p_{2j}+1,b_{2j}])$
    \item[\textbf{(C)}] $\boldsymbol{E}_{[1,d],[p_{1j},p_{2j}]}\in\cE_{k'}(\boldc_{\cI_j})$ for some edit isolated interval $\cI_j$ and $k'\ge 1$.
    \item[\textbf{(D)}] $|[b_{1j},b_{2j}]|\le (2kdt+2t+1)(k+1)+kdt+2k$ for $j\in [1,J]$.
    \item[\textbf{(E)}] $\boldsymbol{E}_{w,i}=\boldsymbol{E}_{w',i}$ for any $w,w'\in[1,d]$ and $i\in [1,n+1]\backslash(\cup^J_{j=1}[b_{1j},b_{2j}])$.
\end{enumerate}
The algorithm is similar to the one in Section \ref{subsection:a}.
However, different from the intervals $\cI^*_j$, $j\in [1,J]$ generated in Section \ref{subsection:a}, which satisfy properties \textbf{(P1)} and \textbf{(P2)} in Section \ref{section:synchronization},
here we do not necessarily have an edit isolated interval $\mathcal{I}'_j$ satisfying $\boldsymbol{E}_{[1,d],[b_{1j},b_{2j}]}\in \cE_{k'}(\boldc_{\cI'_j})$ for every $j\in [1,J]$. Also, the error locations $(\boldsymbol{\gamma}_w\cup\boldsymbol{\delta}_w)$, $w\in [1,d]$ may not be contained in the collection of intervals $\cup^J_{j=1}\cI_j$. Given a read matrix $\boldsymbol{E}\in\cE_k(\boldc)$, where $\boldc\in\{0,1\}^{n+k+1}$ is a binary input. The algorithm is given as follows.
\begin{enumerate}
\item \textbf{Initialization:}
Set all integers~$m\in [1,n']$ unmarked, where $n'$ is the number of columns in $\boldsymbol{E}$.
Let~$i=1$. Find the largest positive integer~$L$ such that the sequences~$\boldsymbol{E}_{w,[i,i+L-1]}=\boldsymbol{E}_{w',[i,i+L-1]}$ for any~$w,w'\in [1,d]$. If such~$L$ exists and satisfies $L>kdt+t$, mark the integers~$m\in[1,L-(kdt+t)]$ and go to Step~$1$. Otherwise, go to Step~$1$.
\item \textbf{Step 1:} Find the largest positive integer~$L$ such that the sequences~$\boldsymbol{E}_{w,[i,i+L-1]}=\boldsymbol{E}_{w',[i,i+L-1]}$ for any~$w,w'\in [1,d]$. Go to Step~$2$. If no such~$L$ is found, set~$L=0$ and go to Step $2$. 
\item \textbf{Step 2:} If~$L\ge 2(kdt+t)+1$, mark the integers $m\in[i+kdt+t,\min\{i+L-1,n'\}-(kdt+t)]$. Set~$i=i+L+1$ and go to Step $3$. Else~$i=i+1$ and go to Step $3$.
\item \textbf{Step 3:} If~$i\le n'$, go to Step 1. Else go to Step 4.
\item \textbf{Step 4:} Output all unmarked intervals.
\end{enumerate}
We now show that the output intervals  satisfy the properties \textbf{(A)}, \textbf{(B)}, \textbf{(C)},  \textbf{(D)}, and \textbf{(E)} above.
\begin{lemma}\label{lemma:properties}
    For a read matrix $\boldsymbol{E}\in\cE_k(\boldc)\in\{0,1\}^{d\times n'}$, 
    Let $[b_{1j},b_{2j}]$, $j\in[1,J]$ be the output intervals in the above procedure such that $b_{11}<b_{12}<\ldots<b_{1J}$.
    There exists a set of intervals $[p_{1j},p_{2j}]$, $j\in[1,J]$, satisfying \textbf{(A)}, \textbf{(B)}, \textbf{(C)},  \textbf{(D)}, and \textbf{(E)} above. 
\end{lemma}
\begin{proof}
Note that for each interval $[b_{1j},b_{2j}]$, we have $\boldsymbol{E}_{w,[b_{1j},b_{1j}+kdt+t-1]}=\boldsymbol{E}_{w',[b_{1j},b_{1j}+kdt+t-1]}$ and 
$\boldsymbol{E}_{w,[b_{2j}-kdt-t+1,b_{2j}]}=\boldsymbol{E}_{w',[b_{2j}-kdt-t+1,b_{2j}]}$
for any~$w,w'\in [1,d]$, except for $j=1$, $\boldsymbol{E}_{w,[b_{1j},b_{1j}+kdt+t-1]}$ may not be equal to $\boldsymbol{E}_{w',[b_{1j},b_{1j}+kdt+t-1]}$, in which case, we let that $p_{11}=1$ and the following arguments hold. 
Consider the set of intervals $[b_{1j}+(i-1)t,b_{1j}+it-1]$ for $i\in [1,kd+1]$. Note that an error occurs in at most $d$ intervals, each in one of the $d$ heads. Therefore, at most $kd$ intervals contain errors. Then, there exists an interval $[b_{1j}+(i_1-1)t,b_{1j}+i_1t-1]$ for some $i_1\in [1,kd+1]$ such that 
$[b_{1j}+(i_1-1)t,b_{1j}+i_1t-1]\cap(\boldsymbol{\gamma}_w\cup\boldsymbol{\delta}_w)=\emptyset$ for $w\in [1,d]$. Similarly, there exists an interval $[b_{2j}-i_2t+1,b_{2j}-(i_2-1)t]$ for some $i_2\in [1,kd+1]$, such that $[b_{2j}-i_2t+1,b_{2j}-(i_2-1)t]\cap(\boldsymbol{\gamma}_w\cup\boldsymbol{\delta}_w)=\emptyset$ for $w\in[1,d]$.
This implies that $[b_{1j}+i_1t-1-k,b_{2j}-i_2t+1+k]$ is an edit isolated interval.
Let $\boldsymbol{E}_{[1,d],[p_{1j},p_{2j}]}\in \cE_{k'_j}(\boldc_{[b_{1j}+i_1t-1-k,b_{2j}-i_2t+1+k]})$, where $k'_j=|[b_{1j}+i_1t-1-k,b_{2j}-i_2t+1+k]\cap\boldsymbol{\delta}_1|+|[b_{1j}+i_1t-1-k,b_{2j}-i_2t+1+k]\cap\boldsymbol{\gamma}_1|$, be the read matrix obtained from $\boldc_{[b_{1j}+i_1t-1-k,b_{2j}-i_2t+1+k]}$ after deletion errors at locations $\boldsymbol{\delta}_w\cap[b_{1j}+i_1t-1-k,b_{2j}-i_2t+1+k]$ and insertion errors at locations $\boldsymbol{\gamma}_w\cap[b_{1j}+i_1t-1-k,b_{2j}-i_2t+1+k]$, $w\in[1,d]$.
Then we have that $p_{1j}\in [b_{1j}+t-1-2k,b_{1j}+kdt+t-1]$ and $p_{2j}\in  [b_{2j}-kdt-t+1,b_{2j}-t+1+2k]$. Therefore, the intervals $[p_{1j},p_{2j}]$, $j\in[1,J]$ satisfy \textbf{(A)}, \textbf{(B)}.
To show that $[p_{1j},p_{2j}]$, $j\in[1,J]$ satisfy
\textbf{(C)}, we need to show $k'_j\ge 1$ for each $j$. Suppose on the contrary, $k'_j=0$. Then since $\boldsymbol{E}_{[1,d],[p_{1j},p_{2j}]}\in \cE_{k'_j}(\boldc_{[b_{1j}+i_1t-1-k,b_{2j}-i_2t+1+k]})$, we have that $\boldsymbol{E}_{w,[p_{1j},p_{2j}]}=\boldsymbol{E}_{w',[p_{1j},p_{2j}]}$ for any $w,w'\in [1,d]$. Then we have $\boldsymbol{E}_{w,[b_{1j},b_{2j}]}=\boldsymbol{E}_{w',[b_{1j},b_{2j}]}$ for any $w,w'\in [1,d]$, and $b_{1j}+kdt+t$ should have been marked, a contradiction to the fact that $[b_{1j},b_{2j}]$ is an unmarked interval.

Next, we show that $|[b_{1j},b_{2j}]|< (2kdt+2t+1)(k+1)+kdt+2k$.
Note that an error that occurs at location $i$ in the first head also occurs at $i+(w-1)t$ in the $w$-th head. These locations are contained in an interval $[i,i+(d-1)t]$ of length less than $dt$. The locations of $k$ errors in $d$ heads are contained in $k$ intervals, each of length at most $dt$. If $|[b_{1j},b_{2j}]|\ge (2kdt+2t+1)(k+1)+kdt+2k$, there exists a sub-interval $[b'_{1j},b'_{2j}]\subseteq [b_{1j}+k,b_{2j}-k]$ with length at least $2kdt+2t+1$, that is disjoint with the $k$ intervals that contain locations of all errors in all heads. Therefore, $[b'_{1j},b'_{2j}]\cap(\boldsymbol{\delta}_w\cup\boldsymbol{\gamma}_w)=\emptyset$ for $w\in[1,d]$. 
Since the interval $[b'_{1j},b'_{2j}]$ has length more than $t$, the intervals $[1,b'_{1j}-1]$ and $[b'_{2j}+1,n+k+1]$ are edit isolated, where $n+k+1$ is the length of $\boldc$. Moreover,
$\boldsymbol{E}_{w,i}=\boldsymbol{E}_{w',i}$ for any $w,w'\in[1,d]$ and $i\in [b'_{1j}-|\boldsymbol{\delta}_1\cap[1,b'_{1j}-1]|+|\boldsymbol{\gamma}_1\cap[1,b'_{1j}-1]|,b'_{2j}-|\boldsymbol{\delta}_1\cap[1,b'_{1j}-1]|+|\boldsymbol{\gamma}_1\cap[1,b'_{1j}-1]|]$. This implies that $i=b'_{1j}-|\boldsymbol{\delta}_1\cap[1,b'_{1j}-1]|+|\boldsymbol{\gamma}_1\cap[1,b'_{1j}-1]+kdt+t$ should be marked, contradicting to the fact that $[b'_{1j},b'_{2j}]$ is unmarked, $j\in [1,J]$. Therefore, we proved \textbf{(D)}. Finally, for marked indices $i$, we have that $\boldsymbol{E}_{w,i}=\boldsymbol{E}_{w',i}$ for any $w,w'\in[1,d]$. Theorefore, we have \textbf{(E)}. 
\end{proof}
In the remaining of this section, we first show how to determine the shifts caused by errors in the edit isolated intervals that can be detected. This provides a way to correct most of the bits in $\boldc$. Then, we show how to correct $k<d$ deletions and insertions in total, and show that when $k\ge d$ and the the errors are not corrected, there is a constrait on the number of errors that occur. Finally, we present our encoding and decoding algorithms for the general cases when $k\ge d$. The code  is the same as the construction in Section \ref{section:greaterthan2d}, but with a different decoding algorithm. 
Before dealing with the $k<d$ case, we present a proposition that is repeatedly used in this section.
\begin{proposition}\label{proposition:basicshifts}
Let $\boldsymbol{E}\in \cE_k(\boldc)$ be a read matrix for some sequence $\boldc$ satisfying $L(\boldc,\le k)\le T$. 
For any integers $i\in [1,n]$ and $w,w'\in[1,d]$ such that no error occurs in interval $[i-T-2k,i]$ in the $w$-th and $w'$-th head, i.e., 
\begin{align}\label{equation:noerrornearistar}
&(\boldsymbol{\delta}_w\cup \boldsymbol{\gamma}_w)\cap [i-T-2k,i]= \emptyset,\text{ and}\nonumber\\
    &(\boldsymbol{\delta}_{w'}\cup \boldsymbol{\gamma}_{w'})\cap [i-T-2k,i]= \emptyset,
\end{align}
If 
\begin{align}\label{equation:shiftinread}
    \boldsymbol{E}_{w,[i-T-2k,i-x]}=\boldsymbol{E}_{w',[i-T-2k+x,i]}
\end{align}
for some integer $x\in[0,k]$, then
\begin{align}\label{equation:shiftediterrors}
|\boldsymbol{\gamma}_w\cap[1,i-T-2k-1]|-|\boldsymbol{\delta}_w\cap [1,i-T-2k-1]|+x=|\boldsymbol{\gamma}_{w'}\cap [1,i-T-2k-1]|- |\boldsymbol{\delta}_{w'}\cap [1,i-T-2k-1]|    
\end{align}
\end{proposition}
\begin{proof}
Suppose on the contrary,
\begin{align}\label{equation:contrary}
|\boldsymbol{\gamma}_w\cap[1,i-T-2k-1]|-|\boldsymbol{\delta}_w\cap [1,i-T-2k-1]|+x'=|\boldsymbol{\gamma}_{w'}\cap [1,i-T-2k-1]|- |\boldsymbol{\delta}_{w'}\cap [1,i-T-2k-1]|    
\end{align}
for some $x'\ne x$. If $x'>x$, then we have that
\begin{align*}
    &c_{[i-T-k+x'-x,i-k]}\\
    \overset{(a)}{=}&\boldsymbol{E}_{w,[i-T-k+x'-x+|\boldsymbol{\gamma}_w\cap[1,i-T-2k-1]|-|\boldsymbol{\delta}_w\cap [1,i-T-2k-1]|,i-k+|\boldsymbol{\gamma}_w\cap[1,i-T-2k-1]|-|\boldsymbol{\delta}_w\cap [1,i-T-2k-1]|]}\\
    \overset{(b)}{=}&\boldsymbol{E}_{w',[i-T-k+x'+|\boldsymbol{\gamma}_w\cap[1,i-T-2k-1]|-|\boldsymbol{\delta}_w\cap [1,i-T-2k-1]|,i-k+|\boldsymbol{\gamma}_w\cap[1,i-T-2k-1]|-|\boldsymbol{\delta}_w\cap [1,i-T-2k-1]|+x]}\\
    \overset{(c)}{=}&\boldsymbol{E}_{w',[i-T-k+|\boldsymbol{\gamma}_{w'}\cap[1,i-T-2k-1]|-|\boldsymbol{\delta}_{w'}\cap [1,i-T-2k-1]|,i-k+|\boldsymbol{\gamma}_{w'}\cap[1,i-T-2k-1]|-|\boldsymbol{\delta}_{w'}\cap [1,i-T-2k-1]|+x-x']}\\
    \overset{(d)}{=}&c_{[i-T-k,i-k+x-x']},
\end{align*}
where $(a)$ and $(d)$ follows from \eqref{equation:noerrornearistar} and the fact that $|\gamma_w|+|\delta_w|\le k$ for $w\in[1,d]$, $(b)$ follows from \eqref{equation:shiftinread}, and $(c)$ follows from \eqref{equation:contrary}.

If $x'<x$, we have that
\begin{align*}
    &c_{[i-T-k,i-k-x+x']}\\
    \overset{(a)}{=}&\boldsymbol{E}_{w,[i-T-k+|\boldsymbol{\gamma}_w\cap[1,i-T-2k-1]|-|\boldsymbol{\delta}_w\cap [1,i-T-2k-1]|,i-k-x+x'+|\boldsymbol{\gamma}_w\cap[1,i-T-2k-1]|-|\boldsymbol{\delta}_w\cap [1,i-T-2k-1]|]}\\    
    =&\boldsymbol{E}_{w',[i-T-k+|\boldsymbol{\gamma}_w\cap[1,i-T-2k-1]|-|\boldsymbol{\delta}_w\cap [1,i-T-2k-1]|+x,i-k+x'+|\boldsymbol{\gamma}_w\cap[1,i-T-2k-1]|-|\boldsymbol{\delta}_w\cap [1,i-T-2k-1]|]}\\
    =&\boldsymbol{E}_{w',[i-T-k+|\boldsymbol{\gamma}_{w'}\cap[1,i-T-2k-1]|-|\boldsymbol{\delta}_{w'}\cap [1,i-T-2k-1]|+x-x',i-k+|\boldsymbol{\gamma}_{w'}\cap[1,i-T-2k-1]|-|\boldsymbol{\delta}_{w'}\cap [1,i-T-2k-1]|]}\\
    \overset{(b)}{=}&c_{[i-T-k+x-x',i-k]},
\end{align*}
In both cases, we have that $L(\boldc,|x-x'|)\ge T+1$, contradicting to the fact that $L(\boldc,\le k)\le T$. Hence, $x'=x$ and the proof is done.
\end{proof}

\subsection{Determine Bits Outside Edit Isolated Intervals}
The following lemma shows that the bit shifts caused by errors in intervals $\cI_j$, $j\in[1,J]$ can be determined.
\begin{lemma}\label{lemma:determineshiftedit}
Let $\boldsymbol{E}\in\cE_{k}(\boldc)$ be a read matrix for some sequence $\boldc$ satisfying $L(\boldc,\le k)\le T$. Let the head distance $t$ satisfy $t>(4K+1)(T+4k+1)$. If there is an interval $[b_1,b_2]$, an interval $[p_1,p_2]\subseteq [b_1,b_2]$, and an edit isolated interval $\mathcal{I}$ satisfying $\boldsymbol{E}_{[1,d],[p_1,p_2]}\in\cE_{k'}(\boldc_{\mathcal{I}})$ for some $0<k'\le d-1$, and $\boldsymbol{E}_{w,j}=\boldsymbol{E}_{w',j}$ for any $w,w'\in[1,d]$ and $j\in([b_1,p_1-1]\cup[p_2+1,b_2])$, then
	the number of bit shifts caused by errors in interval $\mathcal{I}$, which is $|\boldsymbol{\gamma}_w\cap\mathcal{I}|-|\boldsymbol{\delta}_w\cap\mathcal{I}|$, can be decided from $\boldsymbol{E}_{[1,d],[b_1,b_2]}$, for $w\in[1,d]$. Moreover, if $\boldsymbol{E}_{w,[b_1,b_2]}=\boldsymbol{E}_{w',[b_1,b_2]}$ for any $w,w'\in[1,d]$, then $|\boldsymbol{\gamma}_w\cap\mathcal{I}|=|\boldsymbol{\delta}_w\cap\mathcal{I}|$ for any $w\in[1,d]$.
\end{lemma}
\begin{proof}
Similar to what we did in Section \ref{subsection:determinenumbers}. 
consider a set of intervals 
\begin{align*}
\mathcal{B}_{i,m} &= \begin{cases}
&[b_{1}+(i-1)t+(m-1)(T+4k+1),b_{1}+(i-1)t+m(T+4k+1)-1],\\ &\text{for $m\in [1,4k+1]$ and $i\in [0,\lceil\frac{b_2-b_1+1}{t}\rceil+1]$ satisfying $b_{1}+(i-1)t+m(T+4k+1)-1\le b_2$.}
\end{cases}.
\end{align*}
Note that the intervals $\mathcal{B}_{i,m}$ are disjoint when $t>(4k+1)(T+4k+1)$. 
For notation convenience, let 
\begin{align*}
    q_{i,m}\triangleq b_{1}+(i-1)t+(m-1)(T+4k+1)
\end{align*}
for $m\in [1,4k+1]$ and $i\in [0,\lceil\frac{b_2-b_1+1}{t}\rceil+1]$ satisfying $b_{1}+(i-1)t+m(T+4k+1)-1\le b_2$.
Let $$\cU_m=\cup_{i:q_{i,m}-1\le b_2,i\in [1,\lceil\frac{b_2-b_1+1}{t}\rceil+1]}\cB_{i,m},$$ for $m\in [1,4k+1]$. Since there are at most $2k$ errors in the first two heads, there are at least $(2k+1)$ choices of $m\in [1,4k+1]$, $m_1,\ldots,m_{2k+1}$, such that $\cU_{m_\ell}\cap (\boldsymbol{\delta}_1\cup \boldsymbol{\gamma}_1\cup \boldsymbol{\delta}_2\cup \boldsymbol{\delta}_2)\cap \cI=\emptyset$ for $\ell \in [1,2k+1]$.
For each $m\in [1,4k+1]$ and integer $i\ge 1$ such that $q_{i,m}-1\le b_2$, find the unique integer $x_{m,i}\in [0,k]$ such that
\begin{align}\label{equation:shiftdetermine1}
\boldsymbol{E}_{1,[q_{i,m+1}+k,q_{i,m+2}-k-1-x_{m,i}]}
=\boldsymbol{E}_{2,[q_{i,m+1}+k+x_{m,i},q_{i,m+2}-k-1]} 
\end{align}
or~$x_{m,i}\in [-k,-1]$ such that 
\begin{align}\label{equation:shiftdetermine2}
\boldsymbol{E}_{1,[q_{i,m+1}+k-x_{m,i},q_{i,m+2}-k-1]}
=\boldsymbol{E}_{2,[q_{i,m+1}+k,q_{i,m+2}-k-1+x_{\ell,i}]} 
\end{align}
If no such index or more than one exist, let $x_{m,i}=k+1$. 
Since $[q_{i,m_\ell},q_{i,m_\ell+1}-1]\cap(\boldsymbol{\delta}_1\cup \boldsymbol{\gamma}_1\cup \boldsymbol{\delta}_2\cup \boldsymbol{\delta}_2)\cap \cI=\emptyset$, we have that 
\begin{align*}
&\boldsymbol{E}_{1,[q_{i,m_\ell}+|\boldsymbol{\gamma}_1\cap[b_1,q_{i,m_\ell}-1]|-|\boldsymbol{\delta}_1\cap[b_1,q_{i,m_\ell}-1]|,q_{i,m_\ell+1}-1+|\boldsymbol{\gamma}_1\cap[b_1,q_{i,m_\ell}-1]|-|\boldsymbol{\delta}_1\cap[b_1,q_{i,m_\ell}-1]|]}\\
=&\boldc_{[q_{i,m_\ell},q_{i,m_\ell+1}-1]}\\
=&\boldsymbol{E}_{2,[q_{i,m_\ell}+|\boldsymbol{\gamma}_2\cap[b_1,q_{i,m_\ell}-1]|-|\boldsymbol{\delta}_2\cap[b_1,q_{i,m_\ell}-1]|,q_{i,m_\ell+1}-1+|\boldsymbol{\gamma}_2\cap[b_1,q_{i,m_\ell}-1]|-|\boldsymbol{\delta}_2\cap[b_1,q_{i,m_\ell}-1]|]}\\    
\end{align*}
which implies that the integer $x_{m_\ell,i}\in [-k,k]$ satisfying \eqref{equation:shiftdetermine1} and \eqref{equation:shiftdetermine2} can be found for $\ell\in [1,2k+1]$. According to Proposition \ref{proposition:basicshifts}, such $x_{m_\ell,i}$ is unique.
In the following, we show that 
\begin{align}\label{equation:determineshiftsfinal}
    |\boldsymbol{\gamma}_w\cap\mathcal{I}|-|\boldsymbol{\delta}_w\cap\mathcal{I}|
    =\sum_{i:q_{i,m}-1\le b_2,i\in [1,\frac{b_2-b_1+1}{t}+1]}x_{m_\ell,i}
\end{align}
for $\ell\in [1,2k+1]$.

For any fixed $\ell\in [1,2k+1]$, let $i^*$ be the largest integer such that $(\boldsymbol{\gamma}_1\cup \boldsymbol{\delta}_1)\cap\mathcal{I}\cap[1,q_{i^*,m+1}+k-1]=\emptyset$.
Note that $x_{m_\ell,i}=0$ for $i\in [1,i^*]$, because $\cI$ is edit isolated and $(\boldsymbol{\gamma}_w\cup \boldsymbol{\delta}_w)\cap\mathcal{I}\cap[1,q_{i^*,m+1}+k-1]=\emptyset$ for $w\in[1,d]$. Hence, we have
 $|\boldsymbol{\gamma}_w\cap\mathcal{I}\cap[1,q_{i,m+1}-1]|-|\boldsymbol{\delta}_w\cap\mathcal{I}\cap[1,q_{i,m+1}-1|=0=\sum^{i^*}_{i=1}x_{m_\ell,i}$
 for $w\in\{1,2\}$.
According to Proposition  \ref{proposition:basicshifts} and definition of $x_{m,i}$, we have that \begin{align*}
x_{m_\ell,i}=&|\boldsymbol{\gamma}_2\cap[1,q_{i,m_\ell+1}+k-1]|-|\boldsymbol{\delta}_2\cap[1,q_{i,m_\ell+1}+k-1]|\\
&-|\boldsymbol{\gamma}_1\cap[1,q_{i,m_\ell+1}+k-1]|+|\boldsymbol{\delta}_1\cap[1,q_{i,m_\ell+1}+k-1]|\\
\overset{(a)}{=}&|\boldsymbol{\delta}_1\cap[q_{i-1,m_\ell+1}+k,q_{i,m_\ell+1}+k-1]|\\
&-|\boldsymbol{\gamma}_1\cap[q_{i-1,m_\ell+1}+k,q_{i,m_\ell+1}+k-1]|
\end{align*}
for $i\ge i^*+1$, where $(a)$ follows since $|\boldsymbol{\gamma}_2\cap[1,q_{i,m_\ell+1}+k-1]|=|\boldsymbol{\gamma}_1\cap[1,q_{i-1,m_\ell+1}+k-1]|$ and $|\boldsymbol{\delta}_2\cap[1,q_{i,m_\ell+1}+k-1]|=|\boldsymbol{\delta}_1\cap[1,q_{i-1,m_\ell+1}+k-1]|$. Moreover $x_{m_\ell,i}=0$ for $q_{i,m_\ell}\ge p_2+1$.
Therefore, 
\begin{align*}
&\sum_{i:q_{i,m_\ell}-1\le b_2,i\in [1,\lceil\frac{b_2-b_1+1}{t}\rceil+1]}x_{m_\ell,i}\\
=&\sum_{i:q_{i,m_\ell+1}-1\le p_2, i\ge i^*+1}(|\boldsymbol{\delta}_1\cap[q_{i-1,m_\ell+1}+k,q_{i,m_\ell+1}+k-1]|\\
&-|\boldsymbol{\gamma}_1\cap[q_{i-1,m_\ell+1}+k,q_{i,m_\ell+1}+k-1]|) \\
=&|\boldsymbol{\delta}_1\cap\cI|-|\boldsymbol{\gamma}_1\cap\cI|,
\end{align*}
where the last equality holds since $(\boldsymbol{\gamma}_1\cup \boldsymbol{\delta}_1)\cap\mathcal{I}\cap[1,q_{i^*,m+1}+k-1]=\emptyset$ and $\cI$ is edit isolated. Therefore, we have \eqref{equation:determineshiftsfinal} for $\ell\in [1,2k+1]$. Find the majority of $\sum_{i:q_{i,m}-1\le b_2,i\in [1,\frac{b_2-b_1+1}{t}+1]}x_{m,i}$
for $m\in [1,4k+1]$, we obtain the value $|\boldsymbol{\delta}_w\cap\cI|-|\boldsymbol{\gamma}_w\cap\cI|$ for $w\in [1,d]$.

Finally, since $x_{m,i} =0$ for each pair of $(m,i)$ when $\boldsymbol{E}_{w,[b_1,b_2]}=\boldsymbol{E}_{w',[b_1,b_2]}$ for any $w,w'\in[1,d]$, we have $|\boldsymbol{\gamma}_w\cap\mathcal{I}|=|\boldsymbol{\delta}_w\cap\mathcal{I}|$ for any $w\in[1,d]$. 
\end{proof}
The next lemma shows that we can recover most of the bits in $\boldc$. Before stating the lemma, we define the notion of a minimum edit isolated interval. An interval $\cI$ is called a minimum edit isolated interval if there is no strict sub-interval $\cI'\subsetneq \cI$ of $\cI$ that is edit isolated.

We note that the error locations in all heads are contained in a disjoint set of minimum isolated intervals.
\begin{lemma}\label{lemma:determineshiftsforedits}
	Let $\{[b_{1j},b_{2j}]\}^J_{j=1}$ be the set of output intervals in the algorithm before Lemma \ref{lemma:properties} and $\{\cI_j\}^J_{j=1}$ be the corresponding edit isolated intervals.   Let $\boldsymbol{E}\in\cE_{k}(\boldc)$ be a read matrix for some sequence $\boldc$ satisfying $L(\boldc,\le k)\le T$. For an index $i$ not in 
    any
	minimum edit isolated interval that is disjoint with $\cI_j$, $j\in [1,J]$, 
	if the column index of the bit $\boldsymbol{E}_{1,i-|[1:i-1]\cap\boldsymbol{\delta}_1|+|[1:i-1]\cap\boldsymbol{\gamma}_1|}$ coming from $c_i$ in the first read is not contained in one of the output intervals $[b_{1j},b_{2j}]$, 
	 the bit $c_i$ can be correctly recovered given $\boldsymbol{E}$.
\end{lemma} 
\begin{proof}
Assume that $b_{11}<b_{12}<\ldots<b_{1J}$.
For each output interval $[b_{1j},b_{2j}]$, $j\in[1,J]$, let $q_j\in [b_{1j},b_{2j}]$ be the largest integer such that there exist $w,w'\in[1,d]$ satisfying $\boldsymbol{E}_{w,q_j}\ne \boldsymbol{E}_{w',q_j}$. We show that $q_j\in[b_{1j}+k+1, b_{2j}-k-1]$ for $j\in[1,J]$, unless when $b_{11}=1$ or $b_{2J}=n'\in[n+1,n+2k+1]$, we can assume that $b_{11}=-k-1$ and $b_{2J}=n+2k+2$, which does not affect the result. 
Note that for any index $i$ such that there exist $w,w'\in[1,d]$ satisfying $\boldsymbol{E}_{w,i}\ne \boldsymbol{E}_{w',i}$, the indices $[i+1,i+kdt+t]$ are not marked and contained in some output interval. We have that $q_j\le b_{2j}-k-1$. Similarly, $q_j\ge b_{1j}+k+1$ because the intervals $[q_j-kdt-t,q_j]$, $j\in[1,J]$ is not marked.

According to Lemma \ref{lemma:properties},
each output interval $[ b_{1j},b_{2j}]$ is associated with an edit isolated interval $\cI_j$, $j\in[1,J]$. 
Note that for any minimum edit isolated interval $[i_1,i_2]$ that is disjoint with $\cI_j$  for $j\in[1,J]$, we have that
\begin{align*}
\boldsymbol{E}_{w,i}=\boldsymbol{E}_{w',i}
\end{align*}
for any $w,w'\in[1,d]$ and $i\in[i_1,i_2]$. 
By Lemma \ref{lemma:determineshiftedit}, we have that $|[i_1,i_2]\cap\boldsymbol{\delta}_1|=|[i_1,i_2]\cap\boldsymbol{\gamma}_1|$, i.e., there is no bit shift caused by errors in interval $[i_1,i_2]$. In addition, the shift caused by errors in interval $\cI_j$, $j\in[1,J]$, which is $|\cI_j\cap\boldsymbol{\gamma}_1|-|\cI_j\cap\boldsymbol{\delta}_1|=s_j$, can be determined. This implies that 
\begin{align}\label{equation:shiftfornoneditisolatedintervals}
  &|[1:i'-1]\cap\boldsymbol{\gamma}_1|-|[1:i'-1]\cap\boldsymbol{\delta}_1|\nonumber\\
  =&\sum_{j:b_{2j}<i'+|[1:i'-1]\cap\boldsymbol{\gamma}_1|-|[1:i'-1]\cap\boldsymbol{\delta}_1|}s_j \nonumber\\
  =&\sum_{j:q_{j}<i'}s_j \nonumber\\  
  \overset{(a)}{=}&\sum_{j:q_{j}<i'+|[1:i'-1]\cap\boldsymbol{\gamma}_1|-|[1:i'-1]\cap\boldsymbol{\delta}_1|}s_j  
\end{align}
for any $i'$ satisfying: (1) $i'-|[1:i'-1]\cap\boldsymbol{\delta}_1|+|[1:i'-1]\cap\boldsymbol{\gamma}_1|$ not in any output interval $[b_{1j},b_{2j}]$, $j\in [1,J]$. (2) $i'$ is 
not in any minimum edit isolated interval that is disjoint with $\cI_j$, $j\in [1,J]$ 
. The equality $(a)$ holds because $q_{j}\in [b_{1j}+k+1, b_{2j}-k-1]$, and
$i'>q_{j}$ only when $b_{2j}<i'+|[1:i'-1]\cap\boldsymbol{\gamma}_1|+|[1:i'-1]\cap\boldsymbol{\delta}_1|$
. For the same reason, $i'<q_{j}$ only when  $b_{1j}>i'+|[1:i'-1]\cap\boldsymbol{\gamma}_1|+|[1:i'-1]\cap\boldsymbol{\delta}_1|$. 

For any $\boldsymbol{E}_{1,i}$ such that $i$ is not included in any output interval $[b_{1j},b_{2j}]$, $j\in[1,J]$, let
\begin{align}\label{equation:editrecoverbit}
    c'_{i-\sum_{j:q_{j}<i}s_j}=\boldsymbol{E}_{1,i}.
\end{align}
be an estimate of the bit $c_{i-\sum_{j:q_{j}<i}s_j}$. 
Then, for any index $i'$ such that 
$i'-|[1:i'-1]\cap\boldsymbol{\delta}_1|+|[1:i'-1]\cap\boldsymbol{\gamma}_1|$ is not included in any output interval and $i'$ is
not in any minimum edit isolated interval that is disjoint with $\cI_j$, $j\in [1,J]$, we have that 
\begin{align*}
&c_{i'}\\
=&\boldsymbol{E}_{1,i'-|[1:i'-1]\cap\boldsymbol{\delta}_1|+|[1:i'-1]\cap\boldsymbol{\gamma}_1|}\\
=&c'_{i'-|[1:i'-1]\cap\boldsymbol{\delta}_1-|+|[1:i'-1]\cap\boldsymbol{\gamma}_1|-\sum_{j:q_{j}<i'-|[1:i'-1]\cap\boldsymbol{\delta}_1-|+|[1:i'-1]\cap\boldsymbol{\gamma}_1|}s_j}\\
=&c'_{i'},
\end{align*}
where the last equality follows from \eqref{equation:shiftfornoneditisolatedintervals}. 
Therefore, the proof is done.
\end{proof}

\subsection{Correcting $k<d$ Deletions and Insertions}
The cases when $k<d$ are addressed in the following lemma, which proves the first part of Theorem \ref{theorem:mainedit}, where $k<d$.
\begin{lemma}\label{lemma:keditkheads}
	Let $\boldsymbol{E}\in\cE_{k}(\boldc)$ be a read matrix for some sequence $\boldc$ satisfying $L(\boldc,\le k)\le T$. Let the distance $t$ satisfy $t>(\frac{k^2}{4}+ 3k)(T+3k+1)+T+5k+1$. 
	If there is an interval $[b_1,b_2]$, an interval $[p_1,p_2]\subseteq [b_1,b_2]$, and an edit isolated interval $\mathcal{I}$ satisfying $\boldsymbol{E}_{[1,d],[p_1,p_2]}\in\cE_{k'}(\boldc_{\mathcal{I}})$ for some $k'\le d-1$, and $\boldsymbol{E}_{w,j}=\boldsymbol{E}_{w',j}$ for any $w,w'\in[1,d]$ and $j\in([b_1,p_1-1]\cup[p_2+1,b_2])$, then we can obtain a sequence $\bolde\in\{0,1\}^{p_1-b_1+b_2-p_2+|\cI|}$ such that $\bolde_{[1,p_{1}-b_{1}]}=\boldsymbol{E}_{w,[b_1,p_1-1]}$ for $w\in[1,d]$, $\bolde_{[p_{1}-b_{1}+1,p_1-b_1+|\cI|]}=\boldc_{\cI}$, and $\bolde_{[p_1-b_1+|\cI|+1,p_1-b_1+|\cI|+b_2-p_2]}=\boldsymbol{E}_{w,[p_2+1,b_2]}$ for $w\in[1,d]$. 
\end{lemma}
\begin{proof}
Let $i^*$ be the minimum index such that $i^*\ge p_1$ and there exist different $w,w'\in[1,d]$ satisfying $\boldsymbol{E}_{w,i^*}\ne \boldsymbol{E}_{w',i^*}$. Let $\boldsymbol{E}_{w^*,i^*}$ be the minority bit among $\{\boldsymbol{E}_{w,i^*}\}^d_{w=1}$, i.e., there are at most $\lfloor \frac{d}{2}\rfloor$ bits among $\{\boldsymbol{E}_{w,i^*}\}^d_{w=1}$ being equal to $\boldsymbol{E}_{w^*,i^*}$.
 We will first show that there are edit errors occur near index $i^*$ in the $w^*$-th head, unless when the numbers of $1$-bits and $0$-bits among $\{\boldsymbol{E}_{w,i^*}\}^d_{w=1}$ are equal, edit errors occur near index $i^*$ in the first head.
To this end, we begin with the following proposition.
\begin{proposition}\label{proposition:majoritybit}
Let $\boldsymbol{E}\in\cE_k(\boldc)$ be a read matrix for some sequence $\boldc$ satisfying $L(\boldc,\le k)\le T$. Let $i^*>0$ be an integer such that $\boldsymbol{E}_{w,[i^*-T-2k-1,i^*-1]}=\boldsymbol{E}_{w',[i^*-T-2k-1 ,i^*-1]}$ for any $w',w\in [1,d]$.
For any $w_1,w_2\in[1,d]$ such that no error occurs in interval $[i^*-T-2k-1,i^*+k-1]$ in the $w_1$-th and $w_2$-th head, i.e., 
\begin{align}\label{equation:noerrornearistar1}
&(\boldsymbol{\delta}_{w_1}\cup \boldsymbol{\gamma}_{w_1})\cap [i^*-T-2k-1,i^*+k-1]= \emptyset, \text{ and}\nonumber\\
&(\boldsymbol{\delta}_{w_2}\cup \boldsymbol{\gamma}_{w_2})\cap [i^*-T-2k-1,i^*+k-1]= \emptyset,
\end{align}
the bits $\boldsymbol{E}_{w_1,i^*}$ and $\boldsymbol{E}_{w_2,i^*}$ are equal.
\end{proposition}
\begin{proof}
According to Proposition \ref{proposition:basicshifts}, we have that
\begin{align}\label{equation:equalshifts}
|\boldsymbol{\gamma}_{w_1}\cap[1,i^*-T-2k-2]|-|\boldsymbol{\delta}_{w_1}\cap [1,i^*-T-2k-2]|=|\boldsymbol{\gamma}_{w_2}\cap [1,i^*-T-2k-2]|- |\boldsymbol{\delta}_{w_2}\cap [1,i^*-T-2k-2]|.   
\end{align}
Then,
\begin{align*}
\boldsymbol{E}_{w_1,i^*} \overset{(a)}{=} &\boldc_{
i^*-|\boldsymbol{\gamma}_{w_1}\cap[1,i^*-T-2k-2]|+|\boldsymbol{\delta}_{w_1}
\cap [1,i^*-T-2k-2]|}\\
\overset{(b)}{=}&\boldc_{i^*-|\boldsymbol{\gamma}_{w_2}\cap[1,i^*-T-2k-2]|+|\boldsymbol{\delta_{w_2}}
\cap [1,i^*-T-2k-2]|}\\
\overset{(c)}{=}&\boldsymbol{E}_{w_2,i^*}
\end{align*}
where $(a)$ and $(c)$ follow from \eqref{equation:noerrornearistar1} and the fact that $|\boldsymbol{\gamma}_w\cap[1,i^*-T-2k-2]|-
|\boldsymbol{\delta}_w\cap [1,i^*-T-2k-2]|\le k$. Equality $(b)$ follows from \eqref{equation:equalshifts}.
\end{proof}
From Proposition \ref{proposition:majoritybit}, we can easily conclude that there is at least one error in interval $[i^*-T-2k-1,i^*+k-1]$ in one of the heads, i.e., 
$(\boldsymbol{\delta}_w\cup \boldsymbol{\gamma}_w)\cap [i^*-T-2k-1,i^*+k-1]\ne \emptyset$ 
for some $w\in[1,d]$. Otherwise the bits $\boldsymbol{E}_{w,i^*}$ are equal for all $w\in [1,d]$, contradicting to the definition of $i^*$.

Next, we need the following proposition.
\begin{proposition}\label{proposition:numberofheadswitherrors}
Let $\boldsymbol{E}\in\cE_k(\boldc)$ be a read matrix for some sequence $\boldc$ satisfying $L(\boldc,\le k)\le T$. Let  $i^*>0$ be an integer such that $\boldsymbol{E}_{w,[1,i^*-1]}=\boldsymbol{E}_{w',[1,i^*-1]}$ for any $w',w\in [1,d]$.
If $T^*\ge T+2k+1$ and $t> (k+1)T^*$, then the number of heads where at least one error occurs in interval $[i^*-T^*,i^*+k-1]$ is at most $\lfloor \frac{k+1}{2}\rfloor$, i.e., 
\begin{align*}
|\{w: (\boldsymbol{\delta}_w\cup \boldsymbol{\gamma}_w)\cap [i^*-T^*,i^*+k-1]\ne \emptyset\}|\le \lfloor \frac{k+1}{2}\rfloor
\end{align*}
Moreover, when $|\{w: (\boldsymbol{\delta}_w\cup \boldsymbol{\gamma}_w)\cap [i^*-T^*,i^*+k-1]\ne \emptyset\}|=\frac{k+1}{2}$,  at least one error occurs in $[i^*-T^*,i^*]$ in the first head, i.e., $(\boldsymbol{\delta}_1\cup \boldsymbol{\gamma}_1)\cap [i^*-T^*,i^*]\ne \emptyset$.
\end{proposition}
\begin{proof}


Let $\{w:w\in [2,d], (\boldsymbol{\delta}_w\cup \boldsymbol{\gamma}_w)\cap [i^*-T^*,i^*+k-1]\ne \emptyset\}=\{w_1,w_2,\ldots,w_M\}$ be the set of heads (not including the first head) that contains at least one error in interval $[i^*-T^*,i^*+k-1]$. Let $w_1>w_2>\ldots>w_M$. We will show that there exist a set of integers $i_1,i_2,\ldots,i_{M}\in[0,k]$ such that $i_1\ge i_2\ge \ldots\ge i_{M}$ and
\begin{align}\label{equation:atleasttwoerrors}
    &|(\boldsymbol{\delta}_1\cap [i^*-T^*-(w_\ell-1)t-(T^*+k)i_\ell,i^*-T^*-(w_\ell-2)t-(T^*+k)i_{\ell}-1]|\nonumber\\
    &+| \boldsymbol{\gamma}_1\cap [i^*-T^*-(w_\ell-1)t-(T^*+k)i_\ell,i^*-T^*-(w_\ell-2)t-(T^*+k)i_{\ell}-1]|\nonumber\\
    \ge &2
\end{align}
for $\ell\in [1,M]$. Note that the intervals $[i^*-T^*-(w_\ell-1)t-(T^*+k)i_\ell,i^*-T^*-(w_\ell-2)t-(T^*+k)i_{\ell}-1]$ are disjoint for different $\ell\in[1,M]$ and are within the interval $[-T^*-(T^*+k)(k+1),i^*-T^*-1]$, since $t>(k+1)(T^*+1)$ for $j\in [1,d]$ and $i_{\ell}\le k$ for $\ell\in [1,M]$. Then, the number of errors in the first head is at least $2|\{w:(\boldsymbol{\delta}_w\cup \boldsymbol{\gamma}_w)\cap [i^*-T^*,i^*+k-1]\ne \emptyset,w\in [2,d]\}|+\1((\boldsymbol{\delta}_1\cup \boldsymbol{\gamma}_1)\cap [i^*-T^*,i^*+k-1]\ne \emptyset)$, where $\1(A)$ is the indicator that equals $1$ when $A$ is true and equals $0$ otherwise. 
Hence, we have that
\begin{align*}
    2|\{w:(\boldsymbol{\delta}_w\cup \boldsymbol{\gamma}_w)\cap [i^*-T^*,i^*+k-1]\ne \emptyset,w\in [2,d]\}|+\1((\boldsymbol{\delta}_1\cup \boldsymbol{\gamma}_1)\cap [i^*-T^*,i^*+k-1]\ne \emptyset)\le k
\end{align*}
Then, it can be easily verified that the proposition follows.


Now we find the set of integers $i_1\ge i_2\ge \ldots \ge i_{M}$ satisfying \eqref{equation:atleasttwoerrors}. Let $i_0=k$. Starting from $\ell=1$ to $\ell=M$, find the largest integer $i_\ell$ such that $i_{\ell}\le i_{\ell-1}$ and no errors occur in interval $[i^*-T^*-(T^*+k)(i_\ell+1),i^*-T^*-(T^*+k)i_\ell-1]$ in the $w_\ell$-th or the $(w_\ell-1)$-th heads, i.e.,
\begin{align}\label{equation:findiell}
    (\boldsymbol{\gamma}_{w_\ell}\cup \boldsymbol{\delta}_{w_\ell}\cup \boldsymbol{\gamma}_{w_\ell-1}\cup \boldsymbol{\delta}_{w_\ell-1})\cap [i^*-T^*-(T^*+k)(i_\ell+1),i^*-T^*-(T^*+k)i_\ell-1]=\emptyset.
\end{align}
We show that such an $\ell\in[1,M]$ can be found as long as $t>(T^*+k)(k+2)$. Note that in the above procedure, for each integer $i\in[i_\ell+1,k]$, there is at least an edit error occurring in interval $[i^*-T^*-(T^*+k)(i+1),i^*-T^*-(T^*+k)i-1]$ in one of the heads $w$, which corresponds to an error that occurs in interval
$[i^*-T^*-(T^*+k)(i+1)-(w-1)t,i^*-T^*-(T^*+k)i-1-(w-1)t]$ in the first head. 
In addition, the intervals $[i^*-T^*-(T^*+k)(i+1)-(w-1)t,i^*-T^*-(T^*+k)i-1-(w-1)t]$ are disjoint for different pairs $(i,w)$, as long as $t\ge (T^*+k)(k+2)$. Since there are at most $k$ errors in the first head and there are $k+1$ choices of $i_\ell$, such an $i_\ell$ satisfying \eqref{equation:findiell} can be found.

Since $\boldsymbol{E}_{w_\ell,i}=\boldsymbol{E}_{w_\ell-1,i}$ for $i\in [i^*-T^*-(T^*+k)(i_\ell+1),i^*-T^*-(T^*+k)i_\ell-1]$, by Proposition \ref{proposition:basicshifts} we have that 
\begin{align}\label{equation:equalshift1}
    &|\boldsymbol{\gamma}_{w_\ell-1}\cap[1,i^*-T^*-(T^*+k)(i_\ell+1)-1]|-|\boldsymbol{\delta}_{w_\ell-1}\cap [1,i^*-T^*-(T^*+k)(i_\ell+1)-1]|\nonumber\\
    =&|\boldsymbol{\gamma}_{w_\ell}\cap [1,i^*-T^*-(T^*+k)(i_\ell+1)-1]|- |\boldsymbol{\delta}_{w_\ell}\cap [1,i^*-T^*-(T^*+k)(i_\ell+1)-1]|. 
\end{align}
On the other hand, we have that 
\begin{align}\label{equation:equalshift2}
    &|\boldsymbol{\gamma}_{w_\ell-1}\cap[1,i^*-T^*-(T^*+k)(i_\ell+1)-1-t]|-|\boldsymbol{\delta}_{w_\ell-1}\cap [1,i^*-T^*-(T^*+k)(i_\ell+1)-1-t]|\nonumber\\
    =&|\boldsymbol{\gamma}_{w_\ell}\cap [1,i^*-T^*-(T^*+k)(i_\ell+1)-1]|- |\boldsymbol{\delta}_{w_\ell}\cap [1,i^*-T^*-(T^*+k)(i_\ell+1)-1]|.     
\end{align}
Eq. \eqref{equation:equalshift1} and Eq. \eqref{equation:equalshift2} imply that 
\begin{align}\label{equation:equalshift3}
    &|\boldsymbol{\gamma}_{w_\ell-1}\cap[i^*-T^*-(T^*+k)(i_\ell+1)-t,i^*-T^*-(T^*+k)(i_\ell+1)-1]|\nonumber\\
    =&|\boldsymbol{\delta}_{w_\ell-1}\cap [i^*-T^*-(T^*+k)(i_\ell+1)-t,i^*-T^*-(T^*+k)(i_\ell+1)-1]|     
\end{align}
Since $(\boldsymbol{\gamma}_{w_\ell}\cup \boldsymbol{\delta}_{w_\ell})\cap [i^*-T^*,i^*+k-1]\ne\emptyset$ by definition of $w_\ell$, we have that  
\begin{align*}
  &(\boldsymbol{\gamma}_{w_\ell-1}\cup \boldsymbol{\delta}_{w_\ell-1})\cap
  [i^*-T^*-t,i^*+k-1-t]\\
    \subseteq &(\boldsymbol{\gamma}_{w_\ell-1}\cup \boldsymbol{\delta}_{w_\ell-1})\cap[i^*-T^*-(T^*+k)(i_\ell+1)-t,i^*-T^*-(T^*+k)(i_\ell+1)-1]\\
  \ne &\emptyset  
\end{align*}
Together with\eqref{equation:equalshift3}, we have that 
\begin{align*}
    &|\boldsymbol{\gamma}_{w_\ell-1}\cap[i^*-T^*-(T^*+k)(i_\ell+1)-t,i^*-T^*-(T^*+k)(i_\ell+1)-1]|\\
    +&|\boldsymbol{\delta}_{w_\ell-1}\cap [i^*-T^*-(T^*+k)(i_\ell+1)-t,i^*-T^*-(T^*+k)(i_\ell+1)-1]|\\
    \ge&2,
\end{align*}
which implies \eqref{equation:atleasttwoerrors} because $\boldsymbol{\gamma}_{w_\ell-1}=\boldsymbol{\gamma}_{1}+(w_\ell-2)t$ and $\boldsymbol{\delta}_{w_\ell-1}=\boldsymbol{\delta}_{1}+(w_\ell-2)t$. Hence, the proof is done.
\end{proof}
Recall that $w^*\in[1,d]$ is a head index such that $\boldsymbol{E}_{w^*,i^*}$ is a minority bit among $\{\boldsymbol{E}_{w,i^*}\}^d_{w=1}$, i.e., there are at most $\frac{d}{2}$ bits among $\{\boldsymbol{E}_{w,i^*}\}^d_{w=1}$ that is equal to $\boldsymbol{E}_{w^*,i^*}$. 
By Proposition \ref{proposition:majoritybit} and Proposition \ref{proposition:numberofheadswitherrors}, we conclude that when $k<d$, we have that $ (\boldsymbol{\delta}_{w^*}\cup \boldsymbol{\gamma}_{w^*})\cap [i^*-T-2k-1,i^*+k-1]\ne \emptyset$, if the number of bits among $\{\boldsymbol{E}_{w,i^*}\}^d_{w=1}$ being equal to $\boldsymbol{E}_{w^*,i^*}$ is  is less than $d/2$. If $k<d$ and the number of bits among $\{\boldsymbol{E}_{w,i^*}\}^d_{w=1}$ being equal to $\boldsymbol{E}_{w^*,i^*}$ is  is exactly $d/2$, we have that $ (\boldsymbol{\delta}_{1}\cup \boldsymbol{\gamma}_{1})\cap [i^*-T-2k-1,i^*+k-1]\ne \emptyset$.

Now we have found a $w^*$ with 
\begin{align}\label{equation:nonemptywstar}
(\boldsymbol{\delta}_{w^*}\cup \boldsymbol{\gamma}_{w^*})\cap [i^*-T-2k-1,i^*+k-1]\ne \emptyset.
\end{align}
In the remaining part of the proof, we show how to use knowledge of $w^*$ to correct at least one error for each head, and reduce the $d$-head case to a $(d-1)$-head case. Then, the lemma follows by induction, since the case when $d=1$ is obvious.
Assume that $w^*\le d-1$. The procedure when $w^*=d$ will be similar.

Note that $(\boldsymbol{\delta}_{w}\cup \boldsymbol{\gamma}_{w})\cap [i^*-T-2k-1+(w-w^*)t,i^*+k-1+(w-w^*)t]\ne \emptyset$ by \eqref{equation:nonemptywstar}.
Consider the set of intervals
\begin{align*}
 [i^*+2k+(\ell-1)(T+3k+1)+(w-w^*)t,i^*+2k-1+\ell(T+3k+1)+(w-w^*)t]    
\end{align*}
for $\ell\in[1,\frac{k^2}{4}+ 3k]$ and $w\in[1,d]$. For notation convenience, denote
\begin{align}\label{equation:intervalii}
    v_{w,\ell}\triangleq i^*+2k+(\ell-1)(T+3k+1)+(w-w^*)t
\end{align}
for $\ell\in[1,\frac{k^2}{4}+ 3k]$ and $w\in[1,d]$.
For each pair $\ell\in [1,\frac{k^2}{4}+ 3k]$ and $w\in[1,d-1]$, find a unique index $x_{w,\ell}\in[0,k]$, such that 
\begin{align}\label{equation:positiveshift}
\boldsymbol{E}_{w,[v_{w,\ell},v_{w,\ell+1}-1-x_{w,\ell}]}
=\boldsymbol{E}_{w+1,[v_{w,\ell}+x_{w,\ell},v_{w,\ell+1}-1]} 
\end{align}
or~$x_{w,\ell}\in [-k,-1]$ such that 
\begin{align}\label{equation:negativeshift}
&\boldsymbol{E}_{w,[v_{w,\ell}-x_{w,\ell},v_{w,\ell+1}-1]}=\boldsymbol{E}_{w+1,[v_{w,\ell},v_{w,\ell}+x_{w,\ell+1}-1]} 
\end{align}
If no such index or more than one exist, let $x_{w,\ell}=k+1$. 

Given $x_{w,\ell}$, $\ell\in[1,\frac{k^2}{4}+ 3k]$ and $w\in[1,d]$, define a binary vector $\boldz\in\{0,1\}^{\frac{k^2}{4}+ 3k}$ as follows:
\begin{align}\label{equation:definey}
z_\ell &= \begin{cases}
&1,\text{if there exists a $w\in[1,d-1]$ such that $x_{w,\ell}=k+1$}\\
&1,\text{if there exists a $w\in[1,d-1]$ such that $x_{w,\ell}\ne x_{w,\ell-1}$ and $x_{w,\ell},x_{w,\ell-1}\in[-k,k]$}\\
&0, \text{else}\\
\end{cases}
\end{align}
for $\ell\in \frac{k^2}{4}+ 3k$.
In \eqref{equation:definey}, it is assumed that $x_{w,0}=x_{w,1}$ for  $w\in [1,d-1]$.


Let $y^*=|(\boldsymbol{\gamma}_{w^*}\cup\boldsymbol{\delta}_{w^*}\cup \boldsymbol{\gamma}_{w^*+1}\cup\boldsymbol{\delta}_{w^*+1})\cap [v_{w^*,1}-k,v_{w^*,\frac{k^2}{4}+ 3k}+T+4k]|$ be the number of errors that occur in interval $[v_{w^*,1}-k,v_{w^*,\frac{k^2}{4}+ 3k}+T+4k]$ in the $w^*$-th or $(w^*+1)$-th head.
Note that $y^*=|(\boldsymbol{\gamma}_{w}\cup\boldsymbol{\delta}_{w}\cup \boldsymbol{\gamma}_{w+1}\cup\boldsymbol{\delta}_{w+1})\cap [v_{w,1}-k,v_{w,\frac{k^2}{4}+ 3k}+T+4k]|$ for $w\in[1,d]$. Moreover,  $\boldsymbol{E}_{w,[v_{w,1},v_{w,\frac{k^2}{4}+ 3k}+T+3k]}$ can be obtained by a subsequence of $\boldc_{[v_{w,1}-k,v_{w,\frac{k^2}{4}+ 3k}+T+4k]}$ after at most $y^*$ deletions and insertions in interval $[v_{w,1}-k,v_{w,\frac{k^2}{4}+ 3k}+T+4k]$ in the $w$-th head, $w\in [1,d-1]$.

We first show that $y^*\le k-1$. Note that the $|(\boldsymbol{\gamma}_{w^*}\cup\boldsymbol{\delta}_{w^*})\cap [i^*+k,n']|$ errors that occur after index $i^*+k$ in the $w^*$-th head, occur after index $i^*+k+t>v_{w^*,\frac{k^2}{4}+ 3k}+T+4k+1$ in the $(w^*+1)$-th head. Moreover, the errors that occur in interval $[i^*-T-2k-1,i^*+k-1]$ in the $w^*$-th head occur after $i^*-T-2k-1+t>v_{w^*,\frac{k^2}{4}+ 3k}+T+4k+1$ in the $(w^*+1)$-th head, since $t>(\frac{k^2}{4}+ 3k)(T+3k+1)+T+5k+1$. Recall that $(\boldsymbol{\delta}_{w^*}\cup \boldsymbol{\gamma}_{w^*})\cap [i^*-T-2k-1,i^*+k-1]\ne \emptyset$. 
Hence, there are at most  $k-|(\boldsymbol{\gamma}_{w^*}\cup\boldsymbol{\delta_{w^*}})\cap [i^*+k,n']|-1+|(\boldsymbol{\gamma}_{w^*}\cup\boldsymbol{\delta_{w^*}})\cap [i^*+k,n']|=k-1$ errors that occur in interval $[i^*+k,v_{w^*,\frac{k^2}{4}+ 3k}+T+4k]$ in the $w^*$-th or $(w^*+1)$-th head.  

Next, we show that there are at most $(2k-2)$ 1 entries in $\boldz$. Note that a single error in interval $[i^*+k,v_{w^*,\frac{k^2}{4}+ 3k}+T+4k]$ in the $w^*$-th or $(w^*+1)$-th head affects the value of at most a single entry $x_{w,\ell}$ and the entries $x_{w,\ell+1},\ldots,x_{w,\frac{k^2}{4}+ 3k}$ increase or decrease by $1$ for $w\in[1,d]$. This generates at most two $1$ entries in $\boldz$. Hence there are at most $2y^*\le 2k-2$ $1$ entries in $\boldz$.


Let $y$ be the number of $1$ runs in $\boldz$.
We show that there exists a $0$-run $(z_{i+1},\ldots,z_{i+k-y+2})$ of length $k-y+2$, for some $i\in[0,\frac{k^2}{4}+ 2k+y]$, which indicates that
\begin{align}\label{equation:longshift1}
\boldsymbol{E}_{w,[v_{w,i+1},v_{w,i+k-y+3}-x_{w,i+1}-1]}=\boldsymbol{E}_{w+1,[v_{w,i+1}+x_{w,i+1},v_{w,i+k-y+3}-1]}    
\end{align}
if $x_{w,i+1}\in [0,k]$ or
\begin{align}\label{equation:longshift2}
\boldsymbol{E}_{w,[v_{w,i+1}-x_{w,i+1},v_{w,i+k-y+3}-1]}=\boldsymbol{E}_{w+1,[v_{w,i+1},v_{w,i+k-y+3}+x_{w,i+1}-1]}    
\end{align}
if $x_{w,i+1}\in [-k,-1]$, for every $w\in[1,d-1]$.

Suppose on the contrary, each $0$ run has length no more than $k-y+1$. Note that there are at most $y+1$ $0$ runs with $y$ 1 runs. Therefore, the length of $\boldz$ is upper bounded by  
\begin{align*}
    \frac{k^2}{4}+ 3k\le & (y+1)(k-y+1) +2k-2\\
   = &  -y^2+ky+3k-1 \\
   \le & \frac{k^2}{4}+ 3k-1\\
\end{align*}
a contradiction. 

We have proved the existence of a 0 run $(z_{i+1},\ldots,z_{i+k-y+2})$, which implies \eqref{equation:longshift1} and \eqref{equation:longshift2}. We now show that there are at most $k-y+1$ errors occur in interval $[v_{w,i+1},v_{w,i+k-y+3}-1]$ in the $w$ and/or $(w+1)$-th head, for $w\in [1,d-1]$. As mentioned above, a single error in interval $[i^*+k,v_{w^*,\frac{k^2}{4}+ 3k}+T+4k]$ in the $w^*$-th or $(w^*+1)$-th head affects the value of at most a single entry $x_{w,\ell}$ and the entries $(x_{w,\ell+1},\ldots,x_{w,\frac{k^2}{4}+ 3k})$ increase or decrease by $1$ for $w\in[1,d]$. This generates at most a single $1$ run in $\boldz$.
In addition, errors in interval $[v_{w,i+1},v_{w,i+k-y+3}-1]$ in the $w$ and/or $(w+1)$-th head generate at most two $1$ runs that include $z_{i}$ and $z_{i+k-y+3}$. Therefore, there are at least $y-2$ $1$ runs in $\boldz$ that are generated by at least $y-2$ errors in $[i^*+k,[i^*+k,v_{w^*,\frac{k^2}{4}+ 3k}+T+4k]]\backslash [v_{w,i+1},v_{w,i+k-y+3}-1]$. Hence, the number of errors in interval $[v_{w,i+1},v_{w,i+k-y+3}-1]$ in the $w$ and/or $(w+1)$-th head is at most $y^*-y+2\le k-y+1$.

Therefore, there exists an integer $\ell\in [i+1,i+k-y+2]$ such that no errors occur in interval $[v_{w,\ell},v_{w,\ell+1}-1]$ in the $w$ and/or $(w+1)$-th head, which implies that $\boldsymbol{E}_{w+1,[p_1,v_{w,\ell}-|\boldsymbol{\delta}_{w+1}\cap \mathcal{I}\cap[1,v_{w,\ell}-1]|+|\boldsymbol{\gamma}_{w+1}\cap \mathcal{I}\cap[1,v_{w,\ell}-1]|+k]}$ is obtained from $\boldc_{\mathcal{I}\cap [1,v_{w,\ell}+k]}$, after deletion errors at locations $\boldsymbol{\delta}_{w+1}\cap\mathcal{I}\cap[1,v_{w,\ell}-1]$ and insertion errors at locations $\boldsymbol{\gamma}_{w+1}\cap\mathcal{I}\cap[1,v_{w,\ell}-1]$. 
Moreover, 
\begin{align*}
&\boldsymbol{E}_{w,[v_{w,\ell}+k+1-|\boldsymbol{\delta}_{w+1}\cap \mathcal{I}\cap[1,v_{w,\ell}-1]|+|\boldsymbol{\gamma}_{w+1}\cap \mathcal{I}\cap[1,v_{w,\ell}-1]|-x_{w,\ell},p_2]} \\
\overset{(a)}{=}&\boldsymbol{E}_{w,[v_{w,\ell}+k+1-|\boldsymbol{\delta}_{w}\cap \mathcal{I}\cap[1,v_{w,\ell}-1]|+|\boldsymbol{\gamma}_{w}\cap \mathcal{I}\cap[1,v_{w,\ell}-1]|,p_2]},
\end{align*}
where $(a)$ follows from Proposition \ref{proposition:basicshifts},
can be obtained from $\boldc_{\mathcal{I}\cap [v_{w,\ell}+k+1,n+k+1]}$, after deletion errors at locations $\boldsymbol{\delta}_{w}\cap\mathcal{I}\cap[v_{w,\ell},n+k+1]$ and insertion errors at locations $\boldsymbol{\gamma}_{w}\cap\mathcal{I}\cap[v_{w,\ell},n+k+1]$.
Therefore, by concatenating $$\boldsymbol{E}_{w+1,[p_1,v_{w,\ell}-|\boldsymbol{\delta}_{w+1}\cap \mathcal{I}\cap[1,v_{w,\ell}-1]|+|\boldsymbol{\gamma}_{w+1}\cap \mathcal{I}\cap[1,v_{w,\ell}-1]|+k]}$$ and
$$\boldsymbol{E}_{w,[v_{w,\ell}+k+1-|\boldsymbol{\delta}_{w+1}\cap \mathcal{I}\cap[1,v_{w,\ell}-1]|+|\boldsymbol{\gamma}_{w+1}\cap \mathcal{I}\cap[1,v_{w,\ell}-1]|-x_{w,\ell},p_2]},$$ we have a sequence obtained from $\boldc_{\mathcal{I}}$ by deletion errors with locations $(\boldsymbol{\delta}_{w+1}\cap\mathcal{I}\cap [1,v_{w,\ell}-1])\cup (\boldsymbol{\delta}_{w}\cap\mathcal{I}\cap [v_{w,\ell},n+k+1])$ and insertion errors at locations $(\boldsymbol{\gamma}_{w+1}\cap\mathcal{I}\cap [1,v_{w,\ell}-1])\cup (\boldsymbol{\gamma}_{w}\cap\mathcal{I}\cap [v_{w,\ell},n+k+1])$, $w\in [1,d-1]$ in $\boldc_{\mathcal{I}}$. Note that there are at most $|\boldsymbol{\delta}_{w}\cap\mathcal{I}|+|\boldsymbol{\gamma}_{w}\cap\mathcal{I}|-1$ errors in total in the concatenation, since $$|\boldsymbol{\delta}_{w+1}\cap\mathcal{I}\cap [1,v_{w,\ell}-1]|=|\boldsymbol{\delta}_{w}\cap\mathcal{I}\cap [1,v_{w,\ell}-1-t]|,$$ and
$$|\boldsymbol{\gamma}_{w+1}\cap\mathcal{I}\cap [1,v_{w,\ell}-1]|=|\boldsymbol{\gamma}_{w}\cap\mathcal{I}\cap [1,v_{w,\ell}-1-t]|,$$
and the errors occur in $[(\boldsymbol{\delta}_w\cup\boldsymbol{\gamma}_{w})\cap [v_{w,1}-T-4k-1,v_{w,1}-k-1]|]\ne\emptyset$ (see \eqref{equation:nonemptywstar}) are not included in the concatenation.
Finally, since $(z_{i+1},\ldots,z_{i+k-y+2})$ is a $0$ run, we have that $x_{w,i+1}=x_{w,i+2}=\ldots=x_{w,i+k-y+2}$ for $w\in [1,d-1]$. Hence, concatenating $\boldsymbol{E}_{w+1,[p_1,v_{w,i+1}+k]}$ and $\boldsymbol{E}_{w,[v_{w,i+1}+k+1-x_{w,i+1},p_2]}$ results in the same sequence as
concatenating $\boldsymbol{E}_{w+1,[p_1,v_{w,\ell}-|\boldsymbol{\delta}_{w+1}\cap \mathcal{I}\cap[1,v_{w,\ell}-1]|+|\boldsymbol{\gamma}_{w+1}\cap \mathcal{I}\cap[1,v_{w,\ell}-1]|+k]}$ and $\boldsymbol{E}_{w,[v_{w,\ell}+k+1-|\boldsymbol{\delta}_{w+1}\cap \mathcal{I}\cap[1,v_{w,\ell}-1]|+|\boldsymbol{\gamma}_{w+1}\cap \mathcal{I}\cap[1,v_{w,\ell}-1]|-x_{w,\ell},p_2]}$, $w\in[1,d-1]$. Note that $p_1,p_2$ and $x_{w,\ell}$ are not known by the algorithm. We concatenate $\boldsymbol{E}_{w+1,[b_1,v_{w,i+1}+k]}$ and $\boldsymbol{E}_{w,[v_{w,i+1}+k+1-x_{w,i+1},b_2]}$ for each $w\in[1,d-1]$.
Let the $d-1$ concatenated sequences be represented by a read matrix $\boldsymbol{E}'\in\{0,1\}^{(d-1)\times m'}$. Then from the above arguments, we have that $\boldsymbol{E}'_{[1,d-1],[p_1-b_1+1,m'-(b_2-p_2)]}\in \cE_{k''}(\boldc_{\cI})$ for some $k''\le k'-1$. In addition, we have that $\boldsymbol{E}'_{w,[1,p_1-b_1]}=\boldsymbol{E}_{w,[b_1,p_1-1]}$ and $\boldsymbol{E}'_{w,[m'-(b_2-p_2)+1,m']}=\boldsymbol{E}_{w,[p_2+1,b_2]}$ for any $w\in[1,d-1]$. 

For cases when $w^*=d$, the proof is similar, where instead of looking at intervals $[i^*+2k+(\ell-1)(T+3k+1)+(w-w^*)t,i^*+2k-1+\ell(T+3k+1)+(w-w^*)t]$ and defining $x_{w,\ell}$ and $z_{\ell}$, $w\in [1,d-1]$, $\ell\in[1,\frac{k^2}{4}+ 3k]$ on these intervals, we define $x_{w,\ell}$ and $z_{\ell}$ on intervals  $[i^*-T-3k-\ell(T+3k+1)+(w-w^*)t,i^*-T-3k-1-(\ell-1)(T+3k+1)+(w-w^*)t]$ for $\ell\in[1,\frac{k^2}{4}+ 3k+2]$ and $w\in[1,d-1]$. Then we find a $0$ run $(z_{i+1},\ldots,z_{i+k-y+2})$ of length $k-y+2$ in $\boldz$, where $y$ is the number of $1$ runs in $\boldz$, and concatenate 
$$\boldsymbol{E}_{w+1,[p_1,i^*-T-3k-(i+1)(T+3k+1)+(w-w^*)t+k]}$$ and $$\boldsymbol{E}_{w,[i^*-T-3k-(i+1)(T+3k+1)+(w-w^*)t+k+1-x_{w,i+1},p_2]}$$
for $w\in [1,d-1]$.

We have shown how to correct at least one error using $d$ reads. To correct all errors, we
repeat the same procedure iteratively. In each iteration, we get a read matrix with one less read. The algorithm stops when we get a read matrix where all rows are equal. Let $d^*$ be the number of rows of the read matrix after the algorithm stops. Let $\bolde$ be the first row of this matrix.

We complete the proof of Lemma \ref{lemma:keditkheads} with the following proposition, which claims that either the errors contained in the isolated interval $\cI$ are corrected in $\bolde$ and the number of errors in $\cI$ is at least $d-d^*$, or the errors contained in $\cI$ are not corrected in $\bolde$ and the number of errors in $\cI$ is at least $d+d^*$. In particular, for errors that are not detectable, i.e., when all the rows in $\boldsymbol{E}$ are equal, the number of errors contained in $\boldsymbol{E}$ is either $0$ or at least $2d$. The proposition 
will be used to prove the correctness of the encoding/decoding algorithms in Section \ref{subsection:encoding/decoding}.
\begin{proposition}\label{proposition:numberofuncorrectableerrors}
    Let the length of $\bolde$ be $m$. Then, $\bolde_{[1,p_{1}-b_{1}]}=\boldsymbol{E}_{w,[b_1,p_1-1]}$ and $\bolde_{[m-(b_2-p_2)+1,m]}=\boldsymbol{E}_{w,[p_2+1,b_2]}$ for $w\in[1,d]$. In addition, the number of errors $|(\boldsymbol{\delta}_w\cup\boldsymbol{\gamma}_w)\cap\cI|\ge d-d^*$.  
    If $\bolde_{[p_{1}-b_{1}+1,m-(b_2-p_2)]}\ne\boldc_{\cI}$, then the number of errors $|(\boldsymbol{\delta}_w\cup\boldsymbol{\gamma}_w)\cap\cI|\ge d+d^*$ for $d^*\ge 2$.
    When $d^*=1$, either $|(\boldsymbol{\delta}_w\cup\boldsymbol{\gamma}_w)\cap\cI|\ge d+d^*$ or it can be determined that $|(\boldsymbol{\delta}_w\cup\boldsymbol{\gamma}_w)\cap\cI|\ge d$. 
    In particular, if the number of errors $|(\boldsymbol{\delta}_w\cup\boldsymbol{\gamma}_w)\cap\cI|\le d-1$, then $\bolde_{[p_{1}-b_{1}+1,m-(b_2-p_2)]}=\boldc_{\cI}$.
\end{proposition}
\begin{proof}
Let $\boldsymbol{E}^i$ be the read matrix obtained after the $i$-th iteration, $i\in[1,d-d^*]$. Note that the number of rows in $\boldsymbol{E}^i$ is $d-i$. Let the number of columns in $\boldsymbol{E}^i$ be $m_i$. 

Note that $\boldsymbol{E}_{w,[b_1,p_1-1]}$ and $\boldsymbol{E}_{w,[p_2+1,b_2]}$ are equal for $w\in[1,d]$. The concatenation in the algorithm keeps the first $p_1-b_1$ bits and the last $b_2-p_2$ bits in each row. Therefore, $\boldsymbol{E}^i_{w,[1,p_{1}-b_{1}]}=\boldsymbol{E}^i_{1,[b_1,p_1-1]}$ and $\boldsymbol{E}^i_{w,[m_i-(b_2-p_2)+1,m_i]}=\boldsymbol{E}_{1,[p_2+1,b_2]}$ for $w\in[1,d]$ and $i\in[1,d-d^*]$, which implies that $\boldsymbol{s}^j_{[1,p_{1}-b_{1}]}=\boldsymbol{E}^i_{w,[b_1,p_1-1]}$ and $\boldsymbol{s}^j_{[m-(b_2-p_2)+1,m]}=\boldsymbol{E}_{w,[p_2+1,b_2]}$ for $w\in[1,d]$.

Furthermore, we proved that $\boldsymbol{E}^1_{[1,d-1],[p_1-b_1+1,m_1-(b_2-p_2)]}\in \cE_{k_1}(\boldc_{\cI})$ for some $k_1\le k'-1$. By induction, it can be proved that $\boldsymbol{E}^i_{[1,d-i],[p_1-b_1+1,m_i-(b_2-p_2)]}\in \cE_{k_i}(\boldc_{\cI})$ for some non-negative $k_i\le k_{i-1}-1$ for $i\in[2,d-d^*]$. Therefore, we have that $0\le k_{d-d^*}\le k'-(d-d^*)$ and thus, $k'\ge d-d^*$.

If $\bolde_{[p_{1}-b_{1}+1,m-(b_2-p_2)]}=\boldsymbol{E}^{d-d^*}_{w,[1,m_{d-d^*}]}\ne\boldc_{\cI}$ for $w\in[1,d^*]$, 
we let $\mathcal{I}=[i_1,i_2]$. 
When $d^*=1$, according to Lemma \ref{lemma:determineshiftsforedits}, the number of errors occur in $\cI$ can be determined. Therefore, if the parity of the number of errors occur in $\cI$ is different from the parity of $d+1$, we determine that the number of errors in $\cI$ is at least $d$.
Otherwise, the number of errors in $\cI$ is at least $d+1$. 

When $d^*\ge 2$, 
let $i'$ be the minimum index such that $i'\ge p_1-b_1+1$  satisfying $\boldsymbol{E}^{d-d^*}_{w,i'}\ne \boldc_{i_1+i'-(p_1-b_1)-1}$, i.e., $\boldsymbol{E}^{d-d^*}_{w,[p_1-b_1+1,i'-1]}=\boldc_{[i_1,i_1+i'-(p_1-b_1)-2]}$ and $\boldsymbol{E}^{d-d^*}_{w,i'}\ne \boldc_{i_1+i'-(p_1-b_1)-1}$. We show that $(\boldsymbol{\delta}_w\cup \boldsymbol{\gamma}_w)\cap [i'-T-2k'-1,i'+k'-1]\ne \emptyset$ for $w\in [1,d^*]$. Otherwise, there exists a $w^*\in [1,d^*]$, such that $(\boldsymbol{\delta}_{w^*}\cup \boldsymbol{\gamma}_{w^*})\cap [i'-T-2k'-1,i'+k'-1]= \emptyset$. Assume now that there is a virtual read $\boldsymbol{E}^{d-d^*}_{d^*+1,[1,m_{d-d^*}]}$. Assume that the distance between the $d^*$-th head and the $d^*+1$-th head is far enough so that the first error occurs after index $i'$ in the read\footnote{This might require extending the length of the heads to larger than $m_{d-d^*}$ when an error occurs near index $m_{d-d^*}$ in the $d^*$-th read in $\boldsymbol{E}^{d-d^*}$.} $\boldsymbol{E}^{d-d^*}_{d^*+1,[1,m_{d-d^*}]}$. Therefore, $\boldsymbol{E}^{d-d^*}_{d^*+1,[p_1-b_1+1,i'-1]}=\boldc_{[i_1,i_1+i'-(p_1-b_1)-2]}=\boldsymbol{E}^{d-d^*}_{w,[p_1-b_1+1,i'-1]}$ for $w\in [1,d^*]$ and $\boldsymbol{E}_{d+1,i'}=\boldc_{i_1+i'-(p_1-b_1)-1}$. Applying Proposition \ref{proposition:majoritybit} by considering a two-row matrix where the first row is $\boldsymbol{E}^{d-d^*}_{w^*,[1,m_{d-d^*}]}$ and the second row is $\boldsymbol{E}^{d-d^*}_{d^*+1,[1,m_{d-d^*}]}$, we have that $\boldsymbol{E}_{w^*,i'}=\boldsymbol{E}_{d^*+1,i'}=\boldc_{i_1+i^*-(p_1-b_1)-1}$, contradicting to the definition of $i'$. Hence, we have that $(\boldsymbol{\delta}_w\cup \boldsymbol{\gamma}_w)\cap [i'-T-2k'-1,i'+k'-1]\ne \emptyset$ for $w\in [1,d^*]$.

According to Proposition \ref{proposition:numberofheadswitherrors}, we have that $k_{d-d^*}\ge 2d^*-1$. Furthermore, by Lemma \ref{lemma:determineshiftedit}, we have that $k_{d-d^*}$ is an even number and thus $k_{d-d^*}\ge 2d^*$. Therefore, we have that $k'\ge k_{d-d^*}+d-d^*\ge d+d^*$, when $\bolde_{[p_{1}-b_{1}+1,m-(b_2-p_2)]}\ne\boldc_{\cI}$.
\end{proof}
\end{proof}

\subsection{Encoding/Decoding Algorithms}\label{subsection:encoding/decoding}
We are now ready to present the encoding and decoding algorithms. We first deal with cases when $d\ge 2d$. 
Given any input sequence $\boldc\in\{0,1\}^n$, the encoding is similar to the one in Section \ref{section:greaterthan2d}, stated in  \eqref{equation:encoding2d} and \eqref{equation:encoding2d1} and is given by
\begin{align}
    Enc_2(\boldc)=(F(\boldc),R^{'}_2(\boldc),R^{''}_2(\boldc)),
\end{align}
where~
\begin{align}
R^{'}_2(\boldc)&=RS_{2\lfloor k/d \rfloor}(S(F(\boldc))),\nonumber\\
R^{''}_2(\boldc)&=Rep_{k+1}(Hash(R^{'}_2(\boldc))).
\end{align}
The difference here is in the definition of the 
 function $S(F(\boldc))$, 
 instead of splitting $F(\boldc)$ into blocks of length $B$, as in \eqref{equation:split}, we split $F(\boldc)$ into blocks of length $B'=(2kdt+2t+1)(k+1)+kdt+3k$, which is $k$ plus the upper bound on the length of the output intervals $[b_{1j},b_{2j}]$, $j\in[1,J]$, i.e.,
 \begin{align}\label{equation:blocks}
     F(\boldc)=&(\bolda'_1,\ldots,\bolda'_{\lceil\frac{n+k+1}{B'}\rceil}), \text{ and}\nonumber\\
     S(F(\boldc))=&(Hash(\bolda'_1),\ldots,Hash(\bolda'_{\lceil\frac{n+k+1}{B'}\rceil}))
 \end{align}
It can be verified that the code has asymptotically the same redundancy $2\lfloor k/d \rfloor\log n+o(\log n)$ as in deletion only cases. In the following, we show that the codeword $Enc_2(\boldc)$ can be correctly decoded. 

We first show that any two sequences $\boldc,\boldc'\in\{0,1\}^n$ such that $F(\boldc)$ and $F(\boldc')$ differ in at least $2\lfloor k/d\rfloor+1$ blocks of length $B'$ cannot result in the same read matrix $\boldsymbol{E}\in\cE_k(Enc_2(\boldc))$.
Suppose on the contrary, $F(\boldc)$ and $F(\boldc')$ differ in at least $2\lfloor k/d\rfloor+1$ blocks. 
Note that according to Lemma \ref{lemma:determineshiftsforedits}, the bits not contained in any minimum edit isolated interval and not in any interval $[b_{1j},b_{2j}]$ after errors can be determined. Therefore, these $2\lfloor k/d\rfloor+1$ blocks either intersects a minimum edit isolated interval or contains a bit $\boldc_i$ that falls within interval $[b_{1j},b_{2j}]$ in the first head. 
Since each minimum edit isolated interval or interval $[b_{1j},b_{2j}]$ is contained within at most two blocks, then at least $2\lfloor k/d\rfloor+1$ blocks where $F(\boldc)$ and $F(\boldc')$ differ contain at least $\lfloor k/d\rfloor+1$ minimum edit isolated intervals. 
According to Proposition \ref{proposition:numberofuncorrectableerrors}, for each minimum edit isolated interval $\cI$, the number of errors $|(\boldsymbol{\delta}_1\cup\boldsymbol{\gamma}_1)\cap\cI|$ in interval $\cI$ when $Enc_2(\boldc)_{\cI}$ and $Enc_2(\boldc')_{\cI}$ is the true sequence should be at least $d-d^*$ and $d+d^*$, or $d+d^*$ and $d-d^*$, respectively, or both at least $d$. Hence, the sum of number of errors in a minimum edit isolated interval when $Enc_2(\boldc)_{\cI}$ and $Enc_2(\boldc')_{\cI}$ is the true sequence is at least $2d$, when $Enc_2(\boldc)_{\cI}\ne Enc_2(\boldc')_{\cI}$. By assumption, we have at least $\lfloor k/d\rfloor+1$ minimum edit isolated intervals $\cI$ satisfying $Enc_2(\boldc)_{\cI}\ne Enc_2(\boldc')_{\cI}$. This implies that the total number of errors occur in the first read when $Enc_2(\boldc)$ and $Enc_2(\boldc')$ is the true sequence should be at least $2d(\lfloor k/d\rfloor+1)>2k$, a contradiction.
Therefore, the encoding $Enc_2(\boldc)$ gives a valid code because any two different $Enc_2(\boldc)$ and $Enc_2(\boldc')$ have block distance at least $2\lfloor k/d\rfloor+1$.

Next, we present the decoding algorithm. 
Similar to what we did in the proof of Theorem \ref{theorem:lessthan2d} and Theorem \ref{theorem:2d}, we use the first row in $\boldsymbol{E}$ to decode  $R'_2(\boldc)$. From Lemma \ref{lemma:deletionequivalenttoinsertion}, we conclude that we can first recover $Hash(R'_2(\boldc))$ and then $R'_2(\boldc)$ using the deletion correcting hash in Lemma \ref{lemma:hash}.

Then,
We use the algorithm before Lemma \ref{lemma:properties} to obtain a set of output intervals $[b_{1j},b_{2j}]$ for $j\in[1,J]$. Then we use \eqref{equation:editrecoverbit} to recover the bits $\boldc_i$ that do not fall within output intervals $[b_{1j},b_{2j}]$, $j\in[1,J]$, in the read of the first head, i.e., the first row in $\boldsymbol{E}$. To compute \eqref{equation:editrecoverbit}, we need Lemma \ref{lemma:determineshiftedit} to determine the shifts caused by errors in each edit isolated interval $\cI_j$, $j\in[1,J]$, which is $s_j$ in \eqref{equation:editrecoverbit}.
  
Then, we apply the algorithm in Lemma \ref{lemma:keditkheads} to every output interval $[b_{1j},b_{2j}]$ and $\boldsymbol{E}_{[1,d],[b_{1j},b_{2j}]}$ and obtain an estimate sequence $\bolde^j$ for $j\in[1,J]$. The sequence $\bolde^j$ is an estimate of $\boldc_{i}$ for all $i$ such that the index $i+|[1,i-1]\cap\boldsymbol{\gamma}_1|-|[1,i-1]\cap\boldsymbol{\delta}_1|$, which can be determined using \eqref{equation:shiftfornoneditisolatedintervals}, is in $[b_{1j},b_{2j}]$. 

According to
Lemma  \ref{lemma:determineshiftsforedits}, 
$\boldc_i$ can be correctly recovered if $i$ is not in any minimum edit isolated interval and the index $i+|[1,i-1]\cap\boldsymbol{\gamma}_1|-|[1,i-1]\cap\boldsymbol{\delta}_1|$, where $\boldc_i$ locates in the first row in the read matrix $\boldsymbol{E}$, is not in the output intervals $[b_{1j},b_{2j}]$, $j\in[1,J]$. 
Moreover, 
note that there are at most $2k$ intervals $[b_{1j},b_{2j}]$ such that $\boldsymbol{E}_{[1,d],[b_{1j},b_{2j}]}$ is not correctly decoded by the algorithm in Lemma \ref{lemma:keditkheads}. We enumerate all possibilities of the set of intervals among $\{[b_{1j},b_{2j}]\}^J_{j=1}$ that are not corrected. There are at most $\sum^{2k}_{i=0}\binom{i}{J}\le (2k)J^k\le (2k)k^k$ such choices. For each choice, we assume that the set of the chosen intervals $\{[b_{1j_i},b_{2j_i}]\}^M_{i=1}$, which cover at most $2M$ blocks in $F(\boldc)$, are corrupted by erasures. In addition, we calculate the number of errors needed for a given choice of the set of intervals. Recall that for each interval $[b_{1j},b_{2j}]$, $j\in[1,J]$, we use the iterative algorithm in Lemma \ref{lemma:keditkheads}. Let $d^*_j$ be the number of rows in the read matrix after applying iterative algorithm in Lemma \ref{lemma:keditkheads} on interval $[b_{1j},b_{2j}]$. Then, according to Proposition \ref{proposition:numberofuncorrectableerrors}, if interval $[b_{1j},b_{2j}]$ is selected, let the number of errors needed in $[b_{1j},b_{2j}]$ for the given choice be $r_j=d+d^*_j$ or $r_j=d$, when $d^*_j=1$ and the parity of $d-1$ are different from the parity of number of errors in $\cI_j$. Otherwise, if interval $[b_{1j},b_{2j}]$ is not selected, then $r_j=d-d^*_j$, or $r_j=d$, when $d^*_j=1$ and the parity of $d-1$ are different from the parity of number of errors in $\cI_j$. 
Let $S=\sum^J_{j=1}r_j$ be a lower bound on the total number of errors needed in intervals $\{[b_{1j},b_{2j}]\}^J_{j=1}$ to generate the given read matrix $\boldsymbol{E}$. Then, by Proposition \ref{proposition:numberofuncorrectableerrors}, a minimum edit isolated interval that does not intersect $\{[b_{1j},b_{2j}]\}^J_{j=1}$ after deletions and insertions contain at least $2d$ errors. Therefore, there are at most $\lfloor (k-S)/2d\rfloor$ such intervals, that cause at most $2\lfloor (k-S)/2d\rfloor$ block substitutions in $F(\boldc)$.  
Therefore, we have at most $2M$ block erasures in $S(F(\boldc))$. In addition, we assume at most $2\lfloor (k-S)/2d\rfloor$ block substitution errors in $S(F(\boldc))$. We use the Reed-Solomon decoding algorithm in \cite{truong88} to correct a combination of $2M$ erasure errors and $\lfloor k/d\rfloor-M$ substitution errors in $S(F(\boldc))$ using Reed-Solomon codes of distance $2\lfloor k/d\rfloor+1$. 

When the correct choice of set of intervals $[b_{1j},b_{2j}]$ is given, i.e., the set of intervals $[b_{1j},b_{2j}]$ are exactly the intervals that are not correctly recovered by the algorithm in Lemma \ref{lemma:keditkheads}, $S(F(\boldc))$ can be correctly recovered given $R'_2(\boldc)$, using the decoder in \cite{truong88}. Then, given $S(F(\boldc))$, the algorithm to recover $\boldc$ is the same as the algorithm in Section \ref{section:greaterthan2d}. 

When the incorrect choice of set of intervals $[b_{1j},b_{2j}]$ is given, we have shown at the beginning of this section that to satisfy requirement on the number of errors $r_j$  needed to generate the given read matrix $\boldsymbol{E}$, the number of blocks where $F(\boldc)$, recovered by the correct choice, and $F(\boldc')$, recovered by the incorrect choice differ is at most $2\lfloor k/d\rfloor$. Therefore, either the sequence $S(F(\boldc))$ cannot be uniquely decoded or the decoded $S(F(\boldc))$ does not satisfy the number of errors requirement, when the incorrect interval choice is given. Therefore, when $2d\le k$, the sequence $\boldc$ can be correctly encoded and decoded. The time complexity is dominated by the time needed to compute $S(F(\boldc))$, which is $O(n\log^{2k}n)$.

For cases when $d\le k\le 2d-1 $. The codes are similar to 
those in Section \ref{section:lessthan2mdeletions}, 
which is given by
\begin{align}
    Enc_1(\boldc)=(F(\boldc),R^{'}_1(\boldc),R^{''}_1(\boldc))
\end{align}
where~
\begin{align}
R^{'}_1(\boldc)&=ER(S(F(\boldc))),\nonumber\\
R^{''}_1(\boldc)&=Rep_{k+1}(Hash(R'_1(\boldc))).
\end{align}
Similar to cases when $k\ge 2d$, we split $F(\boldc)$ into blocks of length $B'$, the same as in \eqref{equation:blocks}. The redundancy is $4k\log\log n +o(\log\log n)$. The correctness of the code is similar to cases when $k\ge 2d$.
Note that when $d\le k\le 2d-1$, there is no minimum edit isolated interval that is not within some interval $[b_{1j},b_{2j}]$. In addition, there is at most one interval $[b_{1j},b_{2j}]$ that where at least $d$ errors occur. We enumerate all choices of such interval. Similar to the case when $k\ge 2d$, to satisfy the number of errors requirement, only the correct choice of the interval gives a valid and correct decoded sequence $\boldc$.

\section{Conclusions}\label{section:conclusion}
We constructed~$d$-head~$k$-deletion racetrack memory codes for any~$k\ge d+1$, extending previous works which addressed cases when~$k\le d$. We proved that for small head distances~$t_i=n^{o(1)}$ and for~$k\ge 2d$, the redundancy of our codes is asymptotically at most four times the optimal redundancy. We also generalized the results and proved that the same redundancy results hold for $d$-head codes correcting a combination of at most $k$ deletions and insertions. 
Finding a lower bound on the redundancy for~$d\le k\le 2d-1$ would be interesting, for both deletion correcting codes and codes correcting a combination of deletions and insertions. It is also desirable to tighten the gap between the upper and lower bounds of the redundancy for cases when~$k\ge 2d$. 

\bibliographystyle{IEEEtran}

\end{document}